\documentclass[12pt]{article}

\usepackage[margin=1in]{geometry}
\usepackage[T1]{fontenc}
\usepackage{lmodern}
\usepackage{amsmath,amssymb,amsthm,mathtools}
\usepackage{booktabs}
\usepackage{enumitem}
\usepackage{xcolor}
\usepackage{microtype}
\usepackage{fancyhdr}
\usepackage{natbib}

\usepackage{comment}
\usepackage[colorlinks=true,linkcolor=blue!50!black,citecolor=blue!50!black,
            urlcolor=blue!50!black]{hyperref}
\hypersetup{
  pdftitle={Sharp Identification in Linear Panel Quantile Models},
  pdfauthor={Shakeeb Khan and Elie Tamer}
}

\allowdisplaybreaks
\numberwithin{equation}{section}
\setlist[itemize]{leftmargin=1.6em,itemsep=0.2em,topsep=0.3em}
\setlist[enumerate]{leftmargin=2.0em,itemsep=0.25em,topsep=0.35em}
\newtheorem{theorem}{Theorem}[section]
\newtheorem{proposition}[theorem]{Proposition}
\newtheorem{corollary}[theorem]{Corollary}
\newtheorem{lemma}[theorem]{Lemma}
\theoremstyle{definition}
\newtheorem{definition}[theorem]{Definition}

\newtheorem{example}[theorem]{Example}
\newtheorem{algorithm}[theorem]{Algorithm}
\newtheorem{remark}[theorem]{Remark}

\newcommand{\R}{\mathbb{R}}
\newcommand{\cB}{\mathcal{B}}
\newcommand{\cC}{\mathcal{C}}

\newcommand{\Pobs}{P^{\mathrm{obs}}}
\newcommand{\Pstr}{P^{\mathrm{s}}}

\newcommand{\E}{\mathbb{E}}
\newcommand{\one}{\mathbf{1}}

\newcommand{\essinf}{\operatorname*{ess\,inf}}
\newcommand{\esssup}{\operatorname*{ess\,sup}}
\newcommand{\ThetaI}{\Theta_I}

\begin{document}

\title{Identification in Linear Quantile Panel Models}
\author{Shakeeb Khan \\ Boston College  \and Elie Tamer \\ Harvard University}
\date{August 2026}
\maketitle

\begin{abstract}
This paper studies identification in linear quantile panel models with unrestricted individual heterogeneity when the number of time periods is fixed and small. We impose quantile strict exogeneity, whereby the conditional quantile restriction holds given the individual’s complete regressor history and latent individual effect, but otherwise allow the disturbances to be arbitrarily dependent over time. We show that the common slope coefficient can be partially identified from the requirement that the residuals in all periods be compatible with a common latent individual effect. We characterize the sharp identified set as the set of coefficients for which there exists a latent coupling satisfying the period-specific quantile restrictions. This characterization yields observable crossing inequalities that provide computationally convenient outer bounds and an observable dual representation that is sharp. When the regressors and outcomes have finite support, the sharp identified set can be computed exactly using finite-dimensional linear programs, without discretizing the latent individual effect; for continuously distributed regressors, we develop nested dual-sieve procedures that converge to the sharp set. We further characterize how identification depends on directional variation in the regressor paths, showing that additional periods need not generate identification unless they provide sufficiently rich variation relative to the support of the composite residual. Under a finite-width composite-residual condition, we obtain explicit bounds on the identified set and give conditions under which even a two-period model is point identified. Finally, we consider a quantile-varying factor-loading extension in which the latent individual effect is common across quantiles, and show how cross-quantile restrictions can identify economically meaningful relative loadings.
\end{abstract}



\newpage
\section{Introduction}
\label{sec:introduction}

Linear panel data models are important in empirical work because they allow researchers
to study the relation between outcomes and covariates in the presence of persistent
unobserved heterogeneity.  In the linear conditional mean model, identification with a
fixed number of periods is well understood. It is developed most clearly in
\citet{Chamberlain1984}, where under a fixed-$T$ regime, the model uses restrictions on the longitudinal
information available about the disturbances---either strict exogeneity or predeterminedness restrictions. The corresponding fixed-$T$ question is
much less developed for conditional quantile models.  This paper asks what the quantile
analogue of mean strict exogeneity identifies about a common slope in a linear model when the individual
effect is unrestricted.  

Consider
\[
        Y_{it}=X_{it}'\beta+A_i+U_{it}, \qquad t=1,\ldots,T.
\]
The classical strict-exogeneity restriction is
\[
        E[U_{it}\mid X_{i1},\ldots,X_{iT},A_i]=0,
        \qquad t=1,\ldots,T.
\]
The conditioning set is important: the mean-zero restriction is imposed given the
complete regressor path and the individual effect.  If future regressors can respond
to current disturbances, the corresponding predeterminedness restriction conditions
instead on $(X_{i1},\ldots,X_{it},A_i)$.  The choice between strict exogeneity and
predeterminedness is therefore a substantive assumption about feedback and timing.  It
does not depend on whether $T$ is fixed or large.

The short-panel analysis of \citet{Chamberlain1984} keeps $T$ fixed and lets the number of
individuals increase.  In such a setting, $A_i$ cannot be estimated consistently for a
particular individual.  This does not prevent identification of the common slope in
the linear mean model.  Differencing removes $A_i$, and linearity of conditional
expectations preserves moment restrictions for $\beta$.  Under the appropriate
within-panel rank condition, these restrictions identify the slope from changes in the
observed regressors.  \citet{griliches1979sibling}, for example, considers a sibling design in
which the two observations are twins in the same family, so that $T$ is equal to 2  by
construction\footnote{This conclusion uses the full-history restriction.  With only
the contemporaneous condition $E[U_{it}\mid X_{it},A_i]=0$, differencing need not
produce valid moment restrictions, and with fixed $T$ the slope need not be identified
or even bounded.  Chamberlain (1984) makes clear why the contemporaneous restriction
is not sufficient for the change regression.}.

Much of the more recent linear fixed-effects quantile literature studies a different model.  It
starts from a period-specific restriction, understood to hold for every individual, of
the form
\[
        Q_{\tau}(Y_{it}\mid X_{it},A_i)
        =X_{it}'\beta(\tau)+A_i.
\]
There is no general within transformation that eliminates $A_i$ and preserves this
conditional quantile restriction.  The individual effects are therefore commonly
estimated together with the slope, and the resulting incidental-parameters problem is
addressed by allowing the number of observations for each individual to increase.
This is the approach taken by the influential work of \citet{koenker2004} and developed in, among others, \citet{lamarche2010, katoetal2012,galvao2017quantile} and references therein.

This literature often combines two features: contemporaneous conditioning and a
large-$T$ asymptotic sequence.  These are logically separate choices.  The restriction
\[
        0\in Q_{\tau}(U_{it}\mid X_{it},A_i)
\]
does not generally imply the same restriction conditional on the complete regressor
history.  If $X_{i,t+1}$ responds to $U_{it}$, for example, the conditional
$\tau$-quantile of $U_{it}$ can be zero given $(X_{it},A_i)$ and nonzero after one also
conditions on $X_{i,t+1}$.  Conversely, allowing $T$ to increase does not determine
which information set should enter the maintained quantile restriction.  The cited
large-$T$ analyses also restrict temporal sampling.  \citet{koenker2004} assumes independent
responses conditional on the regressor array, while \citet{katoetal2012} first study an
i.i.d. sequence over time and then allow stationary weak dependence under mixing
conditions.

We study a different combination of assumptions.  We keep $T$ fixed, possibly as
small as two, and impose the quantile analogue of classical strict exogeneity.  Let
\[
        X_i=(X_{i1},\ldots,X_{iT})
\]
denote the complete regressor path.  Our maintained restriction is
\begin{equation}
        0\in Q_{\tau}(U_{it}\mid X_i,A_i),
        \qquad t=1,\ldots,T. 
        \tag{A}
\end{equation}
We refer to (A) as $\tau$-quantile strict exogeneity, or Q-SE.  It requires zero to
remain a conditional $\tau$-quantile of $U_{it}$ after conditioning on the individual
effect and on past, current, and future regressors.  This full-history condition is an
exogeneity assumption; it is not a consequence of working with a short panel. Related full-history conditioning appears in  important work of 
\citet{manski1987semiparametric} for a binary panel model with stationarity, and in \citet{honore1992trimmed} for a censored panel
model. When
feedback from current disturbances to future regressors is appropriate, the natural
alternative is quantile predeterminedness, which conditions the period-$t$ restriction
only on $(X_{i1},\ldots,X_{it},A_i)$.  Section~\ref{sec:crossing} shows that one necessary implication,
the pairwise residual-crossing outer bound, extends to this weaker model.  The
comparison between periods $s$ and $t$ must then be conditioned on the regressor
history through $\min\{s,t\}$.  The paper's sharp coupling and computational results are
established under Q-SE. It is worth noting that \citet{rosen2012set} studies also a different fixed-$T$ panel restriction that conditions on
$X_i$ but not on $A_i$.  Under such a quantile restriction alone, he shows that every slope in the parameter space is
observationally equivalent to the true slope; Rosen then obtains nontrivial bounds after adding a 
conditional-independence restriction.  See also \citet{canay2011} for other work on panel quantile regressions.

Aside from our Q-SE restriction, the model we use is weak.  We impose no
independence, stationarity, exchangeability, identical-distribution, or mixing
restriction on $(U_{i1},\ldots,U_{iT})$.  The disturbances may be arbitrarily dependent
and heteroskedastic over time, and their distributions may differ across periods.
Moreover, $A_i$ may be arbitrarily related to the complete regressor path.  
The main identification difficulty is immediate.  From
\[
        0\in Q_{\tau}(U_{it}\mid X_i,A_i)
        \quad\text{and}\quad
        0\in Q_{\tau}(U_{is}\mid X_i,A_i)
\]
one cannot conclude that
\[
        0\in Q_{\tau}(U_{it}-U_{is}\mid X_i,A_i).
\]
Differencing eliminates $A_i$ algebraically, but it does not preserve the quantile
restriction.  At the same time, fixed $T$ does not permit consistent estimation of
$A_i$.  The question of this paper is therefore: what does Q-SE identify about the
common slope when $T$ is fixed and the individual effect is unrestricted?  To our
knowledge, the sharp identified set under this combination of assumptions has not been
characterized.

The key observation that motivates our sharp set construction is simple.  For a candidate slope $b$, define
\[
        R_{it}(b)=Y_{it}-X_{it}'b.
\]
At the true slope,
\[
        R_{it}(\beta)=A_i+U_{it}.
\]
The same scalar $A_i$ must therefore satisfy all $T$ conditional-quantile restrictions
for the residual vector, conditional on $X_i$.  This simultaneous compatibility
condition can restrict $b$ even though $A_i$ is neither estimated nor differenced out.
We characterize exactly the slopes for which there exists a joint law of
$(Y_i,X_i,A_i)$, with the observed $(Y_i,X_i)$ marginal, that satisfies all of these
restrictions.

It helps to see in the simplest case why this compatibility requirement has
bite.  Take $T=2$ and $\tau=1/2$.  At a candidate $b=\beta+d$ the residuals
are $R_{it}(b)=A_i+U_{it}-X_{it}'d$: the true composite residual, shifted by an
amount that differs across the two periods whenever $X_{i2}'d\neq X_{i1}'d$.
Model (A) requires some sorting of individuals into groups, indexed by the
value of the individual effect, such that within each group the same number
is a median of the period-1 residuals and a median of the period-2 residuals.
A relative shift between the two periods can be absorbed by such a sorting
only up to the spread of the residuals within a period; a larger shift forces
one period's residuals to lie systematically above the other's in every
group, which no common median can accommodate.  Thus the identified set is
governed by the size of the within-individual index change $(X_{i2}-X_{i1})'d$
relative to the width of the composite residual distribution.  This
observation underlies the simple outer bounds, the exact characterization,
and the results on covariate support that follow.

The paper has five main results.  First, we give an exact coupling characterization of
the sharp identified set.  A candidate $b$ belongs to $\Theta_I$ if and only if the
observed distribution of $(Y_i,X_i)$ can be coupled with a scalar $A_i$ such that, for
every period,
\[
\begin{split}
        \Pr\{R_{it}(b)\leq A_i\mid X_i,A_i\}&\geq\tau,\\
        \Pr\{R_{it}(b)<A_i\mid X_i,A_i\}&\leq\tau.
\end{split}
\]
The weak and strict inequalities allow for atoms at the conditional quantile.  This
characterization is sharp: every structural model satisfying Q-SE produces a feasible
coupling, and every feasible coupling can be completed into a structural model that
satisfies Q-SE and reproduces the observed distribution.

Second, we separate simple observable implications from the full identifying content
of Q-SE.  Common-support overlap and residual-crossing inequalities give convenient
outer sets for $\beta$.  They can exclude many candidate slopes, but they are not sharp
in general.  Under compact support, we derive a sharp observable dual.  A candidate is
in the sharp set if and only if an infinite collection of observable inequalities,
indexed by nonnegative multiplier functions, holds.  Every candidate outside the set
strictly violates at least one such inequality.  The crossing inequalities are useful
special cases of this dual family; the full multiplier class is generally needed for
sharpness.

Third, we make the characterization constructive.  When the regressor paths and the
conditional outcome distributions have finite support, it is enough to place the
latent effect on the finite set of candidate residual values, with no numerical grid
for $A_i$.  The primal linear program constructs a rationalizing distribution, while
the normalized dual defines a scalar criterion that is zero exactly on the sharp
set and negative elsewhere.  Under compatibility of the maintained model, the sharp
set is therefore the set of global maximizers of that criterion, or the set of global
minimizers of its negative.  This permits joint optimization over coefficients and
profiling over nuisance coordinates, rather than requiring a separate feasibility
test at every point of a coefficient grid.  In low-dimensional coefficient spaces,
the entire optimizing set can be recovered by evaluating the criterion once on each
face of the finite residual-crossing hyperplane arrangement.  For continuous designs,
restricting the dual multiplier class gives criterion functions whose optimizing sets
are population outer sets.  Under the compact-support and approximation conditions
of the duality result, these sets decrease to the sharp identified set.

Fourth, we study what determines the size of the identified set.  When panels of
different lengths are compared as marginals of the same underlying population, adding
periods weakly contracts the set.  But neither additional periods nor the usual
within-panel rank condition guarantees point identification.  The rank condition only
rules out coefficient directions whose index can be absorbed exactly into $A_i$.  The
stronger restriction is directional.  Suppose the true composite residual
$V_{it}=A_i+U_{it}$ lies, conditional on $X_i$, in a common, possibly path-dependent
interval of width at most $C_T$.  Then any feasible perturbation
$d=b-\beta$ must satisfy
\[
        \max_{t\leq T} X_{it}'d-
        \min_{t\leq T} X_{it}'d\leq C_T
        \qquad\text{almost surely}.
\]
Thus identification depends on the directional support of the regressor paths relative
to the support of the composite residual, not simply on the number of regressors or on
a covariance rank condition.  This result gives explicit population bounds on the
identified set as $T$ changes; it is not a large-$T$ estimation argument.  It also
permits point identification with only two periods.  In
particular, if $X_{i2}-X_{i1}$ is Gaussian with a positive-definite covariance matrix
and the composite residual has a uniformly finite-width conditional support envelope,
every nonzero coefficient direction generates unbounded within-panel index variation
and the slope is point identified. The paper also shows that in cases with restricted support for the regressors, the identified set is nontrivial and can be wide (with $T$ fixed).

Finally, the baseline model is pointwise in $\tau$: different quantiles may be
rationalized by different latent structures.  We therefore consider the stronger
factor-loading specification
\[
        Q_q(Y_{it}\mid X_i,A_i)
        =X_{it}'\beta(q)+\rho(q)A_i,
\]
under which the same $A_i$ must satisfy the restrictions at several quantiles.  After
normalizing $\rho(1/2)=1$ and imposing the common-orientation restriction
$\rho(\tau)>0$, we combine the median restriction with one additional quantile.  We
characterize the sharp joint set
for $(\beta(1/2),\beta(\tau),\rho(\tau))$ and its sharp projection onto the
target-quantile slope and relative loading.  The cross-quantile restrictions contain
information about $\rho(\tau)$ that is unavailable when either quantile is considered
separately, and the coupling, duality, and finite-support linear-programming arguments
extend to this joint model.

The paper is limited to identification and population computation.  We do not estimate
the individual effects, use a large-$T$ approximation, or develop confidence regions
for the identified set.  The finite-support criterion gives an exact population
optimization characterization.  The continuous-design criteria give population
outer approximations whose limiting intersection is sharp.  Sampling theory and inference are left for future
work.

The remainder of the paper is organized as follows.  Section~\ref{sec:model}
states the model and defines the sharp identified set.  Section~\ref{sec:crossing}
derives the simple observable outer bounds: the crossing inequalities and the
support-overlap condition.  Section~\ref{sec:coupling} gives the exact coupling
characterization and introduces a two-point example that is used throughout
to illustrate the constructions.  Section~\ref{sec:duality} develops the sharp
observable dual, explains why it loses no information, and shows that the
crossing inequalities are a special, and in general insufficient, subfamily.
Section~\ref{sec:lp} turns the dual into computations: an exact linear program
and criterion function under finite support, and nested approximations for
continuous outcomes and regressors.  Section~\ref{sec:longpanel} studies how
covariate support and the number of periods determine the size of the
identified set, including the point-identification and nonidentification
examples.  Section~\ref{sec:loadingduality} develops the factor-loading
extension, Section~\ref{sec:montecarlo} describes the Monte Carlo design, and
Section~\ref{sec:conclusion} concludes.  Proofs of the duality and computational
results are collected in the appendices.

\section{Model and the sharp identified set}
\label{sec:model}

\subsection{Observed data and the maintained restriction}

For one individual, the observed variables are
\[
  Y=(Y_1,\ldots,Y_T)'\in\R^T,
  \qquad
  X=(X_1,\ldots,X_T),
  \qquad X_t\in\R^k,
\]
where \(T\) is fixed and finite.  Let \(\Pobs\) denote the population law of
\((Y,X)\), and let the slope parameter \(b\) range over
\(\cB\subseteq\R^k\).  All random variables take values in Euclidean, hence standard
Borel, spaces.

Fix a quantile index \(\tau\in(0,1)\).  The model of interest is
\begin{equation}
  Y_t=X_t'b+A+U_t,
  \qquad t=1,\ldots,T,
  \label{eq:structural}
\end{equation}
where \(A\in\R\) is an unobserved scalar individual effect.  The maintained
stochastic restriction is
\begin{equation*}
  P(U_t\leq 0\mid X,A)\geq\tau,
  \qquad
  P(U_t<0\mid X,A)\leq\tau,
  \qquad t=1,\ldots,T.
  \tag{A}\label{eq:modelA}
\end{equation*}
Condition \eqref{eq:modelA} says that zero belongs to the possibly set-valued
conditional \(\tau\)-quantile of \(U_t\).  It handles atoms at zero without choosing
a particular generalized inverse.  When \(\tau=1/2\), it is equivalently
\[
  P(U_t\leq0\mid X,A)\geq\tfrac12,
  \qquad
  P(U_t\geq0\mid X,A)\geq\tfrac12.
\]

The maintained model imposes none of the following:
\begin{itemize}
\item independence, exchangeability, stationarity, or identical distributions across
      \(U_1,\ldots,U_T\);
\item independence between \(A\) and \(X\), or between \(A\) and the idiosyncratic
      errors beyond the conditional quantile restriction;
\item continuity of any observed or latent variable;
\item a model linking coefficient functions or errors across different quantile
      indices.
\end{itemize}
Thus the joint conditional law of \(U=(U_1,\ldots,U_T)'\) given \((X,A)\) and
the conditional law of \(A\) given \(X\) are otherwise unrestricted.

\begin{remark}[Conditioning information and quantile strict exogeneity]
\label{rem:conditioninginformation}
Let
\[
  \mathcal G_t=\sigma(X_t,A),
  \qquad
  \mathcal H_t=\sigma(X_1,\ldots,X_t,A),
  \qquad
  \mathcal F=\sigma(X,A).
\]
Because \(\mathcal G_t\subseteq\mathcal H_t\subseteq\mathcal F\), Condition
\eqref{eq:modelA} implies, by iterated expectations, the corresponding two
quantile inequalities conditional on \(\mathcal H_t\) and on \(\mathcal G_t\).
The converses generally fail.  Full-path conditioning is therefore a genuine
strict-exogeneity restriction at the \(\tau\)-th quantile, not a harmless change
in notation from \(0\in Q_\tau(U_t\mid X_t,A)\).

This restriction is nevertheless much weaker than independence of \(U_t\) from
\(X\) given \(A\).  Only the conditional probabilities of the two events at zero
are restricted.  The magnitudes and shape of \(U_t\mid(X,A)\) may depend
arbitrarily on the regressor path, the disturbances may be arbitrarily dependent
across periods, and \(A\) may be arbitrarily related to \(X\).  We maintain the
full path because it makes all period-specific restrictions valid on one common
conditioning field.  This delivers the conditional-on-path crossing inequalities,
the pathwise coupling characterization, and the decomposition into pathwise LPs.
In applications where current disturbances affect future covariates at the
relevant quantile, a contemporaneous or sequential restriction may be more
credible, but the restriction is weaker, its identified set is generally larger,
and its sharp computation is different.
\end{remark}

\paragraph{Equivalent pointwise formulation.}
Writing $b_\tau=b(\tau)$ and making the quantile index on the latent effect
explicit, the maintained model is equivalently
\[
X_t'b_\tau+A_\tau\in Q_\tau(Y_t\mid X,A_\tau),
\qquad t=1,\ldots,T,
\]
where membership in the possibly set-valued conditional quantile means
\[
P(Y_t-X_t'b_\tau-A_\tau\leq 0\mid X,A_\tau)\geq\tau,
\qquad
P(Y_t-X_t'b_\tau-A_\tau<0\mid X,A_\tau)\leq\tau.
\tag{A}
\]
Equivalently, defining the quantile-specific residual
\[
U_{t,\tau}=Y_t-X_t'b_\tau-A_\tau,
\]
the model can be written as
\[
Y_t=X_t'b_\tau+A_\tau+U_{t,\tau},
\qquad
0\in Q_\tau(U_{t,\tau}\mid X,A_\tau).
\]
Because the baseline analysis is pointwise at the fixed index $\tau$,
we suppress the subscript and write $b=b_\tau$, $A=A_\tau$, and
$U_t=U_{t,\tau}$ throughout.

\begin{remark}[Pointwise versus common-across-quantile heterogeneity]
\label{rem:pointwiseanchor}
The baseline model is maintained at one fixed quantile index.  For each
\(q\) analyzed separately, it permits a different latent completion and a
different scalar anchor \(A_q\) satisfying
\[
  X_t'b(q)+A_q\in Q_q(Y_t\mid X,A_q),
  \qquad t=1,\ldots,T.
\]
Thus \(A_q\) is common across periods at the specified quantile, but neither
\(A_q=A_{q'}\) nor equality of the latent couplings is imposed when
\(q\neq q'\).  Separate pointwise identified sets consequently impose no
cross-quantile compatibility or noncrossing restriction.

Section \ref{sec:loadingduality} changes this maintained structure.  There the
quantile-specific individual intercept is restricted to the one-factor form
\(A_q=\rho(q)A\), and the median and target-quantile restrictions must hold
under one common coupling of \((Y,X,A)\).  The distinction is between separate
pointwise existence of possibly different latent structures and joint existence
of one latent structure satisfying several quantile restrictions.
\end{remark}

\subsection{The structural sharp set}
\label{sec:structuralset}

It is useful to define the identified set first in terms of complete structural
probability laws, before introducing the coupling representation.

\begin{definition}[Admissible structural law]
\label{def:structure}
For a candidate \(b\in\cB\), an admissible structural law is a Borel probability
law \(\Pstr\) of \((Y,X,A,U)\) such that:
\begin{enumerate}[label=(\roman*)]
\item the \((Y,X)\)-marginal of \(\Pstr\) is \(\Pobs\);
\item \(Y_t=X_t'b+A+U_t\) holds \(\Pstr\)-almost surely for every \(t\);
\item the two conditional inequalities in \eqref{eq:modelA} hold
      \(\Pstr\)-almost surely for every \(t\).
\end{enumerate}
Let \(\mathfrak S(b)\) be the collection of all such laws.
\end{definition}

\begin{definition}[Structural sharp identified set]
\label{def:sharpstruct}
The sharp identified set for the slope under \eqref{eq:structural}--\eqref{eq:modelA}
is
\begin{equation}
  \ThetaI^{\mathrm{str}}
  =\{b\in\cB:\mathfrak S(b)\neq\varnothing\}.
  \label{eq:structsharp}
\end{equation}
\end{definition}

Definition \ref{def:sharpstruct} is the usual observational-equivalence definition:
a slope is included precisely when at least one complete latent-variable structure
reproduces the observed distribution and satisfies the maintained restrictions.

\subsection{Time-invariant directions, rank, and nonidentification}
\label{subsec:rankA}

The unrestricted individual effect absorbs every component of the linear index
that is constant over time.  Define the finite-horizon absorption space
\begin{equation}
  \mathcal D_T
  =
  \left\{
    d\in\mathbb R^k:
    (X_t-X_s)'d=0\ \text{almost surely for every }s,t\leq T
  \right\}.
  \label{eq:absorption-space}
\end{equation}
An intercept belongs to \(\mathcal D_T\), as does the coefficient direction of
any time-invariant regressor.  More generally, \(\mathcal D_T\) can contain a
linear combination of regressors even when no individual regressor is constant.

\begin{proposition}[Absorption equivalence and the rank requirement]
\label{prop:absorptionA}
Let \(d\in\mathcal D_T\).  If \(b,b+d\in\mathcal B\), then
\[
  b\in\Theta_I^{\mathrm{str}}
  \quad\Longleftrightarrow\quad
  b+d\in\Theta_I^{\mathrm{str}}.
\]
Consequently, the data can identify \(b\) at most modulo
\(\mathcal D_T\).  The condition
\begin{equation}
  \mathcal D_T=\{0\}
  \label{eq:no-absorption-A}
\end{equation}
is necessary to rule out this exact nonidentification, but it is not sufficient
for point identification under model \textup{(A)}.
\end{proposition}

\begin{proof}
For \(d\in\mathcal D_T\), there is a measurable
\(c_d(X)=X_t'd\) that does not depend on \(t\), outside one common null set.
Given an admissible structure for \(b\), set
\[
  A^{d}=A-c_d(X),\qquad U_t^{d}=U_t.
\]
Then
\[
  X_t'(b+d)+A^{d}+U_t^{d}
  =X_t'b+A+U_t=Y_t.
\]
Conditional on \(X\), the map \(A\mapsto A^d\) is one-to-one, so
\(\sigma(X,A^d)=\sigma(X,A)\).  The two conditional quantile inequalities are
therefore unchanged.  This proves feasibility of \(b+d\); applying the same
argument to \(-d\) proves the converse.
\end{proof}

If the within-path differences are square integrable, define
\begin{equation}
  G_T
  =
  \sum_{1\leq s<t\leq T}
  \mathbb E\!\left[
    (X_t-X_s)(X_t-X_s)'
  \right].
  \label{eq:within-rank-matrix}
\end{equation}
Then \(d'G_Td=0\) if and only if \(d\in\mathcal D_T\); hence
\eqref{eq:no-absorption-A} is equivalent to
\(\operatorname{rank}(G_T)=k\).  For a nonrandom regressor path this reduces to
full column rank of a matrix of within-path differences, for example the rows
\((x_t-x_1)'\), \(t=2,\ldots,T\).  This rank condition removes only the
algebraic absorption invariance.  It does not by itself guarantee that the
common conditional quantile restrictions separate every pair of slopes.

We henceforth write
\[
  \Theta_I:=\Theta_I^{\mathrm{str}}.
\]
All sharp-set statements remain valid without
\eqref{eq:no-absorption-A}; when the condition fails, they characterize the
corresponding union of observationally equivalent directions.  To report a
coefficient on a time-invariant regressor one must impose an additional
normalization or restriction on \(A\); the maintained model alone supplies none.

\section{Simple Outer Set via Moment Inequalities}
\label{sec:crossing}

The sharp set is a latent-variable object, but model (A) has a particularly simple
observable implication.  It supplies an inexpensive outer bound before any coupling
problem is solved.

\subsection{Crossing inequalities}

For a candidate \(b\), continue to write
\[
  R_t(b)=Y_t-X_t'b.
\]
Let
\begin{equation}
  c_\tau=\max\{\tau,1-\tau\}.
  \label{eq:ctau}
\end{equation}

\begin{proposition}[Observable crossing bound]
\label{prop:crossing}
If \(b\in\ThetaI^{\mathrm{str}}\), then for every \(v\in\R\), every ordered pair
\(s\neq t\), and \(P_X\)-almost every \(x\),
\begin{equation}
  P^{\mathrm{obs}}\{R_s(b)\leq v<R_t(b)\mid X=x\}
  \leq c_\tau.
  \label{eq:crossing}
\end{equation}
In particular, at the median no fixed level can be crossed strictly in either
direction by more than one half of the conditional population.
\end{proposition}

In words: among individuals with a given regressor path, those whose
period-\(s\) residual falls at or below a level \(v\) while their period-\(t\)
residual exceeds it cannot make up more than a fraction \(c_\tau\) of the
population.  The reason is that such an individual's common effect must be
either at or below \(v\), in which case the period-\(t\) disturbance is
positive, or above \(v\), in which case the period-\(s\) disturbance is
negative; each of these tail events has conditional probability at most
\(1-\tau\) or \(\tau\), respectively.

\begin{proof}
Fix \(b\in\ThetaI^{\mathrm{str}}\) and write
\(E_{st}(v;b)=\{R_s(b)\leq v<R_t(b)\}\).  Under an admissible structural
law, \(R_j(b)=A+U_j\).  Conditional on \((X,A)\), split according to the
position of \(A\).  If \(A\leq v\), then
\[
  E_{st}(v;b)\subseteq\{R_t(b)>v\}\subseteq\{U_t>0\},
\]
so its conditional probability is at most \(1-\tau\).  If \(A>v\), then
\[
  E_{st}(v;b)\subseteq\{R_s(b)\leq v\}\subseteq\{U_s<0\},
\]
so its conditional probability is at most \(\tau\).  Hence, almost surely,
\begin{equation}
  P\{E_{st}(v;b)\mid X,A\}
  \leq(1-\tau)\one\{A\leq v\}+\tau\one\{A>v\}.
  \label{eq:conditional-crossing-bound}
\end{equation}
Taking expectations conditional on \(X=x\) bounds the crossing probability by
\[
  (1-\tau)P(A\leq v\mid X=x)+\tau P(A>v\mid X=x)
  \leq\max\{\tau,1-\tau\}.
\]
At the median the last expression equals \(1/2\).  The argument uses only one
necessary marginal tail event in each case and never multiplies probabilities across
periods.
\end{proof}

\begin{remark}[Why the bound is conditional on the full path]
\label{rem:crossingconditioning}
The conditional-on-\(X=x\) conclusion in \eqref{eq:crossing} uses the fact that
both period-specific quantile restrictions hold given the same full path.  Under
only contemporaneous restrictions given \((X_s,A)\) and \((X_t,A)\), that
pathwise conclusion generally does not follow.  The same proof can first average
each tail restriction down to \(A\) and then integrate over \(A\), yielding only
the unconditional implication
\[
  P^{\mathrm{obs}}\{R_s(b)\leq v<R_t(b)\}\leq c_\tau.
\]
Accordingly, contemporaneous conditioning would require a different, global
coupling problem rather than the path-by-path construction used in this paper.
\end{remark}

\begin{corollary}[Observable conditional moment inequalities]
\label{cor:crossingmoments}
Every \(b\in\ThetaI\) satisfies
\begin{equation}
  \E_{\mathrm{obs}}\!\left[
    h(X)\{\one(R_s(b)\leq v<R_t(b))-c_\tau\}
  \right]\leq0
  \label{eq:observablecmi}
\end{equation}
for every ordered pair \(s\neq t\), every \(v\in\R\), and every bounded,
nonnegative, Borel-measurable instrument \(h(X)\).
\end{corollary}

\begin{proof}
Multiply \eqref{eq:crossing} by \(h(X)\geq0\) and apply iterated expectations.
\end{proof}

Corollary \ref{cor:crossingmoments} is useful with continuously distributed outcomes
and regressors: it involves only observed variables and the candidate \(b\).
Indicators of covariate cells give finite-partition outer restrictions, while richer
instrument classes tighten them.  Crucially, the event in
\eqref{eq:observablecmi} uses each observation's actual candidate residuals.  Replacing
within-cell regressors by a representative value is a separate approximation.

\begin{remark}[{\bf Crossing inequalities under predetermined regressors}]
\label{rem:crossing-predetermined}

The crossing inequalities extend directly to a predeterminedness restriction,
with one important change: the relevant conditioning set is the history
available at the earlier of the two periods, rather than the complete regressor
path.

To see this first for $T=2$, suppose that, in place of strict exogeneity, we
maintain
\[
    0\in Q_\tau(U_1\mid X_1,A)
    \qquad\text{and}\qquad
    0\in Q_\tau(U_2\mid X_1,X_2,A).
    \tag{P}
\]
Here and below, $0\in Q_\tau(U\mid\mathcal G)$ means
\[
    P(U\leq 0\mid\mathcal G)\geq\tau,
    \qquad
    P(U<0\mid\mathcal G)\leq\tau .
\]
The second restriction in (P) also holds after conditioning on the smaller
information set $(X_1,A)$.  Indeed, by iterated expectations,
\[
\begin{split}
    P(U_2\leq0\mid X_1,A)
    &=
    E\!\left[
        P(U_2\leq0\mid X_1,X_2,A)
        \mid X_1,A
    \right]
    \geq \tau,\\
    P(U_2<0\mid X_1,A)
    &=
    E\!\left[
        P(U_2<0\mid X_1,X_2,A)
        \mid X_1,A
    \right]
    \leq \tau.
\end{split}
\]
Hence
\[
    0\in Q_\tau(U_t\mid X_1,A),
    \qquad t=1,2.
\]
Thus the two disturbances share the same conditional quantile anchor on the
common information set $(X_1,A)$.

Let
\[
    R_t(b)=Y_t-X_t'b,
    \qquad
    c_\tau=\max\{\tau,1-\tau\}.
\]
At the true parameter $b=\beta$,
\[
    R_t(\beta)=A+U_t.
\]
The argument used for the strict-exogeneity crossing bound therefore gives,
for every $v\in\mathbb R$,
\[
    P\!\left(
        R_1(\beta)\leq v<R_2(\beta)
        \mid X_1
    \right)
    \leq c_\tau,
    \tag{P-C$_{12}$}
\]
and
\[
    P\!\left(
        R_2(\beta)\leq v<R_1(\beta)
        \mid X_1
    \right)
    \leq c_\tau.
    \tag{P-C$_{21}$}
\]
The form of the crossing inequality is therefore unchanged.  What changes is
the conditioning set.  Under strict exogeneity the corresponding inequalities
hold conditional on $(X_1,X_2)$; under predeterminedness they are guaranteed
only conditional on $X_1$.

This distinction is substantive.  In general,
\[
    0\in Q_\tau(U_1\mid X_1,A)
\]
does not imply
\[
    0\in Q_\tau(U_1\mid X_1,X_2,A).
\]
The future regressor $X_2$ may depend on $U_1$, so conditioning on $X_2$ may
alter the conditional quantile of the first-period disturbance.  Consequently,
under predeterminedness one cannot in general impose the crossing inequalities
conditional on the complete regressor path.

The same argument applies for arbitrary $T$.  Write
\[
    \bar X_t=(X_1,\ldots,X_t),
\]
and suppose that
\[
    0\in Q_\tau(U_t\mid \bar X_t,A),
    \qquad t=1,\ldots,T.
    \tag{Pred}
\]
For any $s<t$, the period-$t$ restriction can be coarsened from
$(\bar X_t,A)$ to $(\bar X_s,A)$.  Hence
\[
    0\in Q_\tau(U_s\mid \bar X_s,A),
    \qquad
    0\in Q_\tau(U_t\mid \bar X_s,A).
\]
It follows that, for every pair $s<t$ and every $v\in\mathbb R$,
\[
    P\!\left(
        R_s(\beta)\leq v<R_t(\beta)
        \mid \bar X_s
    \right)
    \leq c_\tau,
    \tag{P-C$_{st}^{+}$}
\]
and
\[
    P\!\left(
        R_t(\beta)\leq v<R_s(\beta)
        \mid \bar X_s
    \right)
    \leq c_\tau.
    \tag{P-C$_{st}^{-}$}
\]
Equivalently, for any two distinct periods $s$ and $t$, the crossing
restriction may be conditioned on
\[
    \bar X_{\min\{s,t\}},
\]
the largest regressor history with respect to which both period-specific
quantile restrictions are necessarily valid.

Thus a natural crossing-based outer set under predeterminedness is
\[
\mathcal B_{\mathrm{pred}}^{\,\mathrm{cross}}
=
\left\{
b:
\begin{array}{l}
P\!\left(
R_s(b)\leq v<R_t(b)\mid \bar X_s
\right)\leq c_\tau,\\[2mm]
P\!\left(
R_t(b)\leq v<R_s(b)\mid \bar X_s
\right)\leq c_\tau,
\end{array}
\quad
\forall\, s<t,\ \forall\,v\in\mathbb R
\right\},
\]
where the inequalities are understood to hold almost surely in $\bar X_s$.

Relative to strict exogeneity, predeterminedness therefore does not change the
basic crossing argument; it changes the information on which the argument can
be conditioned.  Strict exogeneity permits pairwise crossings to be examined
within cells of the complete path $X=(X_1,\ldots,X_T)$, whereas
predeterminedness requires the $(s,t)$ comparison to average over regressors
dated after period $s$.  The resulting crossing outer set is therefore, in
general, weaker.
\end{remark}

\subsection{Two-period crossing profiles}

Fix \(T=2\) and a regressor path \(X=x=(x_1,x_2)\).  The path restricts \(b\)
through
\[
  d=(x_2-x_1)'b.
\]
Adding the common number \(x_1'b\) to both residual coordinates does not change
crossing events after the threshold is relabeled.  Define
\begin{align}
  G_x(d)
  &=\sup_{v\in\R}
    P^{\mathrm{obs}}\{Y_1\leq v<Y_2-d\mid X=x\},
  \label{eq:Gprofile}\\
  H_x(d)
  &=\sup_{v\in\R}
    P^{\mathrm{obs}}\{Y_2-d\leq v<Y_1\mid X=x\}.
  \label{eq:Hprofile}
\end{align}
The function \(G_x\) is nonincreasing in \(d\), and \(H_x\) is nondecreasing.
Proposition \ref{prop:crossing} therefore gives the pathwise outer set
\begin{equation}
  D_x^{\mathrm{cross}}
  =\{d\in\R:G_x(d)\leq c_\tau,\ H_x(d)\leq c_\tau\},
  \label{eq:crossinterval}
\end{equation}
which is an interval, possibly empty, unbounded, or with one or both endpoints
excluded.  Atoms determine endpoint inclusion, so it is safest to retain the direct
inequalities in \eqref{eq:crossinterval}.  With finite conditional support, the
profiles change only at finitely many differences \(y_2^j-y_1^h\), and every boundary
point can be checked explicitly.

With finitely many regressor paths, intersecting
\[
  \{b:(x_2-x_1)'b\in D_x^{\mathrm{cross}}\}
\]
over \(x\) produces a convex, easily computed outer set.  For \(T>2\), intersecting
the analogous restriction over every period pair remains outer-valid.

\paragraph{Why first differences alone are insufficient.}
The crossing restrictions use the joint distribution of residual levels, not only
their differences.  This is essential.  Model (A) places no restriction on the law
of \(U_2-U_1\): for any target difference \(W\), one may mix
\((-W,0)\) with \((0,W)\) when \(W>0\), and use the mirrored construction when
\(W<0\), with weights chosen so that both coordinates retain conditional
\(\tau\)-quantile zero.  Consequently, the law of \(\Delta Y\) considered in
isolation cannot exclude a slope under model (A).  Identification comes from the
compatibility of the observed outcome levels with one common latent anchor.

The next section turns the structural definition into an exact restriction on
couplings of the observed data with that scalar anchor.

\subsection{A support-overlap outer bound}
\label{subsec:supportouter}

\begin{lemma}[Common support overlap is necessary]
\label{lem:overlap}
For a candidate \(b\), define
\[
  \ell_t(x,b)=\essinf\{R_t(b)\mid X=x\},
  \qquad
  u_t(x,b)=\esssup\{R_t(b)\mid X=x\}.
\]
If \(b\in\ThetaI\), then for \(P_X\)-almost every \(x\),
\begin{equation}
  \bigcap_{t=1}^T[\ell_t(x,b),u_t(x,b)]\neq\varnothing.
  \label{eq:overlap}
\end{equation}
\end{lemma}

\begin{proof}
The proof uses the coupling representation of Section~\ref{sec:coupling}:
by Theorem~\ref{thm:sharp}, \(b\in\ThetaI\) if and only if there is a joint
law \(\pi\) of \((Y,X,A)\) with the observed \((Y,X)\)-marginal under which
\(A\) is a conditional \(\tau\)-quantile of every \(R_t(b)\) given \((X,A)\);
write \(\cC_b(\Pobs)\) for the set of such laws.
Let \(\pi\in\cC_b(\Pobs)\).  Fix an \(x\) outside the null sets used in the
conditional distributions.  Since the conditional residual marginal is the observed
one, disintegration implies that, for \(P_\pi(A\in da\mid X=x)\)-almost every \(a\),
the conditional law of \(R_t(b)\) given \((X,A)=(x,a)\) is supported on
\([\ell_t(x,b),u_t(x,b)]\).

If \(a<\ell_t(x,b)\), then
\(P_\pi(R_t(b)\leq a\mid X=x,A=a)=0<\tau\).
If \(a>u_t(x,b)\), then
\(P_\pi(R_t(b)<a\mid X=x,A=a)=1>\tau\).
Feasibility therefore requires
\(a\in[\ell_t(x,b),u_t(x,b)]\) for every \(t\), for almost every anchor used by the
coupling.  The conditional law of \(A\) is a probability measure, so the intersection
in \eqref{eq:overlap} must be nonempty.
\end{proof}

\paragraph{Support-overlap outer set.}
Lemma \ref{lem:overlap} implies
\begin{equation}
  \ThetaI\subseteq\Theta_I^{\mathrm{supp}}
  :=
  \left\{b:
  \max_t\ell_t(x,b)\leq\min_tu_t(x,b)
  \text{ for \(P_X\)-almost every }x\right\}.
  \label{eq:supportouterset}
\end{equation}
With finite conditional support, these endpoints are the smallest and
largest candidate residuals in each period.  The restriction can therefore be
checked directly from the observed support, without choosing a distribution for
\(A\).  Section~\ref{sec:lp} explains how to use it as a preliminary check
before computing the sharp set.

\section{Sharp latent-coupling characterization}
\label{sec:coupling}

Definition~\ref{def:sharpstruct} describes the identified set in terms of
complete structural laws for \((Y,X,A,U)\).  This section shows that the
disturbances can be dropped from that description: a slope is in the
identified set if and only if the observed distribution of \((Y,X)\) can be
coupled with a scalar \(A\) so that \(A\) is a conditional \(\tau\)-quantile
of every candidate residual.  The reformulation is elementary, since at any
candidate slope the disturbance is just the residual minus the effect, but it
is the form in which the identification problem becomes tractable: it
replaces a search over structures by a search over one latent scalar per
individual, and Sections~\ref{sec:duality} and~\ref{sec:lp} operate entirely
on it.

For \(b\in\cB\), define the candidate residual coordinate
\begin{equation}
  R_t(b)=Y_t-X_t'b,
  \qquad t=1,\ldots,T.
  \label{eq:residual}
\end{equation}
Let \(\Pi(\Pobs)\) be the collection of all Borel probability laws \(\pi\) of
\((Y,X,A)\) whose \((Y,X)\)-marginal is \(\Pobs\).  The conditional law of \(A\)
may therefore depend arbitrarily on the observed pair \((Y,X)\).

\begin{definition}[Feasible latent coupling]
\label{def:coupling}
For \(b\in\cB\), let \(\cC_b(\Pobs)\) contain every
\(\pi\in\Pi(\Pobs)\) satisfying, for \(t=1,\ldots,T\),
\begin{align}
  \pi\{R_t(b)\leq A\mid X,A\}&\geq\tau,
  \label{eq:coupleweak}\\
  \pi\{R_t(b)<A\mid X,A\}&\leq\tau
  \label{eq:couplestrict}
\end{align}
\(\pi\)-almost surely.
\end{definition}

\begin{theorem}[Exact coupling characterization and sharpness]
\label{thm:sharp}
Under only \eqref{eq:structural}--\eqref{eq:modelA},
\begin{equation}
  \ThetaI^{\mathrm{str}}
  =
  \ThetaI^{\mathrm{coup}}
  :=
  \{b\in\cB:\cC_b(\Pobs)\neq\varnothing\}.
  \label{eq:sharpequality}
\end{equation}
Consequently, the right-hand side of \eqref{eq:sharpequality} is the sharp identified
set: every included \(b\) is rationalized by a complete admissible structure, and no
excluded \(b\) can be rationalized by any structure satisfying model (A).
\end{theorem}

\begin{proof}
\textit{Necessity.}
Let \(b\in\ThetaI^{\mathrm{str}}\).  By
Definition \ref{def:sharpstruct}, choose
\(\Pstr\in\mathfrak S(b)\), and let \(\pi\) be its marginal law of \((Y,X,A)\).
Property (i) of Definition \ref{def:structure} implies
\(\pi\in\Pi(\Pobs)\).  By the structural identity,
\[
  U_t=Y_t-X_t'b-A=R_t(b)-A
  \qquad \Pstr\text{-almost surely}.
\]
Therefore, conditional on \((X,A)\),
\[
  \{U_t\leq0\}=\{R_t(b)\leq A\},
  \qquad
  \{U_t<0\}=\{R_t(b)<A\}
\]
up to null sets.  The two restrictions in \eqref{eq:modelA} become
\eqref{eq:coupleweak}--\eqref{eq:couplestrict}.  Hence
\(\pi\in\cC_b(\Pobs)\), so \(b\in\ThetaI^{\mathrm{coup}}\).

\textit{Sufficiency.}
Let \(b\in\ThetaI^{\mathrm{coup}}\), and choose
\(\pi\in\cC_b(\Pobs)\).  On the probability space carrying
\((Y,X,A)\sim\pi\), define, coordinate by coordinate,
\begin{equation}
  U_t:=Y_t-X_t'b-A=R_t(b)-A.
  \label{eq:constructU}
\end{equation}
Let \(\Pstr\) be the joint law of \((Y,X,A,U)\) generated by this deterministic
transformation.  The \((Y,X)\)-marginal of \(\Pstr\) is \(\Pobs\), and
\eqref{eq:structural} holds identically.  Moreover,
\[
  P^{\mathrm{s}}(U_t\leq0\mid X,A)
  =\pi\{R_t(b)\leq A\mid X,A\}\geq\tau
\]
and
\[
  P^{\mathrm{s}}(U_t<0\mid X,A)
  =\pi\{R_t(b)<A\mid X,A\}\leq\tau.
\]
Thus \(\Pstr\in\mathfrak S(b)\), and
\(b\in\ThetaI^{\mathrm{str}}\).

The two containments prove \eqref{eq:sharpequality}.  The construction
\eqref{eq:constructU} also proves attainability for every included candidate, while
the necessity argument proves that no omitted candidate can satisfy the maintained
model.  This is precisely sharpness.
\end{proof}

\begin{remark}[Kernels are implied, not additionally assumed]
\label{rem:kernels}
Because all spaces are standard Borel, any feasible joint law admits regular
conditional distributions \(P(A\in da\mid X=x)\) and
\(P(U\in du\mid X=x,A=a)\).  Thus Theorem \ref{thm:sharp} also constructs the
usual kernel representation \((F_{A\mid X},F_{U\mid X,A})\).  No independence
factorization of either kernel is used.
\end{remark}

\begin{remark}[Global and pathwise formulations]
\label{rem:globalpathwise}
The global coupling formulation automatically includes the measurable-selection
requirement when \(X\) is continuously distributed.  When \(X\) has finite support,
it separates exactly across positive-probability regressor paths: a candidate is
feasible if and only if a conditional coupling exists in every path.  Section
\ref{sec:lp} uses this pathwise factorization for computation.  Marginalizing a
feasible full-panel coupling to any pair of periods also shows that pairwise sharp
sets are outer for the full-panel set.
\end{remark}

\paragraph{How to read the coupling.}
It is useful to think of a coupling as a sorting rule.  Each individual is
observed with a regressor path and an outcome vector; the coupling assigns to
that individual a value of the effect \(A\), possibly at random.  Conditioning
on \((X,A)=(x,a)\) collects all individuals with regressor path \(x\) who were
assigned the effect \(a\).  Conditions
\eqref{eq:coupleweak}--\eqref{eq:couplestrict} say that, within every such
group, the number \(a\) must be a \(\tau\)-quantile of the period-1 residuals,
and of the period-2 residuals, and so on through period \(T\).  A candidate
slope belongs to the identified set exactly when the observed population can
be sorted into groups with this property.  Two features of the sorting give
the model its content.  The assignment is made once per individual, not once
per period, so the same \(a\) has to work for all \(T\) residual coordinates
of everyone in the group.  And the assignment may depend on the outcomes,
which is why no restriction on the dependence between \(A\) and \((Y,X)\) is
needed or used.

At the true slope, sorting by the true individual effect works, since the
residuals are then \(A+U_t\) and zero is a conditional \(\tau\)-quantile of
each \(U_t\).  At a false slope \(b=\beta+d\), the residuals become
\(A+U_t-X_t'd\): the true composite residual shifted by an amount that
varies over time whenever \(X_t'd\) does.  The question is whether some other
sorting can accommodate the time-varying shift.  When the shift is small
relative to the spread of the residuals it typically can, and when it is
large it cannot.  The following example, which is carried through
Sections~\ref{sec:duality} and~\ref{sec:lp}, makes this concrete.

\begin{example}[Two outcome points]
\label{ex:running}
Let \(T=2\), \(\tau=1/2\), \(\cB=\mathbb R\), and let the regressor be scalar
with a single path \(X=(X_1,X_2)=(0,1)\) of probability one, so that
conditioning on the path is vacuous.  (Statements about groups below hold up
to null sets under the coupling.)  Suppose the observed outcome vector
satisfies
\[
  P\{(Y_1,Y_2)=(0,0)\}=P\{(Y_1,Y_2)=(1,1)\}=\tfrac12 .
\]
These data are generated, for instance, by \(\beta=0\) with \(A\in\{0,1\}\)
equally likely and \(U_1=U_2=0\); or by \(\beta=0\) with \(A=0\) and
perfectly persistent disturbances \(U_1=U_2\in\{0,1\}\) equally likely.  Model
(A) does not distinguish these structures, and both place \(b=0\) in the
identified set.

For a candidate \(b\), the residuals are \(R_1(b)=Y_1\) and
\(R_2(b)=Y_2-b\).  The outcome vector \((0,0)\) has residual pair
\((0,-b)\) and the vector \((1,1)\) has residual pair \((1,1-b)\).  A
group containing only one of the two outcome vectors cannot work unless
\(b=0\): with a single point of mass, \(a\) would have to equal both of its
residuals, which differ.  So for \(b\neq0\) every group must contain both
vectors, and by the same reasoning it must contain them with equal weight:
if one vector carried more than half of a group's mass, the group's median in
each period would be that vector's residual, and \(a\) would again have to
equal two different numbers.  Consider, then, a group with a common effect
\(a\) and equal weight on the two vectors.  Its period-1 residuals are
\(\{0,1\}\), and \(a\) is a median of them if and only if \(0\leq a\leq1\).  The period-2 residuals are \(\{-b,1-b\}\), and
\(a\) is a median of them if and only if \(-b\leq a\leq1-b\).  Both
requirements can be met if and only if \(|b|\leq1\).  Hence
\[
  \ThetaI=[-1,1].
\]
The endpoints belong to the set because of the weak and strict inequalities:
at \(b=1\), the common effect \(a=0\) is a median of \(\{0,1\}\) and of
\(\{-1,0\}\).

The width of the identified set has a simple reading.  Moving from
\(\beta=0\) to \(b\) shifts the period-2 residuals relative to the period-1
residuals by \(-(X_2-X_1)b=-b\).  The pooled group can absorb a relative shift
of at most the spread of the residuals within a period, which is one here.
The identified set is therefore the set of \(b\) for which
\(|X_2-X_1|\,|b|\) does not exceed the residual spread.
Section~\ref{sec:longpanel} shows that this reading is general as an outer
bound: the set of slopes for which the within-individual change in the index
does not exceed the width of the composite residual always contains
\(\ThetaI\).  It is exact here because, with two equally weighted residual
values, every point between them is a median, so the entire spread is
available to absorb the shift.  In other designs the outer bound can be
strict: Section~\ref{sec:longpanel} gives examples in which it is exact, and
Example~\ref{ex:crossingnotsharp} gives a configuration in which even the
finer crossing bounds of Section~\ref{sec:crossing} are not.
\hfill\(\square\)
\end{example}

\section{Sharp observable duality}
\label{sec:duality}

Theorem~\ref{thm:sharp} describes the identified set through an object that
is never observed: a joint distribution of the data and the individual
effect.  Econometricians are more accustomed to identified sets described by
observable restrictions, typically a family of moment inequalities in the
parameter and the distribution of the data.  This section shows that model
(A) can be written in exactly that form, and that, when the observed
variables are bounded, nothing is lost in the translation.  The family of
inequalities is the \emph{dual} of the coupling problem, in the sense of
linear programming, and the result is best understood as a statement that
this linear program has no duality gap.

The construction has four steps, and it is worth previewing them in words.
\begin{enumerate}[label=(\roman*)]
\item The two quantile restrictions are conditional moment inequalities,
      conditional on the regressor path and the individual effect.
\item As in the conditional moment inequality literature, multiplying by
      nonnegative functions of the conditioning variables and averaging
      produces unconditional moment inequalities.  Here the multiplier
      functions depend on \((x,a)\), and there is one for each period and
      each of the two restrictions.
\item The resulting unconditional moments still involve the unobserved
      \(A\).  Because the model says nothing about how \(A\) is related to
      the data beyond the quantile restrictions themselves, we replace \(A\)
      in each observation by whatever value makes the weighted moment
      largest.  This gives an upper bound that depends only on the observed
      distribution; if even the upper bound is negative, no allocation of
      individual effects can rescue the candidate slope.
\item Duality shows that the collection of these observable inequalities,
      over all nonnegative multiplier functions, is sharp: a candidate that
      admits no coupling always violates one of them strictly.
\end{enumerate}
The crossing inequalities of Section~\ref{sec:crossing} turn out to be the
inequalities produced by one particularly simple choice of multipliers.  The
result therefore explains what the crossing bounds capture, what they miss,
and how the finite-dimensional programs of Section~\ref{sec:lp} arise.

\subsection{From conditional restrictions to observable inequalities}
\label{subsec:populationdual}

Write \(W=(Y,X)\), let \(\mathcal W\) denote the support of \(\Pobs\), and fix
\(b\in\cB\).  Throughout this section, assume that \(\mathcal W\) is compact.
Let \(\mathcal X\) denote the support of the complete regressor path, and set
\begin{equation}
  \underline r_b
  =\min_{\substack{(y,x)\in\mathcal W\\1\leq t\leq T}}
       \{y_t-x_t'b\},
  \qquad
  \overline r_b
  =\max_{\substack{(y,x)\in\mathcal W\\1\leq t\leq T}}
       \{y_t-x_t'b\}.
  \label{eq:residualenvelope}
\end{equation}
The interval \(K_b=[\underline r_b,\overline r_b]\) contains all candidate
residuals.  It also contains every value of the individual effect that a feasible
coupling can use.

\begin{lemma}[A bounded range for the individual effect]
\label{lem:anchorcompact}
Every \(\pi\in\cC_b(\Pobs)\) satisfies \(P_\pi(A\in K_b)=1\).
\end{lemma}

\begin{proof}
Below \(\underline r_b\), the conditional probability of \(R_t(b)\leq A\)
is zero, violating its lower bound \(\tau>0\).  Above \(\overline r_b\), the
conditional probability of \(R_t(b)<A\) is one, violating its upper bound
\(\tau<1\).  Neither region can receive positive probability in a feasible
coupling.
\end{proof}

The lemma is the formal version of a simple observation: an individual effect
placed below every candidate residual cannot be a \(\tau\)-quantile of any of
them, and neither can one placed above every candidate residual.  It matters
here because it confines the unobserved effect to a compact interval, which is
what the duality argument below requires.

\paragraph{Step (i): the restrictions as conditional moment inequalities.}
For a coupling \(\pi\), the two restrictions in Definition~\ref{def:coupling}
say that
\[
  E_\pi[\one\{R_t(b)\leq A\}-\tau\mid X,A]\geq0,
  \qquad
  E_\pi[\tau-\one\{R_t(b)<A\}\mid X,A]\geq0,
  \qquad t=1,\ldots,T.
\]
These are \(2T\) conditional moment inequalities.  The conditioning variables
are the regressor path, which is observed, and the individual effect, which is
not.

\paragraph{Step (ii): instrumenting the restrictions.}
Multiplying a conditional moment inequality by a nonnegative function of the
conditioning variables and taking expectations preserves the inequality.  Let
\(\lambda_t^-(x,a)\geq0\) be the weight attached to the period-\(t\) lower
restriction for individuals with regressor path \(x\) and effect \(a\), and let
\(\lambda_t^+(x,a)\geq0\) be the weight attached to the period-\(t\) upper
restriction.  The superscripts identify the restriction, not the sign of the
weight; all weights are nonnegative.  A weight may vary with \((x,a)\), but not
with the outcome within an \((x,a)\) group, because it must be a function of the
conditioning variables only.

\begin{samepage}
For \(w=(y,x)\), define the weighted sum
\begin{equation}
\begin{aligned}
  S_b(w,a;\lambda)
  =\sum_{t=1}^T\Big[
  &\lambda_t^-(x,a)\{\one(y_t-x_t'b\leq a)-\tau\}
  \\[-1mm]
  &+\lambda_t^+(x,a)\{\tau-\one(y_t-x_t'b<a)\}\Big].
\end{aligned}
  \label{eq:dualscore}
\end{equation}
\par\end{samepage}
We refer to \(S_b\) as the \emph{score}.  It is a per-individual quantity:
for an individual with data \(w\) and effect \(a\), each term checks one
quantile restriction at that effect, scores it positively if the event that
should be frequent occurs and negatively otherwise, and weights the result.
Taking expectations under a feasible coupling, every term has nonnegative
conditional mean given \((X,A)\), so
\[
  E_\pi[S_b(W,A;\lambda)]\geq0
  \qquad\text{for every }\lambda\geq0.
\]
This is an unconditional moment inequality, but it cannot yet be checked, because
it averages over the unobserved \(A\).

\paragraph{Step (iii): removing the unobserved effect.}
Whatever value the coupling assigns to \(A\) for an individual with data \(W\),
that value lies in \(K_b\), so the score can be no larger than its maximum over
\(a\in K_b\).  Define
\begin{equation}
  \Psi_b(\lambda)
  =E_{\mathrm{obs}}\!\left[
       \max_{a\in K_b}S_b(W,a;\lambda)
    \right].
  \label{eq:dualfunctional}
\end{equation}
Then, under any feasible coupling,
\[
  0\leq E_\pi[S_b(W,A;\lambda)]
    \leq E_{\mathrm{obs}}\!\left[
          \max_{a\in K_b}S_b(W,a;\lambda)
        \right]
    =\Psi_b(\lambda).
\]
The right-hand side depends only on the observed distribution, the candidate
\(b\), and the chosen weights.  It gives each observation the most favorable
value of the individual effect \emph{for those weights}.  Consequently,
\begin{equation}
  \Psi_b(\lambda)<0
  \quad\Longrightarrow\quad
  b\notin\ThetaI .
  \label{eq:dualnecessity}
\end{equation}
A negative value says that the weighted restrictions fail on average even when
every individual is given the benefit of the doubt about his or her effect; no
coupling can then satisfy them.  Two features of the maximization deserve
emphasis.  The maximizing \(a\) is chosen separately for each observation and
separately for each \(\lambda\); it is not an estimate of anyone's individual
effect, and the collection of maximizers need not itself be a feasible coupling.
And the maximization is over the value of the effect only; the residuals inside
the indicators are the actual candidate residuals of the observation.

\paragraph{Step (iv): nothing is lost.}
The implication in \eqref{eq:dualnecessity} uses one \(\lambda\) at a time.
The important result is the converse: if the inequality \(\Psi_b(\lambda)\geq0\)
holds for \emph{every} nonnegative continuous weight function, then a feasible
coupling exists.  Write \(C_+(\mathcal X\times K_b)\) for the set of nonnegative
continuous functions on \(\mathcal X\times K_b\).

\begin{theorem}[Sharp observable duality]
\label{thm:observabledual}
Suppose \(\mathcal W\) is compact.  Then
\begin{equation}
  b\in\ThetaI
  \quad\Longleftrightarrow\quad
  \Psi_b(\lambda)\geq0
  \quad\text{for every }
  \lambda\in C_+(\mathcal X\times K_b)^{2T}.
  \label{eq:sharpdualcriterion}
\end{equation}
If \(b\notin\ThetaI\), some nonnegative continuous multiplier satisfies
\begin{equation}
  \Psi_b(\lambda)<0.
  \label{eq:strictdualcertificate}
\end{equation}
\end{theorem}

Necessity is \eqref{eq:dualnecessity}.  The proof of sufficiency is in
Appendix~\ref{app:dualityproof}; because the mechanism is the key to everything
that follows, we describe it here in words.

\paragraph{Why the theorem is true.}
Fix \(b\) and consider the set \(\Gamma_b\) of all joint distributions of
\((W,A)\) with \(W\)-marginal \(\Pobs\) and \(A\in K_b\).  This set is convex,
and it is compact in the topology of weak convergence because \(\mathcal W\) and
\(K_b\) are compact.  The coupling problem asks whether \(\Gamma_b\) contains an
element satisfying the \(2T\) conditional restrictions, each of which is an
infinite family of inequalities that are linear in the joint distribution.
It is therefore a linear program whose unknown is a probability measure.  Four
observations, corresponding to Steps 1--4 of the proof in
Appendix~\ref{app:dualityproof}, turn it into \eqref{eq:sharpdualcriterion}.

The first observation is that, for a \emph{fixed} \(\pi\in\Gamma_b\), the
conditional restrictions hold if and only if \(E_\pi[S_b(W,A;\lambda)]\geq0\)
for every nonnegative continuous \(\lambda\).  The reason is that each
conditional restriction says that a certain signed measure on
\((x,a)\)-space is nonnegative, and a signed measure on a compact metric
space is nonnegative exactly when it integrates every nonnegative continuous
function to a nonnegative number.  Continuous test functions are rich enough
to detect any negative part.

The second observation concerns the value of a game.  After normalizing the
scale of \(\lambda\), define
\[
  V_b=\sup_{\pi\in\Gamma_b}\ \inf_{\lambda}\ E_\pi[S_b(W,A;\lambda)].
\]
Think of a critic, who chooses weights to make the weighted restrictions look
as bad as possible, and of nature, who chooses the coupling to make them look
as good as possible.  In \(V_b\), nature commits to a coupling first and the
critic responds to it.  By the first observation, the critic's best response
is worth zero if the coupling is feasible and strictly negative otherwise, so
\(V_b=0\) exactly when a feasible coupling exists, and \(V_b<0\) otherwise.
(Compactness of \(\Gamma_b\) and upper semicontinuity of the inner infimum
ensure that the supremum is attained, so that the ``exactly'' is warranted.)

The third observation is that the order of moves can be reversed.  Sion's
minimax theorem applies because the payoff is affine in \(\pi\), linear in
\(\lambda\), and suitably semicontinuous, and it gives
\[
  V_b=\inf_{\lambda}\ \sup_{\pi\in\Gamma_b}\ E_\pi[S_b(W,A;\lambda)].
\]
Now the critic commits to weights first and nature responds to known weights.

The fourth observation is that nature's best response to fixed weights is easy
to describe.  The only restriction on \(\pi\in\Gamma_b\) is its
\(W\)-marginal; the model imposes no other link between \(A\) and \((Y,X)\).
So nature may attach to each observation whichever value of \(a\) maximizes
the score for that observation, and this is the best it can do:
\[
  \sup_{\pi\in\Gamma_b}E_\pi[S_b(W,A;\lambda)]
  =E_{\mathrm{obs}}\!\left[\max_{a\in K_b}S_b(W,a;\lambda)\right]
  =\Psi_b(\lambda).
\]
(A measurable selection of maximizers exists because the score is upper
semicontinuous in \(a\), so the pointwise choice does define a coupling.)
Combining the four observations, \(b\in\ThetaI\) if and only if
\(\inf_\lambda\Psi_b(\lambda)=0\), which is \eqref{eq:sharpdualcriterion}, and
when \(b\notin\ThetaI\) the infimum is negative, which is
\eqref{eq:strictdualcertificate}.

The observation that makes the unobserved effect disappear is the last one,
and it is exactly where the weakness of the model enters: because nothing
restricts the joint distribution of \(A\) and the data except the quantile
inequalities, and those inequalities are already priced by the weights,
nature's optimal allocation of effects is observation by observation.  The
price of this freedom is that the critic must be allowed \emph{every}
nonnegative weight function of \((x,a)\).  With a restricted class of weights
the argument still gives valid necessary conditions, and hence outer sets, but
sufficiency can fail.  Section~\ref{subsec:crossingdual} shows that the
crossing inequalities are precisely such a restricted class.

\paragraph{A criterion function.}
The theorem also delivers an objective function for the slope.  Because the
score is positively homogeneous in \(\lambda\), the scale of the weights can
be normalized, for instance by requiring the supremum norms of the \(2T\)
weight functions to sum to at most one.  Over the normalized class, the
zero weight gives \(\Psi_b(0)=0\), so
\begin{equation}
  \inf_{\lambda\ \mathrm{normalized}}\Psi_b(\lambda)
  \begin{cases}
    =0 & \text{if } b\in\ThetaI,\\
    <0 & \text{if } b\notin\ThetaI.
  \end{cases}
  \label{eq:dualcriterionpreview}
\end{equation}
The left-hand side is a scalar function of \(b\) that is zero exactly on the
identified set and negative elsewhere.  It plays the role that a population
criterion function plays in the set-identification literature
\citep{chernozhukovhongtamer2007}: whenever \(\ThetaI\neq\varnothing\), the
identified set is its set of maximizers.  The particular normalization of the
weights affects the magnitude of negative values but not the zero set, and
different normalizations will be convenient in different settings below.
Section~\ref{sec:lp} makes this operational, first exactly when the observed
support is finite, and then through a sequence of approximations when it is
not.

\begin{remark}[Weak and strict inequalities]
The two events in \eqref{eq:dualscore} are kept distinct.  This allows atoms at
the conditional quantile and gives the upper semicontinuity used in the duality
proof.  No no-ties or continuity assumption on the observed or latent variables
is imposed.
\end{remark}

\begin{remark}[Bounded Borel multipliers]
\label{rem:boreldual}
Continuous multipliers suffice for sharpness.  The necessity argument in
\eqref{eq:dualnecessity} also applies to bounded nonnegative Borel multipliers
whenever the pointwise supremum is measurable.  In particular, indicators of
regressor cells and indicators of intervals of \(a\) can be used as weights.
This observation permits both the crossing calculation below and the
finite-dimensional multiplier classes in Section~\ref{subsec:continuous}.
\end{remark}

\begin{remark}[Relation to other dual characterizations]
\label{rem:dualliterature}
The device of converting conditional moment inequalities into a family of
unconditional ones by nonnegative instrument functions is standard
\citep{khan2009inference,andrewsshi2013,chernozhukovleerosen2013}; what is specific here is that
one conditioning variable is latent and is removed by the pointwise
maximization in \eqref{eq:dualfunctional}.  The resulting characterization is in
the spirit of the optimal-transport duality used for incomplete models by
\citet{ekelandgalichonhenry2010} and \citet{galichonhenry2011}, of the
generalized instrumental variable framework of \citet{chesherrosen2017}, and of
the latent-variable moment problems studied by \citet{schennach2014}.  In the
finite-support case treated in Section~\ref{sec:lp}, it reduces to linear
programming over the distribution of the individual effect, as in
\citet{honoretamer2006}.
\end{remark}

\subsection{Crossing bounds and the role of the full multiplier class}
\label{subsec:crossingdual}

The crossing inequality in Proposition~\ref{prop:crossing} can be obtained from
Theorem~\ref{thm:observabledual} by a simple choice of weights, and seeing how
clarifies both what the crossing bounds do and why they are not the whole
story.

\paragraph{Crossing inequalities as a choice of weights.}
Work conditionally on \(X=x\), fix \(s\neq t\) and \(v\in K_b\).  The idea is to
let the weights switch at the threshold \(v\).  For effects at or below \(v\),
test only the period-\(t\) lower restriction; for effects above \(v\), test only
the period-\(s\) upper restriction:
\begin{equation}
  \lambda_t^-(a)=\frac{\one(a\leq v)}{\tau},
  \qquad
  \lambda_s^+(a)=\frac{\one(a>v)}{1-\tau},
  \label{eq:crossingmultipliers}
\end{equation}
with all other weights zero.  Consider an individual in the crossing event
\(C_{st}(v;b)=\{R_s(b)\leq v<R_t(b)\}\).  If this individual is assigned an
effect \(a\leq v\), then \(R_t(b)>v\geq a\), so the period-\(t\) weak event
fails and the score is \(\tau^{-1}(0-\tau)=-1\).  If instead \(a>v\), then
\(R_s(b)\leq v<a\), so the period-\(s\) strict event holds and the score is
\((1-\tau)^{-1}(\tau-1)=-1\).  Thus on the crossing event the score equals
\(-1\) for \emph{every} value of the effect: there is no benefit of the doubt to
give.  Off the crossing event, the score is no larger than
\[
  M_\tau
  =\max\left\{\frac{1-\tau}{\tau},\frac{\tau}{1-\tau}\right\}
  =\frac{c_\tau}{1-c_\tau}.
\]
Consequently, the observable dual inequality \(\Psi_b(\lambda)\geq0\), applied
conditionally on \(X=x\), gives
\[
  0\leq E_{\mathrm{obs}}\!\left[
       \sup_{a\in K_b}S_b(W,a;\lambda)\mid X=x
       \right]
  \leq
  -P^{\mathrm{obs}}(C_{st}\mid X=x)
  +M_\tau\{1-P^{\mathrm{obs}}(C_{st}\mid X=x)\},
\]
and rearranging yields \(P^{\mathrm{obs}}(C_{st}\mid X=x)\leq c_\tau\), which is
Proposition~\ref{prop:crossing}.  In Example~\ref{ex:running} at \(b=3/2\),
for instance, the period-2 residuals \(\{-3/2,-1/2\}\) all lie below the
period-1 residuals \(\{0,1\}\).  With \(s=2\), \(t=1\), and \(v=-1/2\), every
individual is in the crossing event, the score equals \(-1\) whatever effect
is assigned, and \(\Psi_b(\lambda)=-1<0\): the candidate is excluded by a
single observable inequality, without any reference to a coupling.

Read this way, a crossing inequality is a test that involves only two of the
\(2T\) restrictions, uses a single threshold, and gives every individual with
an effect on one side of the threshold the same weight.  The support-overlap
condition of Lemma~\ref{lem:overlap} is a further special case: if the
conditional supports of two periods' residuals do not overlap, any threshold
between them gives a crossing event of probability one.  The full multiplier
class allows weights that vary freely with \((x,a)\) and combine all \(2T\)
restrictions at once.  The next example shows that this additional freedom is
needed.

\begin{example}[The crossing inequalities are not sufficient]
\label{ex:crossingnotsharp}
Let \(T=2\) and \(\tau=1/2\), and suppose that at some candidate \(b\) and
some regressor path \(x\) the pair of candidate residuals
\((R_1(b),R_2(b))\) has the following conditional distribution:
\[
\begin{array}{c|ccccccc}
  (R_1(b),R_2(b)) & (3,2) & (3,-1) & (0,-2) & (2,1) & (0,-1) & (-1,0) & (-3,-2)\\
  \hline
  \text{probability}\times27 & 4 & 5 & 6 & 6 & 2 & 3 & 1
\end{array}
\]
The support-overlap condition holds: the period-1 residuals range over
\([-3,3]\) and the period-2 residuals over \([-2,2]\).  Every crossing
inequality holds as well: the largest value of
\(P^{\mathrm{obs}}\{R_s(b)\leq v<R_t(b)\mid X=x\}\) over
\(s\neq t\) and \(v\in\mathbb R\) is \(13/27<1/2\).  Nevertheless, no feasible
coupling exists at this path, so \(b\notin\ThetaI\).

To see why, consider the outcome vector with residuals \((3,-1)\), which
carries probability \(5/27\).  Whatever effect \(a\) it is assigned, it is
grouped with other outcome vectors assigned the same \(a\), and within that
group \(a\) must be a median of the period-1 residuals and of the period-2
residuals.  Suppose \(-1<a<3\).  Then the vector \((3,-1)\) lies strictly
above \(a\) in period 1 and strictly below \(a\) in period 2.  For \(a\) to
be a median in period 1, the group must contain at least as much mass with
\(R_1(b)\leq a\) as with \(R_1(b)>a\); for \(a\) to be a median in period 2,
at least as much mass with \(R_2(b)\geq a\) as with \(R_2(b)<a\).  Each
requirement says that a set of vectors carries at least half of the group's
mass, and the vector \((3,-1)\) belongs to neither set.  Adding the two
requirements, the group must therefore contain mass at least \(q\) of vectors
with \(R_1(b)\leq a\leq R_2(b)\), where \(q\) is the mass of \((3,-1)\) in
the group.  For \(-1<a\leq 0\) the only such vector is \((-1,0)\); for \(0<a<3\)
there is none.  A similar count for \(a\leq-1\) shows that only \((-1,0)\) and
\((-3,-2)\) can serve as partners, and for \(a\geq3\) no partner exists at all.
Summing over all groups, the vector \((3,-1)\) can be balanced by at most the
combined mass of \((-1,0)\) and \((-3,-2)\), which is \(4/27<5/27\).  The
requirement that one effect serve as median in both periods therefore cannot
be met, even though every pairwise threshold comparison is satisfied.

The obstruction is a matching condition of Hall's type: a set of outcome
vectors demands more balancing mass than its admissible partners can supply.
Crossing inequalities are the counting arguments of this kind generated by a
single threshold; the full multiplier class in
Theorem~\ref{thm:observabledual} contains all such counting arguments, and
more, since it also allows fractional weights.  For this configuration the
normalized dual criterion of Section~\ref{subsec:finitecriterion} takes the
value \(-1/378\).

The configuration is not artificial.  Take the single regressor path
\(x=(0,1)\) and let the observed outcome vector be
\((Y_1,Y_2)=(R_1,R_2+2)\), with \((R_1,R_2)\) distributed as in the table.
Then \(\beta=0\) belongs to the identified set, so the model is correctly
specified, and the table is the residual configuration at the candidate
\(b=2\).  That candidate is therefore outside \(\ThetaI\) even though it
satisfies every crossing inequality and the support-overlap condition; in
this design the crossing outer set is strictly larger than the sharp set.
\hfill\(\square\)
\end{example}

The example also indicates the practical division of labor.  Crossing and
support-overlap bounds are cheap, are valid with continuous outcomes and
regressors, and in several of the designs of Section~\ref{sec:longpanel} the
support-overlap bound is already exact.  But sharpness in general requires the
complete family.  This distinction suggests
an approximation strategy: retain a tractable class of weights, compute the
resulting outer set, and enlarge the class to recover more of the identifying
information.  Section~\ref{subsec:continuous} develops this approach for
continuous support.

\begin{remark}[Scope]
Compact support is used for the population duality theorem, not for the
structural definition or the coupling characterization of \(\ThetaI\).
The theorem covers bounded continuous and discrete distributions without
restricting \(A\) to a numerical grid.  An extension to unbounded support
would require additional tightness or coercivity conditions and is not
asserted here.
\end{remark}

\section{Computing the sharp identified set}
\label{sec:lp}

Sections~\ref{sec:coupling} and~\ref{sec:duality} characterize \(\ThetaI\) in
two equivalent ways: as the set of slopes for which a coupling exists, and as
the set of slopes at which every observable dual inequality holds.  This
section turns the two characterizations into computations.  The organizing
idea is the criterion function previewed in
\eqref{eq:dualcriterionpreview}: the smallest normalized value of the
observable dual functional is a scalar function of \(b\) that equals zero on
the identified set and is negative elsewhere.  Under finite observed support
a version of this function, under a convenient finite-dimensional
normalization of the weights, can be evaluated exactly by a linear program,
and then \emph{the entire identified set is the set of global maximizers of
the criterion}, or equivalently the set of global minimizers of its negative.
The identified set can therefore be found by optimizing over the coefficient
vector, profiling out nuisance coordinates when only some coefficients are of
interest, rather than by a separate feasibility decision at every point of a
prespecified grid.

The section proceeds in three settings of increasing generality.
\begin{enumerate}[label=(\alph*)]
\item \emph{Finite support of \((Y,X)\).}  The coupling problem is an
      ordinary linear program in a finite table of probabilities
      (Section~\ref{subsec:membership}), its dual is the criterion
      (Section~\ref{subsec:finitecriterion}), and the criterion is piecewise
      constant in \(b\), which permits exact recovery of its zero set
      (Section~\ref{subsec:fullset}).
\item \emph{Continuous outcomes, finitely many regressor paths.}  Binning
      the candidate residuals with a fixed monotone rule produces a
      finite-support problem of exactly the form in (a).  Each resolution
      gives an outer set, and refining the bins recovers the sharp set
      (Section~\ref{subsec:histogramcriterion}).
\item \emph{Continuous regressors.}  There is no finite reduction of the
      latent effect, so we approximate the \emph{dual} instead: restrict the
      multipliers to a finite-dimensional class, compute the resulting outer
      set, and enlarge the class.  Nested classes recover the sharp set
      (Sections~\ref{subsec:continuous} and~\ref{subsec:sieveevaluation}).
\end{enumerate}
Throughout, approximation of the multiplier class, numerical integration, and
optimization over coefficients are distinct operations, and we try to keep
them separate.  Proofs and supplementary implementation details are collected
in Appendix~\ref{app:computation}.  Example~\ref{ex:running} is carried
through the finite-support development to make each object concrete.

\subsection{Finite support: a linear program in probabilities}
\label{subsec:membership}

Suppose \(X\) has finite support.  Fix a positive-probability regressor path
\(X=x\) and a candidate \(b\).  Write the observed conditional distribution as
\begin{equation}
  P^{\mathrm{obs}}(Y=y^j\mid X=x)=p_j>0,
  \qquad j=1,\ldots,J,
  \qquad \sum_{j=1}^Jp_j=1,
  \label{eq:finitesupport}
\end{equation}
and define the residual at each outcome support point by
\begin{equation}
  z_{jt}(b)=y_t^j-x_t'b.
  \label{eq:finiteresidual}
\end{equation}
Here \(y^j\) is a complete \(T\)-period outcome vector, not a support point
for a single period.  By Remark~\ref{rem:globalpathwise}, a coupling exists
if and only if one exists separately in each positive-probability path, so
we may work path by path.

\paragraph{Which values of the individual effect must be considered?}
The coupling assigns to each outcome vector a distribution over values of
the individual effect.  In principle these values range over a continuum.
The first observation is that only finitely many of them matter.  Let
\(v_1<\cdots<v_P\) be the distinct values among the \(JT\) residual
coordinates \(z_{jt}(b)\).  It is enough to consider
\begin{equation}
  \mathcal A(b)=\{v_1,\ldots,v_P\},\qquad P\leq JT.
  \label{eq:valueanchors}
\end{equation}
This is an exact reduction, not an assumption that the individual effect is
discrete, and it is worth seeing why.  Suppose a coupling places some
individuals at an effect \(a\) lying strictly between two consecutive residual
values \(v_p<a<v_{p+1}\).  Move all of them to \(v_p\).  Every weak event
\(\{z_{jt}(b)\leq a\}\) is unchanged, because no residual lies in
\((v_p,a]\); every strict event \(\{z_{jt}(b)<a\}\) can only become less
likely, because the residuals equal to \(v_p\) no longer count.  The lower
quantile inequality is therefore unaffected and the upper one is weakly
relaxed, so the moved individuals still satisfy both restrictions.  Effects
outside \([v_1,v_P]\) are infeasible by the argument of
Lemma~\ref{lem:anchorcompact}, applied within the path.  The formal argument
in Appendix~\ref{app:finiteproof} also accounts for combining probability
mass that is moved to the same value.

\paragraph{What are the unknowns?}
Consider a table with one row for each outcome vector \(y^j\) and one column
for each candidate effect \(v_p\).  Its unknown entries are
\[
  q_{jp}=P^\ast(Y=y^j,A=v_p\mid X=x),
  \qquad
  w_p=\sum_{j=1}^Jq_{jp}.
\]
The distinction between observed and unobserved probabilities is:
\begin{center}
\renewcommand{\arraystretch}{1.08}
\begin{tabular}{lccc c}
\toprule
 & \(A=v_1\) & \(\cdots\) & \(A=v_P\) & Observed mass \\
\midrule
\(Y=y^1\) & \(q_{11}\) & \(\cdots\) & \(q_{1P}\) & \(p_1\) \\
\(\vdots\) & \(\vdots\) & & \(\vdots\) & \(\vdots\) \\
\(Y=y^J\) & \(q_{J1}\) & \(\cdots\) & \(q_{JP}\) & \(p_J\) \\
\midrule
Latent mass & \(w_1\) & \(\cdots\) & \(w_P\) & \(1\) \\
\bottomrule
\end{tabular}
\end{center}
The observed distribution fixes the row totals at \(p_j\).  The column totals
\(w_p\), which form the distribution of \(A\mid X=x\), are not specified.
Within a nonempty column \(p\), the normalized entries \(q_{jp}/w_p\) must
make \(v_p\) a conditional \(\tau\)-quantile of every residual coordinate.
The table is a transportation problem: the rows supply fixed amounts of
probability, the columns absorb whatever they receive, and the only question is
whether the supplies can be routed so that every column is balanced in every
period.

\begin{samepage}
This gives the following linear feasibility problem:
\begin{align}
  q_{jp}&\geq0,
  &&j=1,\ldots,J,\quad p=1,\ldots,P,
  \label{eq:lpnonnegative}\\*
  \sum_{p=1}^Pq_{jp}&=p_j,
  &&j=1,\ldots,J,
  \label{eq:lpmarginal}\\*
  \sum_{j:z_{jt}(b)\leq v_p}q_{jp}
     &\geq\tau\sum_{j=1}^Jq_{jp},
  &&p=1,\ldots,P,\quad t=1,\ldots,T,
  \label{eq:lpweak}\\*
  \sum_{j:z_{jt}(b)<v_p}q_{jp}
     &\leq\tau\sum_{j=1}^Jq_{jp},
  &&p=1,\ldots,P,\quad t=1,\ldots,T.
  \label{eq:lpstrict}
\end{align}
\par\end{samepage}
The last two lines are the conditional quantile inequalities multiplied by the
column mass \(w_p\).  This removes division by an unknown probability and
makes the restrictions linear.  A column with \(w_p=0\) imposes no restriction.
At the median, each nonempty column must place at least half of its mass on
or below \(v_p\), and at least half on or above \(v_p\), in every period.

The same entries \(q_{jp}\) appear in all \(T\) restrictions.  This is where
the common individual effect enters the computation, and it is the entire
source of identifying power: an outcome vector is allocated to a column once,
and that one allocation must pass the quantile check in every period.  If the
model instead allowed a different effect in each period, the problem would
separate across \(t\), and every \(b\) would be feasible, since all the mass
could then be placed, period by period, at the marginal \(\tau\)-quantile of
that period's residual.  No restriction is imposed on dependence across the
coordinates of \(y^j\): the rows are complete outcome vectors, so all
within-individual dependence is retained.

\begin{theorem}[Exact finite-support membership]
\label{thm:lp}
For the path \(X=x\) and candidate \(b\), a feasible conditional latent
coupling exists if and only if
\eqref{eq:lpnonnegative}--\eqref{eq:lpstrict} is feasible.  Consequently, if
\(X\) and every conditional outcome distribution have finite support,
\begin{equation}
  \ThetaI
  =\{b\in\cB:
      \text{\eqref{eq:lpnonnegative}--\eqref{eq:lpstrict} is feasible
      for every path of positive probability}\}.
  \label{eq:lpsharpset}
\end{equation}
\end{theorem}

A feasible table directly supplies a rationalizing joint distribution.  Its
column totals give \(P^\ast(A=v_p\mid X=x)=w_p\), while
\(q_{jp}/p_j\) gives \(P^\ast(A=v_p\mid Y=y^j,X=x)\).  Theorem~\ref{thm:sharp}
then completes the structural model.  Conversely, the residual-value reduction
shows that searching over these tables omits no feasible slope.

\begin{example}[Example~\ref{ex:running}, continued: the table]
\label{ex:runningtable}
In Example~\ref{ex:running} there is one path and two outcome vectors,
\(y^1=(0,0)\) and \(y^2=(1,1)\), each with probability \(1/2\), and the
residual pairs are \((0,-b)\) and \((1,1-b)\).  Take \(b=1/2\).  The distinct
residual values are \(\{-1/2,0,1/2,1\}\), so the table has two rows and four
columns.  One feasible table places both rows entirely in the column
\(A=0\):
\begin{center}
\renewcommand{\arraystretch}{1.08}
\begin{tabular}{lcccc c}
\toprule
 & \(A=-\tfrac12\) & \(A=0\) & \(A=\tfrac12\) & \(A=1\) & Observed mass\\
\midrule
\(Y=(0,0)\), residuals \((0,-\tfrac12)\) & 0 & \(\tfrac12\) & 0 & 0 & \(\tfrac12\)\\
\(Y=(1,1)\), residuals \((1,\tfrac12)\)  & 0 & \(\tfrac12\) & 0 & 0 & \(\tfrac12\)\\
\midrule
Latent mass & 0 & 1 & 0 & 0 & 1\\
\bottomrule
\end{tabular}
\end{center}
In the column \(A=0\), the period-1 residuals are \(\{0,1\}\) with equal
weight, so \(0\) is a median (half the mass is at or below it and none is
strictly below it); the period-2 residuals are \(\{-1/2,1/2\}\), and again
\(0\) is a median.  Both quantile constraints hold in both periods, and
\(b=1/2\) is feasible.  Placing the rows in different columns cannot work
for any \(b\neq0\): a column containing a single outcome vector must have
its effect equal to both residuals of that vector, which is impossible when
the two residuals differ.

Now take \(b=3/2\).  The residual pairs are \((0,-3/2)\) and \((1,-1/2)\),
and the candidate effects are \(\{-3/2,-1/2,0,1\}\).  Every period-2
residual is now below every period-1 residual.  In any column with
\(a\geq0\), both rows are strictly below \(a\) in period 2, so the strict
inequality \eqref{eq:lpstrict} fails for \(t=2\) unless the column is empty;
in any column with \(a<0\), both rows are strictly above \(a\) in period 1,
so the weak inequality \eqref{eq:lpweak} fails for \(t=1\) unless the column
is empty.  No table exists and \(b=3/2\) is excluded, as found in
Example~\ref{ex:running}. \hfill\(\square\)
\end{example}

\begin{algorithm}[Constructing a rationalizing probability table]
\label{alg:membership}
For a candidate \(b\), and in each positive-probability regressor path:
\begin{enumerate}[label=\arabic*.]
\item Compute all candidate residuals \(z_{jt}(b)\) and sort their distinct
      values to form \(\mathcal A(b)\).
\item Form the unknown probability table \((q_{jp})\) and impose
      \eqref{eq:lpnonnegative}--\eqref{eq:lpstrict}.
\item Solve the linear feasibility problem, equivalently a linear program
      with objective zero.
\end{enumerate}
Retain \(b\) if and only if every path is feasible.  The programs are separate
because the model leaves the distribution of \(A\mid X\) unrestricted across
regressor paths.
\end{algorithm}

The pathwise program has \(JP\leq TJ^2\) nonnegative unknowns, \(J\) row-total
equalities, and \(2TP\) quantile inequalities.  This construction is useful
when a feasible latent distribution is to be reported.  For searching over
slopes, the dual program below supplies a scalar objective with the same
identifying content.  Solving the primal in addition to that dual is not a
requirement of the characterization.

\subsection{A criterion function for the sharp identified set}
\label{subsec:finitecriterion}

The table of Section~\ref{subsec:membership} answers a yes-or-no question.
For searching over slopes it is more convenient to have a number, and the
dual of the table provides one.  The dual is the finite-support version of
the observable functional \(\Psi_b\) of Section~\ref{sec:duality}: the
critic now chooses a nonnegative weight for each quantile restriction
\emph{within each candidate-effect column}, that is, for each pair \((p,t)\)
and each of the two restrictions.

Use the finite-support notation above.  For weights
\(\lambda_{pt}^-\) and \(\lambda_{pt}^+\), define
\begin{equation}
\begin{aligned}
  s_{jp}(b;\lambda)=\sum_{t=1}^T\Big[
  &\lambda_{pt}^-\{\one(z_{jt}(b)\leq v_p)-\tau\}\\[-1mm]
  &+\lambda_{pt}^+\{\tau-\one(z_{jt}(b)<v_p)\}\Big].
\end{aligned}
  \label{eq:finitedualscore}
\end{equation}
This is the score \eqref{eq:dualscore}, evaluated at outcome vector
\(y^j\) and candidate effect \(v_p\); the weights are now finitely many
numbers rather than functions of \(a\), because only the values
\(v_1,\ldots,v_P\) of the effect matter.  They correspond to step functions
of \(a\) in Section~\ref{sec:duality}: setting
\(\lambda_t^\pm(a)=\lambda_{pt}^\pm\) for \(a\in[v_p,v_{p+1})\), and
extending the weights as constants below \(v_1\) and above \(v_P\), the
score for each outcome vector is maximized over \(a\) at one of the
\(v_p\), because moving \(a\) off a residual value can only lower the score.
For each \(j\), taking the largest score over \(p\) gives outcome vector
\(j\) the most favorable effect for those weights.  The observable dual
functional \(\Psi_b(\lambda)\) is therefore exactly the average of these
maxima, \(\sum_jp_j\max_ps_{jp}(b;\lambda)\), and by the necessity argument
of Section~\ref{sec:duality}, together with Remark~\ref{rem:boreldual}, a
negative value rules out the slope.

To find the most negative such average, introduce unrestricted variables
\(d_j\) to represent the maxima and solve
\begin{equation}
  \delta_x(b)=\min_{\lambda^-,\lambda^+,\,d}\sum_{j=1}^Jp_jd_j
  \label{eq:dualcertificateobjective}
\end{equation}
subject to
\begin{align}
  d_j&\geq s_{jp}(b;\lambda),
  &&j=1,\ldots,J,\quad p=1,\ldots,P,
  \label{eq:dualcertificateepigraph}\\*
  \lambda_{pt}^-,\lambda_{pt}^+&\geq0,
  &&p=1,\ldots,P,\quad t=1,\ldots,T,
  \label{eq:dualcertificatenonnegative}\\*
  \sum_{p=1}^P\sum_{t=1}^T(\lambda_{pt}^-+\lambda_{pt}^+)&\leq1.
  \label{eq:dualcertificatenormalize}
\end{align}
Minimizing over \(d_j\) sets it equal to \(\max_p s_{jp}(b;\lambda)\), so
the objective is indeed the average best score.  The last constraint fixes
the scale of the weights.  Without it, any negative average could be made
arbitrarily negative by multiplying the weights by a large constant, and the
program would be unbounded whenever the slope is excluded.  The normalization
does not change which slopes can be excluded, only the number reported.  It
is a sum over all weights, which is convenient for a finite program; the
population class \(\Lambda_b^1\) in Appendix~\ref{app:dualityproof} and the
sieve program of Section~\ref{subsec:sieveevaluation} normalize supremum
norms instead.  The two conventions produce different negative values at an
excluded slope but the same zero set.

\begin{proposition}[The finite-support dual criterion]
\label{prop:dualcertificate}
The program \eqref{eq:dualcertificateobjective}--\eqref{eq:dualcertificatenormalize}
has value \(\delta_x(b)\leq0\), with
\begin{align}
  \delta_x(b)=0
  &\quad\Longleftrightarrow\quad
  \text{the pathwise primal coupling LP is feasible},
  \label{eq:dualcertificatezero}\\
  \delta_x(b)<0
  &\quad\Longleftrightarrow\quad
  \text{the pathwise primal coupling LP is infeasible}.
  \label{eq:dualcertificatenegative}
\end{align}
A negative optimum supplies an observable inequality excluding \(b\).
\end{proposition}

The zero weights give value zero, so the minimum cannot be positive.  That it
is strictly negative precisely when no probability table exists is the
finite-dimensional theorem of the alternative (Farkas' lemma): either the
table is feasible, or there is a nonnegative combination of its quantile
constraints that contradicts the row totals, and the coefficients of that
combination are exactly the weights \(\lambda\).  The proof is in
Appendix~\ref{subsec:finitedual}.  Thus the primal table and the dual
criterion are two views of the same problem: the primal exhibits a
rationalizing latent distribution when one exists, the dual exhibits an
observable inequality that no latent distribution can satisfy when none does.

\begin{example}[Example~\ref{ex:running}, continued: a dual certificate]
\label{ex:runningdual}
Return to \(b=3/2\), where the residual pairs are \((0,-3/2)\) and
\((1,-1/2)\) and the candidate effects are \(v_1<v_2<v_3<v_4\), equal to
\(-3/2,-1/2,0,1\).  The crossing weights \eqref{eq:crossingmultipliers} with
threshold \(v=-1/2\), scaled to satisfy the normalization
\eqref{eq:dualcertificatenormalize}, are
\[
  \lambda_{p1}^-=\tfrac14\ \text{for } v_p\leq-\tfrac12\ (p=1,2),
  \qquad
  \lambda_{p2}^+=\tfrac14\ \text{for } v_p>-\tfrac12\ (p=3,4),
\]
and zero otherwise; the four positive weights sum to one.  In a column with
\(v_p\leq-1/2\), only the period-1 lower weight is active, and both period-1
residuals exceed \(v_p\), so \(s_{jp}=\tfrac14(0-\tfrac12)=-\tfrac18\) for
both rows.  In a column with \(v_p>-1/2\), only the period-2 upper weight is
active, and both period-2 residuals are strictly below \(v_p\), so
\(s_{jp}=\tfrac14(\tfrac12-1)=-\tfrac18\) again.  Every entry of the score
table equals \(-1/8\), so \(d_1=d_2=-1/8\) and the objective is
\(-1/8\).  No choice of weights does better under the normalization: each
bracket in \eqref{eq:finitedualscore} is at least \(-1/2\), so every score
in a column is at least \(-1/2\) times the total weight in that column; the
best column for a row is at least as good as the lightest column, whose
total weight is at most \(1/4\) when four columns share a total of one; hence
no objective value can fall below \(-1/8\).  Thus \(\delta_x(3/2)=-1/8\), the
crossing inequality is here the most damaging observable inequality, and the
negative value certifies that \(b=3/2\) is excluded.  For \(|b|\leq1\), the
feasible table in Example~\ref{ex:runningtable} guarantees
\(\delta_x(b)=0\). \hfill\(\square\)
\end{example}

\paragraph{Combining the regressor paths.}
Let \(x^1,\ldots,x^M\) be the positive-probability paths, with
\(\pi_m=P^{\mathrm{obs}}(X=x^m)>0\).  Define the objective and its
nonnegative loss version by
\begin{equation}
  \delta(b)=\sum_{m=1}^M\pi_m\delta_{x^m}(b),
  \qquad
  \mathcal L(b)=-\delta(b)\geq0.
  \label{eq:globaldualvalue}
\end{equation}
Because every pathwise value is nonpositive and every weight is positive, a
violation in one path cannot be offset by another.  Thus
\begin{equation}
  \ThetaI=\{b\in\cB:\delta(b)=0\}
          =\{b\in\cB:\mathcal L(b)=0\}.
  \label{eq:globaldualzeroset}
\end{equation}
If the maintained model is compatible with the observed distribution, so that
\(\ThetaI\neq\varnothing\), the optimal value is zero and
\begin{equation}
  \boxed{\displaystyle
  \ThetaI
  =\operatorname*{arg\,max}_{b\in\cB}\delta(b)
  =\operatorname*{arg\,min}_{b\in\cB}\mathcal L(b).}
  \label{eq:globaldualargmax}
\end{equation}
The argmax and argmin here denote \emph{sets of coefficients}, not a choice
of one coefficient from those sets.  These equalities turn sharp-set
computation into optimization of a scalar function, in the same way that the
identified set in a moment inequality model is the zero set of a population
criterion.  The nonemptiness qualification matters: a negative global maximum
of \(\delta\), or a positive global minimum of \(\mathcal L\), does not
identify a best-fitting slope as structurally admissible.  It reports that the
model is incompatible with the observed distribution on the coefficient region
being searched.

\subsection{Optimizing the criterion and recovering its solution set}
\label{subsec:fullset}

Equation~\eqref{eq:globaldualargmax} separates two tasks.  An inner calculation
evaluates \(\delta(b)\) by solving the pathwise dual LPs.  An outer
optimization searches over \(b\) for the set of global maximizers.  No
additional primal feasibility check is needed at an exactly evaluated zero,
and no prespecified grid over the entire coefficient vector is required.
This formulation is especially useful with several coefficients: the
coordinates can be optimized jointly, or nuisance coordinates can be profiled
out when only a lower-dimensional target is of interest.

The outer problem is not itself an ordinary LP.  Changing \(b\) changes
residual comparisons, so the criterion need not be smooth or concave.  The
inner minimization over multipliers remains part of each objective evaluation;
it cannot be replaced by maximizing jointly over coefficients and multipliers.
A routine that returns one zero-valued coefficient has found one member of
the identified set, not necessarily all its members.  Recovering the whole
set therefore means recovering the entire zero-level set, including any
separate components.  A local optimum, or repeated starts alone, does not
certify that this has been done.

\paragraph{Optional restrictions on the search.}
The support-overlap and crossing bounds in Section~\ref{sec:crossing} can
restrict the coefficient region before the full criterion is optimized.
Under finite support, they use residual extrema and sums of the known
probabilities, without a latent-probability calculation.  For \(T>2\), the
two-period programs give additional necessary restrictions, always conditional
on the complete regressor path.  These are optional outer restrictions, not
required steps in evaluating \(\delta\).  Passing them is not sufficient:
effects that rationalize separate pairs need not be compatible with one
effect for the whole panel.  For instance, with \(T=3\) and \(\tau=1/2\),
the four residual triples \((-1,-3,1)\), \((3,1,2)\), \((-1,1,-3)\), and
\((-2,0,2)\), with probabilities \(1/8,3/8,3/8,1/8\), admit a feasible
two-period table for each of the three period pairs but no three-period
table; the three-period criterion equals \(-1/112\).  Appendix~\ref{subsec:outer}
records the outer sets.

\paragraph{Exact recovery from residual orderings in low dimensions.}
The key structural fact about the finite-support criterion is that it is
piecewise constant in \(b\).  Every coefficient of the primal and dual
programs is an indicator of a comparison between two candidate residuals,
\(z_{jt}(b)\leq z_{hs}(b)\) or \(z_{jt}(b)<z_{hs}(b)\), because the
candidate effects \(v_p\) are themselves residual values.  Such a comparison
depends on \(b\) only through the sign of
\((y_t^j-y_s^h)-(x_t-x_s)'b\).  As long as no sign changes, the programs are
literally the same programs, and so are their values.  A single objective
evaluation therefore applies to the entire region of coefficient space with
that pattern of signs.  The criterion changes only when some pair of
residuals ties, and it is at such ties that both the weak and the strict
events can change.  In particular, gradient-based local optimization alone is
not suitable for locating all of the zero-valued regions of a piecewise
constant function.

For a scalar slope, possible changes occur at
\[
  b=\frac{y_t^{mj}-y_s^{mh}}{x_t^m-x_s^m},
  \qquad x_t^m\neq x_s^m.
\]
Sort these values, evaluate \(\delta\) once on each intervening open
interval, and evaluate it at the breakpoints themselves.  Retaining all
zero-valued intervals and breakpoints recovers the complete maximizing set
when the model is compatible, without a numerical coefficient grid.  In
Example~\ref{ex:running}, the residual ties occur at \(b\in\{-1,0,1\}\), and
evaluating the criterion on the seven resulting faces gives \(\delta(b)=0\) on
the five faces making up \([-1,1]\) and \(\delta(b)=-1/8\) on the two
unbounded intervals; the identified set is read off directly.

For a vector slope, the corresponding hyperplanes are
\begin{equation}
  H_{mjt,hs}
  =\{b:(x_t^m-x_s^m)'b=y_t^{mj}-y_s^{mh}\}.
  \label{eq:tiehyperplane}
\end{equation}
They divide coefficient space into regions with fixed residual orderings and
boundary regions on which some residuals tie.  A relatively open face is one
such region, including a boundary region without its own lower-dimensional
boundary.

\begin{proposition}[Exact recovery of the zero-level set]
\label{prop:arrangement}
Suppose \(\cB\) is polyhedral and the observed distribution has finite support.
Let \(\mathcal F\) contain the relatively open faces of every dimension
induced on \(\cB\) by \eqref{eq:tiehyperplane}.  The primal feasibility
status and the values of \(\delta\) and \(\mathcal L\) are constant on each
\(F\in\mathcal F\).  For any choice of one point \(b_F\in F\),
\begin{equation}
  \ThetaI=\bigcup_{\substack{F\in\mathcal F:\delta(b_F)=0}}F.
  \label{eq:faceunion}
\end{equation}
\end{proposition}

Tie faces must be retained in this calculation because ties change the weak
and strict quantile events; in Example~\ref{ex:running} the endpoints
\(b=\pm1\) are tie faces and belong to the identified set.  Enumeration gives
an exact full-set construction in low dimensions, but can be costly with many
coefficients.  It is one way to recover the optimizing set, not a prerequisite
for using the criterion.  Appendix~\ref{app:orderings} gives the proof and
enumeration details.

\paragraph{Profiling and bounds with several coefficients.}
Write \(b=(\psi,\eta)\), where \(\psi\) is the coefficient or subvector to
be reported and \(\eta\) contains the remaining coefficients.  Define
\begin{equation}
  \mathcal L_{\mathrm{prof}}(\psi)
  =\inf_{\eta:(\psi,\eta)\in\cB}\mathcal L(\psi,\eta),
  \qquad \inf\varnothing=+\infty.
  \label{eq:profiledualvalue}
\end{equation}
Under finite support the objective has only finitely many values, since there
are only finitely many residual orderings.  Thus the infimum in any nonempty
slice is attained.  It follows that the sharp projection is
\begin{equation}
  \Theta_{I,\psi}
  :=\{\psi:\text{some }(\psi,\eta)\in\ThetaI\}
  =\{\psi:\mathcal L_{\mathrm{prof}}(\psi)=0\}.
  \label{eq:profiledualzeroset}
\end{equation}
When \(\ThetaI\neq\varnothing\), this is also the set of global minimizers
of the profile criterion.  The nuisance coordinates are optimized over,
not placed on a grid.  If only the bounds on a linear functional \(c'b\)
are needed, they can instead be obtained from the two constrained problems
\begin{equation}
  \underline\psi_c
  =\inf_{b\in\cB:\mathcal L(b)=0}c'b,
  \qquad
  \overline\psi_c
  =\sup_{b\in\cB:\mathcal L(b)=0}c'b.
  \label{eq:dualcriterionbounds}
\end{equation}
Neither requires first listing the full joint set.  These are bounds, however:
componentwise endpoints do not describe the joint set or establish that every
point between endpoints belongs to the projection.  Recovering a projection
with gaps requires its zero-level set in \eqref{eq:profiledualzeroset}.

\subsection{Continuous outcomes and discrete regressors: a histogram criterion}
\label{subsec:histogramcriterion}

Suppose \(X\) has finitely many positive-probability paths
\(x^1,\ldots,x^M\), but \(Y\) is continuous.  The residual-value reduction
of Section~\ref{subsec:membership} is then unavailable, because there are
infinitely many residual values.  A natural idea is to coarsen the problem:
bin the candidate residuals and apply the finite-support machinery to the
joint histogram of the bins.  The reason this is legitimate is that model
(A) is preserved by monotone transformations applied simultaneously to the
residuals and to the effect.  If \(g\) is nondecreasing and \(A\) is a
conditional \(\tau\)-quantile of \(R_t(b)\) given \((X,A)\), then
\(\{R_t(b)\leq A\}\subseteq\{g(R_t(b))\leq g(A)\}\) and
\(\{g(R_t(b))<g(A)\}\subseteq\{R_t(b)<A\}\), so \(g(A)\) is a conditional
\(\tau\)-quantile of \(g(R_t(b))\) given \((X,A)\), and hence also given the
coarser information \((X,g(A))\) by iterated expectations.  A feasible
coupling for the original problem therefore induces a feasible coupling for
the binned problem.  Thus
the binned LP is a necessary condition, its zero set is an outer set at
every resolution, and, as we show, refining the bins recovers the sharp set.
The LP has the same constraint matrix at every coefficient value, which is
computationally convenient.

\paragraph{Form the joint residual histogram.}
Choose a fixed continuous strictly increasing bijection
\(H:\mathbb R\to(0,1)\), for example
\(H(r)=\tfrac12+\pi^{-1}\arctan r\). This only specifies bin cutpoints;
it imposes no distributional assumption and does not truncate the tails.
For \(Q=1,2,\ldots\), set \(K_Q=2^Q\) and define the ordered bin labels
\[
  C_{t,Q}(b)=\left\lfloor K_QH(Y_t-X_t'b)\right\rfloor,
  \qquad t=1,\ldots,T.
\]
Use the \emph{same} map for every period, so that a common effect can be
binned by the same rule as every residual.  For
\(j=(j_1,\ldots,j_T)\in\mathcal J_Q:=\{0,\ldots,K_Q-1\}^T\), let
\begin{equation}
  p_{mj,Q}(b)
  =P^{\mathrm{obs}}\!\left(
    (C_{1,Q}(b),\ldots,C_{T,Q}(b))=j\mid X=x^m
  \right).
  \label{eq:histogramprobabilities}
\end{equation}
These are joint bin probabilities, not products of marginal probabilities;
all within-panel dependence is retained at the chosen resolution.

\paragraph{Evaluate a fixed-matrix LP.}
Allow each label \(p\in\{0,\ldots,K_Q-1\}\) as a possible effect label.
For nonnegative weights \(\lambda_{pt}^-\) and \(\lambda_{pt}^+\), write
\[
  s_{jp,Q}(\lambda)
  =\sum_{t=1}^T\Big[
    \lambda_{pt}^-\{\one(j_t\leq p)-\tau\}
    +\lambda_{pt}^+\{\tau-\one(j_t<p)\}
  \Big].
\]
Apply Proposition~\ref{prop:dualcertificate} to these ordered labels:
\begin{equation}
\begin{aligned}
  \delta_{m,Q}^{\mathrm{bin}}(b)
  &=\min_{\lambda^\pm\geq0,\,d}
        \sum_{j\in\mathcal J_Q}p_{mj,Q}(b)d_j,\\
  &\text{subject to}\quad
       d_j\geq s_{jp,Q}(\lambda)\quad\text{for every }j,p,\\
  &\hspace{24mm}
       \sum_{p=0}^{K_Q-1}\sum_{t=1}^T
       (\lambda_{pt}^-+\lambda_{pt}^+)\leq1,
       \qquad d_j\in\mathbb R.
\end{aligned}
  \label{eq:histogramdual}
\end{equation}
Its value is nonpositive and equals zero exactly when the binned residual
law admits a common conditional quantile effect. All comparisons in
\eqref{eq:histogramdual} are between integers: only the objective
coefficients \(p_{mj,Q}(b)\) change with \(b\). The matrix can therefore be
constructed once per resolution and reused across paths and coefficient values.

\paragraph{Optimize and refine.}
With \(\pi_m=P^{\mathrm{obs}}(X=x^m)>0\), define
\begin{equation}
  \mathcal L_Q^{\mathrm{bin}}(b)
  =-\sum_{m=1}^M\pi_m\delta_{m,Q}^{\mathrm{bin}}(b)\geq0,
  \qquad
  \Theta_{I,Q}^{\mathrm{bin}}
  =\{b\in\cB:\mathcal L_Q^{\mathrm{bin}}(b)=0\}.
  \label{eq:histogramloss}
\end{equation}
When \(\ThetaI\neq\varnothing\),
\[
  \Theta_{I,Q}^{\mathrm{bin}}
  =\operatorname*{arg\,min}_{b\in\cB}\mathcal L_Q^{\mathrm{bin}}(b).
\]
Thus the loss can be supplied to an optimization over the coefficient vector,
or profiled for lower-dimensional targets, without a prespecified coefficient
grid. As above, the target is the \emph{set} of global minimizers, not one
reported solution; the outer optimization need not be convex.

The population sets satisfy
\begin{equation}
  \ThetaI\subseteq\Theta_{I,Q+1}^{\mathrm{bin}}
       \subseteq\Theta_{I,Q}^{\mathrm{bin}},
  \qquad
  \ThetaI=\bigcap_{Q\geq1}\Theta_{I,Q}^{\mathrm{bin}}.
  \label{eq:histogramrecovery}
\end{equation}
Outer validity is the monotone-transformation argument above: bin a feasible
\(A\) by the same rule as the residuals, and both conditional quantile
inequalities survive, conditional on the coarser information
\((X,C_Q(A))\) by iterated expectations.  Nesting follows because each coarser
label is a nondecreasing function of the finer label, so a coupling that is
feasible at resolution \(Q+1\) induces one at resolution \(Q\).  For sharp
recovery, rescale feasible labels and effects by \(K_Q\).  Their laws lie in
the compact cube \([0,1]^{T+1}\).  Any weak limit has residual marginal
\((H(R_1(b)),\ldots,H(R_T(b)))\) and retains the weak/strict quantile
inequalities.  Its effect has no mass at \(0\) or \(1\), since the residuals
lie in \((0,1)\) and \(0<\tau<1\).  Applying \(H^{-1}\) gives a feasible
original coupling in every path.

In implementation, compute the joint histogram, evaluate
\eqref{eq:histogramdual} within each path, and optimize
\(\mathcal L_Q^{\mathrm{bin}}\); then double \(K_Q\).
There are \(K_Q^T\) possible outcome bins, so this option is most attractive
for short panels and moderate resolutions, not uniformly faster than the
multiplier sieve. Replacing \eqref{eq:histogramprobabilities} by within-path
sample frequencies gives an empirical objective. The outer-set and recovery
claims above concern the population probabilities; empirical optimizing sets
are not automatically population bounds or confidence sets.

\subsection{Continuous support: a sequence of criterion functions}
\label{subsec:continuous}

Now maintain the compact-support conditions of Theorem~\ref{thm:observabledual}
and allow the regressor path to be continuously distributed.  There is now
no finite list of residual values on which the distribution of \(A\) can be
placed, and the histogram device would require partitioning the regressor
space as well as the outcomes.  We therefore change what is approximated.
Rather than discretizing the latent distribution, we approximate the
\emph{observable dual}: the critic of Section~\ref{sec:duality} is restricted
to a finite-dimensional class of weight functions, and the class is enlarged
step by step.

The logic is that of Theorem~\ref{thm:observabledual} read one direction at
a time.  Every inequality \(\Psi_b(\lambda)\geq0\) is a valid necessary
condition, whatever the class from which \(\lambda\) is drawn, so the set of
slopes passing all inequalities in a restricted class is an outer set.
Enlarging the class adds inequalities and shrinks the outer set.  And
because the excluding multiplier promised by the theorem is continuous, a
class rich enough to approximate continuous functions uniformly eventually
contains an excluding multiplier for every slope outside \(\ThetaI\).  The
approximation is to the \emph{tests}, not to the data, the residuals, or the
support of the individual effect.  This does not replace the conditioning
restriction in the model by a weaker one; at each resolution it checks a
subset of the model's observable implications.

A finite-dimensional class of weights makes this search tractable.  When
\(X\) has finite support, keep its exact paths and approximate only the
dependence of the weights on \(a\).  When \(X\) is continuous, also approximate
their dependence on the complete regressor path.

\paragraph{A concrete multiplier class.}
Let \(Q\) index approximation resolution, not the quantile \(\tau\).  At
resolution \(Q\), partition the support of \(X\) into \(G_Q\) cells
\(D_1^{(Q)},\ldots,D_{G_Q}^{(Q)}\), and select \(L_Q+1\) knots
\(\kappa_0<\cdots<\kappa_{L_Q}\) spanning \(K_b\).  Within each regressor
cell, a weight is specified by its nonnegative values at these knots and is
linearly interpolated between them.  In notation,
\begin{equation}
  \lambda_t^\pm(x,a)
  =\sum_{g=1}^{G_Q}\one\{x\in D_g^{(Q)}\}
     \sum_{\ell=0}^{L_Q}\theta_{tg\ell}^\pm B_\ell(a),
  \qquad \theta_{tg\ell}^\pm\geq0,
  \label{eq:cellhatbasis}
\end{equation}
where \(B_\ell\) is the piecewise-linear interpolation function equal to one
at knot \(\kappa_\ell\) and zero at the other knots.  These functions are
nonnegative and sum to one.  Thus \(\theta_{tg\ell}^\pm\) is simply the value
of the period-\(t\) weight at knot \(\ell\) in regressor cell \(g\).
If \(X\) has finite support, take its singleton paths as cells and refine only
the anchor knots.  Two properties make this class convenient: a weight is
nonnegative if and only if its knot values are nonnegative, so the cone
constraint is a sign constraint; and the score is piecewise affine in
\(a\), so its maximum over \(a\) can be computed exactly on a finite list of
points (Section~\ref{subsec:sieveevaluation}).

To repeat the point in operational terms: every residual remains
\(Y_t-X_t'b\), calculated from the actual \(Y\) and \(X\); there is no
partition of the outcomes and no replacement of a regressor by a cell
representative; and the maximization in \(\Psi_b\) still allows every
\(a\in K_b\), not only the knots.  Only the weights are restricted.

\paragraph{The objective and its optimizing set at resolution \texorpdfstring{\(Q\)}{Q}.}
Let \(\Lambda_{Q,b}\) be the cone of multipliers in \eqref{eq:cellhatbasis}.
Define the population outer set
\begin{equation}
  \Theta_{I,Q}^{\mathrm{dual}}
  =\{b\in\cB:\Psi_b(\lambda)\geq0
          \text{ for every }\lambda\in\Lambda_{Q,b}\}.
  \label{eq:dualsieveouter}
\end{equation}
To compute this condition, normalize the scale of the weights without
changing which slopes they exclude, and minimize the observable functional:
\begin{equation}
\begin{aligned}
  \Lambda_{Q,b}^{1}
  &=\left\{\lambda\in\Lambda_{Q,b}:
        \sum_{t=1}^T
        \bigl(\|\lambda_t^-\|_\infty+\|\lambda_t^+\|_\infty\bigr)
        \leq1\right\},\\
  V_Q(b)&=\inf_{\lambda\in\Lambda_{Q,b}^{1}}\Psi_b(\lambda).
\end{aligned}
  \label{eq:populationsievevalue}
\end{equation}
The zero weights are allowed, so \(V_Q(b)\leq0\).  Define
\(\mathcal L_Q(b)=-V_Q(b)\).  A negative value of \(V_Q\) excludes a slope;
a zero value retains it in the population outer set.  Consequently,
\begin{equation}
  \Theta_{I,Q}^{\mathrm{dual}}
  =\{b\in\cB:V_Q(b)=0\}
  =\{b\in\cB:\mathcal L_Q(b)=0\}.
  \label{eq:sievecriterionzeroset}
\end{equation}
When the maintained model is compatible, every such outer set is nonempty and
\begin{equation}
  \Theta_{I,Q}^{\mathrm{dual}}
  =\operatorname*{arg\,max}_{b\in\cB}V_Q(b)
  =\operatorname*{arg\,min}_{b\in\cB}\mathcal L_Q(b),
  \label{eq:sievecriterionargmax}
\end{equation}
with optimal value zero.  Thus one computes an entire outer set as the set
of optimal solutions, just as for finite support.  The difference is not the
use of optimization, but the identifying content of the objective: at finite
\(Q\), some infeasible slopes may still attain its optimum, because the
inequality that would exclude them is not yet in the class.

\begin{proposition}[Recovery by nested multiplier classes]
\label{prop:dualsieverecovery}
Maintain the compact-support conditions of Theorem~\ref{thm:observabledual}.
Suppose \(\Lambda_{Q,b}\subseteq\Lambda_{Q+1,b}\) for every \(Q\), and every
nonnegative continuous multiplier on \(\mathcal X\times K_b\) can be
approximated uniformly, component by component, by multipliers from these
classes.  Then
\begin{equation}
  \ThetaI\subseteq\Theta_{I,Q+1}^{\mathrm{dual}}
       \subseteq\Theta_{I,Q}^{\mathrm{dual}},
  \qquad
  \ThetaI=\bigcap_{Q\geq1}\Theta_{I,Q}^{\mathrm{dual}}.
  \label{eq:dualsieverecovery}
\end{equation}
\end{proposition}

The proof, in Appendix~\ref{app:sieveproof}, is short.  The containments are
the one-direction reading of the duality theorem given above.  For the
intersection, take \(b\notin\ThetaI\); Theorem~\ref{thm:observabledual}
supplies a continuous multiplier with \(\Psi_b(\lambda)=-\epsilon<0\); the
functional \(\Psi_b\) is Lipschitz in \(\lambda\) under the supremum norm,
with constant one for the sum over the \(2T\) components; so any class element
within \(\epsilon/2\) of \(\lambda\) in that norm also has a negative value and
excludes \(b\).  For \eqref{eq:cellhatbasis}, it is enough to subdivide the
regressor cells, retain existing knots when adding new ones, and let both the
largest cell diameter and the largest knot spacing tend to zero.  If \(X\) has
finite support, its cells stay fixed.  When the maintained model is
compatible, the same result can be written as
\begin{equation}
  \ThetaI
  =\bigcap_{Q\geq1}\operatorname*{arg\,max}_{b\in\cB}V_Q(b)
  =\bigcap_{Q\geq1}\operatorname*{arg\,min}_{b\in\cB}\mathcal L_Q(b).
  \label{eq:sieveargmaxrecovery}
\end{equation}
This is recovery by nested sets of optimal solutions.  Every infeasible slope
is excluded at some finite resolution; no common finite resolution or boundary
convergence rate is asserted.  Selecting only one optimizer at each resolution
would not in general recover the identified set.

\subsection{Evaluating the continuous-support procedure}
\label{subsec:sieveevaluation}

At a fixed resolution, there are finitely many multiplier coefficients, but
\(\Psi_b\) still contains a population expectation.  To evaluate it
numerically, suppose that the expectation is represented or approximated using
weighted points
\[
  P_N=\sum_{i=1}^N\omega_i\delta_{W_i},
  \qquad W_i=(Y_i,X_i),\quad
  \omega_i\geq0,\quad\sum_{i=1}^N\omega_i=1.
\]
The points might be numerical integration nodes or sample observations.  The
empirical distribution uses \(N=n\) and \(\omega_i=1/n\).  Their
interpretation matters for what can be concluded from the calculation.

\paragraph{The maximization over \(a\) is still exact.}
For point \(i\), compute \(r_{it}(b)=Y_{it}-X_{it}'b\) and form
\begin{equation}
  \mathcal A_{i,Q}(b)
  =\{\kappa_0,\ldots,\kappa_{L_Q}\}
      \cup\bigl(\{r_{i1}(b),\ldots,r_{iT}(b)\}\cap K_b\bigr).
  \label{eq:sievebreakpoints}
\end{equation}
\begin{samepage}
For weights in \eqref{eq:cellhatbasis}, the score is affine between consecutive
knots and residual values.  The weak/strict quantile formulation ensures that
its value at a residual is no smaller than either one-sided limit.  Hence the
maximum over \(K_b\) is attained on the finite list
\(\mathcal A_{i,Q}(b)\).  Appendix~\ref{app:breakpoints} proves this claim.
\par\end{samepage}

\begin{samepage}
\paragraph{The resulting linear program.}
Introduce one unrestricted variable \(d_i\) to bound the maximized score at
each point.  The other unknowns are the nonnegative knot values
\(\theta_{tg\ell}^\pm\) and bounds \(z_t^\pm\) used to normalize their scale.
The program is
\begin{equation}
  \widehat V_Q(b)=\min_{d,z,\theta}\sum_{i=1}^N\omega_i d_i
  \label{eq:empiricaldualsieve}
\end{equation}
subject to
\begin{align}
  d_i&\geq S_b(W_i,a;\lambda),
  &&i=1,\ldots,N,\quad a\in\mathcal A_{i,Q}(b),
  \label{eq:dualsieveepigraph}\\*
  0\leq\theta_{tg\ell}^\pm&\leq z_t^\pm,
  &&t=1,\ldots,T,\quad g=1,\ldots,G_Q,\quad \ell=0,\ldots,L_Q,
  \notag\\*
  \sum_{t=1}^T(z_t^-+z_t^+)&\leq1.
  \label{eq:dualsievenormalization}
\end{align}
\par\end{samepage}
The score is linear in the knot values once \(b\) is fixed.  Minimizing over
\(d_i\) implements its maximum at each point, and the last two lines implement
the normalization in \(\Lambda_{Q,b}^1\).  Thus this program evaluates the
chosen multiplier class exactly for the weighted law \(P_N\).  It has the same
structure as the finite-support dual \eqref{eq:dualcertificateobjective}: one
epigraph variable per observation, one sign-constrained weight per test, and
one normalization.

When \(X\) has finite support, the same calculation can be run separately
in each path.  Use the path's conditional outcome distribution, weights summing
to one within that path, and its residual-support interval \(K_{m,b}\).
Only the dependence on \(a\) is approximated.  For a pathwise implementation,
combine the nonpositive path values with strictly positive weights as in
\eqref{eq:globaldualvalue}; the aggregate is zero exactly when every path
value is zero.  With continuously distributed \(X\), use the regressor-cell
multipliers in \eqref{eq:cellhatbasis}; treating
each distinct sample path as a separate population conditioning cell would not
recover the conditional outcome distribution.  In either case both quantile
inequalities remain necessary, even when the observed outcomes are continuous,
since conditioning further on \(A\) can produce atoms.

\paragraph{Using the LP as an objective evaluator.}
For fixed \(Q\) and weighted law \(P_N\), supply
\(\widehat V_Q(b)\) to an outer optimization over \(b\), or minimize
\(\widehat{\mathcal L}_Q(b)=-\widehat V_Q(b)\).  At each requested
coefficient value, the inner LP evaluates the objective; a separate primal
feasibility calculation is unnecessary.  The numerical target is the set of
global optimizers, not just the first optimizer returned.  The same objective
can be used in constrained optimization when only particular coefficient
bounds are required.  As in finite support, evaluating an LP exactly does not
make the outer problem an LP or establish that a numerical search has found
all optima.

\paragraph{Population conclusions and numerical conclusions.}
With the population expectation, \(V_Q(b)<0\) excludes the candidate and
\(V_Q(b)=0\) retains it in an outer set.  Increasing \(Q\) recovers the sharp
set as in Proposition~\ref{prop:dualsieverecovery}.  When the expectation is
approximated by quadrature, integration error must be controlled before a
computed negative value can be interpreted as a population exclusion.  With
sample averages, \(\widehat V_Q(b)\) is a sample criterion, not a population
certificate or a confidence procedure.  Its optimal value need not be zero,
and its optimizing set is not automatically the population outer set in
\eqref{eq:sievecriterionargmax}.  Statistical inference is outside the scope
of this paper.

The same distinction applies to the interval \(K_b\): population calculations
use a support bound valid for the population; sample residual extrema give an
exact bound only for the empirical law.  A finite truncation of an unbounded
population is not covered by the compact-support result.  Finally, the
coefficient-space optimization is separate from refinement of the multiplier
class.  A mesh can be used for a numerical display, but is not required by the
criterion formulation and does not certify the unexamined coefficients.
Report the regressor partition, knot resolution, anchor interval, expectation
used, optimization method, and numerical tolerances, distinguishing objective
evaluation accuracy from recovery of the full optimizing set.

\paragraph{Summary.}
In every setting the computation has the same shape.  An inner linear program
plays the critic's role: it searches a class of nonnegative weights for the
observable inequality that a candidate slope violates most, after each
observation has been given the most favorable value of its individual
effect.  Its optimal value is nonpositive, and it is zero exactly when no
inequality in the class excludes the slope.  An outer optimization over the
slope then recovers the zero-level set of this value.  With finite support
the class is complete and the zero-level set is \(\ThetaI\); with continuous
support the class is finite-dimensional, the zero-level set is an outer set,
and nested classes drive it down to \(\ThetaI\).

\section{Covariate support, $T$, and point identification}
\label{sec:longpanel}

This section relates the fixed-\(T\) sharp sets to identification as the number
of periods grows.  Adding a period always weakly contracts the sharp set, but
nesting alone does not imply convergence to a singleton under model
\textup{(A)}.  We first establish nesting and isolate the exact
nonidentification created by time-invariant regressor directions.  We then give
the necessary-and-sufficient long-panel characterization furnished by observable
dual separation.  Finally, we derive a primitive support-spread bound.  The key
quantity in the new bound measures spread in a fixed coefficient direction
\emph{across regressor paths}.  It is strictly more informative than requiring
every realized path to have full directional rank and, in particular, can imply
point identification with only two periods.

\subsection{Projective setup and nesting}
\label{subsec:projectivenesting}

Let
\[
  (Y^{(\infty)},X^{(\infty)})
  =
  \{(Y_t,X_t):t\geq1\}
\]
have population law \(\Pobs_\infty\).  For each finite \(T\), write
\[
  Y^{(T)}=(Y_1,\ldots,Y_T)',
  \qquad
  X^{(T)}=(X_1,\ldots,X_T),
\]
and let \(\Pobs_T\) be the corresponding marginal law.  The
finite-dimensional observed laws are projectively consistent by construction.

For \(b\in\cB\), let \(\cC_{T,b}(\Pobs_T)\) be the collection of laws of
\((Y^{(T)},X^{(T)},A)\) having observed marginal \(\Pobs_T\) and satisfying,
for every \(t\leq T\),
\begin{align}
  P\{Y_t-X_t'b\leq A\mid X^{(T)},A\}&\geq\tau,
  \label{eq:Tcoupleweak}\\
  P\{Y_t-X_t'b<A\mid X^{(T)},A\}&\leq\tau
  \label{eq:Tcouplestrict}
\end{align}
almost surely.  Define the \(T\)-period sharp set by
\begin{equation}
  \ThetaI^{(T)}
  =
  \{b\in\cB:\cC_{T,b}(\Pobs_T)\neq\varnothing\}.
  \label{eq:Tsharpdef}
\end{equation}

\begin{proposition}[Monotonicity in the number of periods]
\label{prop:Tnesting}
For every \(T\geq1\),
\begin{equation}
  \ThetaI^{(T+1)}
  \subseteq
  \ThetaI^{(T)}.
  \label{eq:Tnesting}
\end{equation}
No stationarity, independence, exchangeability, continuity, or
common-distribution condition across periods is required.
\end{proposition}

\begin{proof}
Take \(b\in\ThetaI^{(T+1)}\), and choose
\(\pi_{T+1}\in\cC_{T+1,b}(\Pobs_{T+1})\).  Let \(\pi_T\) be its marginal law
of \((Y^{(T)},X^{(T)},A)\).  Projective consistency implies that the
\((Y^{(T)},X^{(T)})\)-marginal of \(\pi_T\) is \(\Pobs_T\).

For every \(t\leq T\), iterated expectations give
\begin{align*}
  &P_{\pi_T}\{Y_t-X_t'b\leq A\mid X^{(T)},A\}\\
  &\quad=
  E_{\pi_{T+1}}\!\left[
    P_{\pi_{T+1}}\{Y_t-X_t'b\leq A\mid X^{(T+1)},A\}
    \,\middle|\,X^{(T)},A
  \right]
  \geq\tau.
\end{align*}
The same argument for the strict event gives
\[
  P_{\pi_T}\{Y_t-X_t'b<A\mid X^{(T)},A\}
  \leq\tau.
\]
Thus \(\pi_T\in\cC_{T,b}(\Pobs_T)\), and
\(b\in\ThetaI^{(T)}\).
\end{proof}

\begin{remark}[The coupling need not be projectively consistent]
The definition of \(\ThetaI^{(T)}\) permits the rationalizing distribution of
\(A\) to change with \(T\).  Proposition \ref{prop:Tnesting} does not require a
researcher to select mutually consistent couplings at all horizons.  It uses only
the fact that any one feasible \((T+1)\)-period coupling can be marginalized to
produce a feasible \(T\)-period coupling.
\end{remark}

\subsection{Time-invariant directions and failure of automatic contraction}
\label{subsec:Tcounterexample}

A decreasing sequence of sets need not have a singleton intersection.  Under
model \textup{(A)}, the intersection can equal the entire parameter space.

\begin{example}[No contraction with time-invariant regressors]
\label{ex:TinvariantX}
Suppose \(X_t=Z\) almost surely for every \(t\), and suppose the observed process
is generated by
\begin{equation}
  Y_t=Z'\beta_0+A_0+U_t,
  \qquad t\geq1,
  \label{eq:TinvariantDGP}
\end{equation}
where, at every finite horizon, \(U_t\) satisfies model \textup{(A)}
conditional on \((X^{(T)},A_0)\).  For any \(b\in\cB\), define
\begin{equation}
  A^{(b)}
  =
  A_0+Z'(\beta_0-b),
  \qquad
  U_t^{(b)}=U_t.
  \label{eq:Tinvariantreparam}
\end{equation}
Then
\[
  Y_t=Z'b+A^{(b)}+U_t^{(b)}.
\]
Conditional on \(X^{(T)}\), the transformation between \(A_0\) and \(A^{(b)}\)
is one-to-one, so
\[
  \sigma(X^{(T)},A^{(b)})
  =
  \sigma(X^{(T)},A_0).
\]
The conditional quantile restrictions are therefore unchanged.  Every
\(b\in\cB\) is feasible at every horizon:
\begin{equation}
  \ThetaI^{(T)}=\cB
  \quad\text{for every }T,
  \qquad
  \bigcap_{T\geq1}\ThetaI^{(T)}=\cB.
  \label{eq:Tnoidentification}
\end{equation}
\end{example}

More generally, suppose there is a nonzero \(d\in\R^k\) and a measurable scalar
\(c(X^{(\infty)})\) such that
\begin{equation}
  X_t'd=c(X^{(\infty)})
  \qquad\text{for every }t.
  \label{eq:Tabsorbdirection}
\end{equation}
Then the change from \(\beta_0\) to \(\beta_0+d\) can be offset by subtracting
\(c(X^{(\infty)})\) from the individual effect.  Thus a necessary condition for
excluding every false slope is that no nonzero coefficient direction be constant
over time almost surely.  At a finite horizon this is the aggregate rank
condition \(\mathcal D_T=\{0\}\) from Section \ref{subsec:rankA}, equivalently
\(\operatorname{rank}(G_T)=k\) under square integrability.  Importantly, this
condition concerns a direction that is constant almost surely across the
population of regressor paths.  It does not require every realized path to
affinely span \(\R^k\), and it is not sufficient for singleton identification
under model \textup{(A)}.

\subsection{Exact long-panel characterization by dual separation}
\label{subsec:Tdual}

For each \(T\), suppose the compact-support conditions of Theorem
\ref{thm:observabledual} hold, and let \(\Psi_{T,b}(\lambda)\) denote the
observable dual functional formed from the first \(T\) periods.  Let
\(\Lambda_{T,b}\) be the corresponding cone of nonnegative continuous
multipliers.

\begin{proposition}[Finite-horizon separation of every false slope]
\label{prop:Tdualseparation}
Suppose \(\beta_0\in\ThetaI^{(T)}\) for every \(T\).  Then
\begin{equation}
  \bigcap_{T\geq1}\ThetaI^{(T)}=\{\beta_0\}
  \label{eq:Tsingleton}
\end{equation}
if and only if, for every \(b\in\cB\setminus\{\beta_0\}\), there are a finite
horizon \(T(b)\) and a multiplier
\(\lambda_b\in\Lambda_{T(b),b}\) such that
\begin{equation}
  \Psi_{T(b),b}(\lambda_b)<0.
  \label{eq:Tdualnegative}
\end{equation}
\end{proposition}

\begin{proof}
Suppose \eqref{eq:Tsingleton} holds, and fix \(b\neq\beta_0\).  Then
\(b\notin\ThetaI^{(T(b))}\) for some finite \(T(b)\).  The strict-separation
statement in Theorem \ref{thm:observabledual} supplies a multiplier satisfying
\eqref{eq:Tdualnegative}.

Conversely, suppose every false \(b\) satisfies \eqref{eq:Tdualnegative}.
Theorem \ref{thm:observabledual} implies
\(b\notin\ThetaI^{(T(b))}\), so \(b\) cannot belong to the intersection.  Since
\(\beta_0\) belongs to every finite-horizon set, the intersection is exactly
\(\{\beta_0\}\).
\end{proof}

Proposition \ref{prop:Tdualseparation} is exact but not primitive: it says that
every false slope must eventually be rejected by an observable finite-horizon
dual inequality.  The next result gives a primitive sufficient condition in
terms of composite-residual support and directional regressor spread.

\subsection{Necessity, sufficiency, and the relevant notion of rank}
\label{subsec:Tlogic}

Three logically different requirements must be distinguished.
\begin{enumerate}[label=(\roman*)]
\item The condition \(\mathcal D_T=\{0\}\) is necessary to eliminate exact
      absorption by the unrestricted fixed effect.  A time-invariant regressor,
      including an intercept, violates it.  Additional periods do not identify
      its coefficient without a normalization or restriction on \(A\).
\item Finite-horizon dual rejection of every false slope, as in Proposition
      \ref{prop:Tdualseparation}, is necessary and sufficient for
      \(\bigcap_T\Theta_I^{(T)}=\{\beta_0\}\), under the compactness conditions
      supporting sharp observable duality.
\item The support-envelope and directional-spread conditions below are primitive
      sufficient conditions.  They deliver explicit bounds and rates but are not
      claimed to be necessary.
\end{enumerate}

There are also two different quantitative notions of directional rank.  A
\emph{pathwise} minimum asks whether every realized regressor path spans all
coefficient directions.  A \emph{cross-path} minimum instead fixes a coefficient
direction first and then asks how widely its index varies over the population of
regressor paths.  The latter is the relevant and less restrictive notion for the
support argument below.

\subsection{A cross-path directional-spread bound}
\label{subsec:Tsupportrate}

Define the true composite residual
\begin{equation}
  V_t
  =
  Y_t-X_t'\beta_0
  =
  A_0+U_t.
  \label{eq:Tcompositeresidual}
\end{equation}
Assume that there are deterministic constants \(C_T<\infty\) and measurable
functions \(L_T(X^{(T)})\), \(U_T(X^{(T)})\) such that, for every \(t\leq T\),
\begin{align}
  P\!\left\{
    L_T(X^{(T)})\leq V_t\leq U_T(X^{(T)})
    \,\middle|\,X^{(T)}
  \right\}&=1,
  \label{eq:Tenvelope}\\
  U_T(X^{(T)})-L_T(X^{(T)})&\leq C_T
  \quad\text{almost surely}.
  \label{eq:Tenvelopewidth}
\end{align}
The conditional location of the common envelope may depend arbitrarily on the
regressor path.  Its width must be uniformly bounded across paths and across the
periods observed at horizon \(T\).

For \(d\in\R^k\) and a regressor path
\(x^{(T)}=(x_1,\ldots,x_T)\), define the directional range
\begin{equation}
  \omega_T(d;x^{(T)})
  =
  \max_{1\leq t\leq T}x_t'd
  -
  \min_{1\leq t\leq T}x_t'd.
  \label{eq:directionalrange}
\end{equation}
This function is nonnegative and positively homogeneous:
\(\omega_T(cd;x)=|c|\omega_T(d;x)\).

For comparison with the earlier pathwise formulation, define
\begin{align}
  \kappa_T^{\mathrm{pw}}(x^{(T)})
  &=
  \inf_{\|d\|=1}\omega_T(d;x^{(T)}),
  \label{eq:kappaTpath}\\
  \underline\kappa_T^{\mathrm{pw}}
  &=
  \essinf
  \left\{
    \kappa_T^{\mathrm{pw}}(X^{(T)})
  \right\}.
  \label{eq:kappaT}
\end{align}
Positivity of \(\underline\kappa_T^{\mathrm{pw}}\) requires almost every
realized path to span \(\R^k\), with a uniform quantitative margin.  This is a
strong condition.  For example, it is zero whenever \(T=2\) and \(k>1\).

The cross-path modulus instead reverses the order of the direction and
path operations:
\begin{equation}
  \overline\kappa_T
  =
  \inf_{\|d\|=1}
  \esssup
  \left\{
    \omega_T(d;X^{(T)})
  \right\}
  \in[0,\infty].
  \label{eq:kappaTcross}
\end{equation}
Here the direction \(d\) is fixed before variation over regressor paths is
measured.  Always
\begin{equation}
  \underline\kappa_T^{\mathrm{pw}}
  \leq
  \overline\kappa_T.
  \label{eq:kappacomparison}
\end{equation}
Indeed, for every unit \(d\),
\[
  \essinf\inf_{\|v\|=1}\omega_T(v;X^{(T)})
  \leq
  \esssup\omega_T(d;X^{(T)}),
\]
and taking the infimum over \(d\) gives
\eqref{eq:kappacomparison}.  The inequality can be strict, and
\(\overline\kappa_T\) can be infinite even when the pathwise modulus is zero.

\begin{theorem}[Sharp-set contraction under cross-path directional spread]
\label{thm:Tcontraction}
Suppose \(\beta_0\in\ThetaI^{(T)}\) and
\eqref{eq:Tenvelope}--\eqref{eq:Tenvelopewidth} hold.  Then every
\(b\in\ThetaI^{(T)}\), with \(d=b-\beta_0\), satisfies the pathwise restriction
\begin{equation}
  \omega_T(d;X^{(T)})\leq C_T
  \qquad\text{almost surely}.
  \label{eq:directionalnecessary}
\end{equation}
Equivalently,
\begin{equation}
  \ThetaI^{(T)}
  \subseteq
  \left\{
    \beta_0+d:
    \esssup\omega_T(d;X^{(T)})\leq C_T
  \right\}.
  \label{eq:directionalouterset}
\end{equation}
If \(0<\overline\kappa_T<\infty\), then
\begin{equation}
  \ThetaI^{(T)}
  \subseteq
  \left\{
    b\in\cB:
    \|b-\beta_0\|
    \leq
    \frac{C_T}{\overline\kappa_T}
  \right\}.
  \label{eq:Tcontractionbound}
\end{equation}
If \(\overline\kappa_T=\infty\), then
\begin{equation}
  \ThetaI^{(T)}=\{\beta_0\}.
  \label{eq:finiteTpointid}
\end{equation}
Finally, if \(\beta_0\in\ThetaI^{(T)}\) for every \(T\) and
\begin{equation}
  \frac{C_T}{\overline\kappa_T}\longrightarrow0,
  \qquad
  \text{with }C_T/\infty:=0,
  \label{eq:Tcontractioncondition}
\end{equation}
then
\begin{equation}
  \bigcap_{T\geq1}\ThetaI^{(T)}
  =
  \{\beta_0\},
  \label{eq:Tcontractionintersection}
\end{equation}
and
\begin{equation}
  d_H\!\left(\ThetaI^{(T)},\{\beta_0\}\right)
  =
  \sup_{b\in\ThetaI^{(T)}}\|b-\beta_0\|
  \leq
  \frac{C_T}{\overline\kappa_T}
  \longrightarrow0.
  \label{eq:Tradialrate}
\end{equation}
\end{theorem}

\begin{proof}
Fix \(T\), take \(b\in\ThetaI^{(T)}\), and write \(d=b-\beta_0\).
The candidate residual satisfies
\begin{equation}
  R_t(b)
  =
  Y_t-X_t'b
  =
  V_t-X_t'd.
  \label{eq:Tcandidatefromtrue}
\end{equation}
Conditional on \(X^{(T)}=x^{(T)}\), conditions
\eqref{eq:Tenvelope}--\eqref{eq:Tenvelopewidth} imply that the conditional
support of \(R_t(b)\) is contained in
\begin{equation}
  [L_T(x^{(T)})-x_t'd,\,
   U_T(x^{(T)})-x_t'd].
  \label{eq:Tcandidateenvelope}
\end{equation}

Because \(b\) is feasible, Lemma \ref{lem:overlap}, applied to the first \(T\)
periods, implies that the actual conditional residual support intervals have a
common point for almost every \(x^{(T)}\).  Since those intervals are contained
in \eqref{eq:Tcandidateenvelope}, the larger envelope intervals must also have a
common point.  Hence, for almost every \(x^{(T)}\), there is an
\(a=a(x^{(T)})\) such that
\[
  L_T(x^{(T)})-x_t'd
  \leq a
  \leq
  U_T(x^{(T)})-x_t'd
  \qquad\text{for every }t\leq T.
\]
Equivalently,
\[
  L_T(x^{(T)})-a
  \leq x_t'd
  \leq
  U_T(x^{(T)})-a
  \qquad\text{for every }t\leq T.
\]
All \(T\) values \(x_t'd\) therefore lie in an interval of length at most
\(C_T\), which proves \eqref{eq:directionalnecessary}.

If \(d\neq0\), let \(v=d/\|d\|\).  Positive homogeneity and
\eqref{eq:directionalnecessary} give
\[
  \|d\|\omega_T(v;X^{(T)})\leq C_T
  \qquad\text{almost surely}.
\]
Taking the essential supremum over regressor paths yields
\[
  \|d\|
  \esssup\omega_T(v;X^{(T)})
  \leq C_T.
\]
Since
\[
  \esssup\omega_T(v;X^{(T)})
  \geq\overline\kappa_T,
\]
we obtain
\(\|b-\beta_0\|\overline\kappa_T\leq C_T\).  This proves
\eqref{eq:Tcontractionbound} when the modulus is finite and positive.  If
\(\overline\kappa_T=\infty\), the same inequality is impossible for any
\(d\neq0\); feasibility of \(\beta_0\) then gives
\eqref{eq:finiteTpointid}.  The intersection and Hausdorff conclusions follow
from \eqref{eq:Tcontractionbound}, the convention \(C_T/\infty=0\), and
\(\beta_0\in\ThetaI^{(T)}\).
\end{proof}

The theorem formalizes the reading of Example~\ref{ex:running}.  There the
composite residual \(V_t=Y_t\) takes values in \(\{0,1\}\), so \(C_2=1\); the
single path \((0,1)\) gives \(\omega_2(d;x)=|d|\); and
\eqref{eq:directionalnecessary} yields \(|b|\leq1\), which is exactly the
sharp set.  In general the bound is only an outer set, because it uses only
the support of the composite residual and not the shape of its distribution;
the examples in Section~\ref{subsec:Tgeometric} show that it is nevertheless
exact in several natural designs.

\subsection{Geometric and high-quantile refinements}
\label{subsec:Tgeometric}

The Euclidean bound in Theorem \ref{thm:Tcontraction} summarizes the
directional restriction \eqref{eq:directionalnecessary} by its least favorable
direction.  The underlying restriction is generally anisotropic and can be
represented more sharply by a convex set generated by the within-panel
regressor differences.

For each \(T\), define the closed symmetric convex difference set
\begin{equation}
  \mathcal K_T
  =
  \overline{\operatorname{co}}
  \left(
    \bigcup_{1\leq s,t\leq T}
    \operatorname*{ess\,ran}(X_t-X_s)
  \right)
  \subseteq\mathbb R^k,
  \label{eq:KTdifferencebody}
\end{equation}
where the essential ranges are taken under the observed distribution of
\(X^{(T)}\).  Because the union includes both \(X_t-X_s\) and \(X_s-X_t\),
the set \(\mathcal K_T\) is symmetric about zero.  It need not be bounded.

Let
\begin{equation}
  h_{\mathcal K_T}(d)
  =
  \sup_{g\in\mathcal K_T}g'd
  \in[0,\infty]
  \label{eq:KTsupportfunction}
\end{equation}
denote its extended support function.  Define also its polar set
\begin{equation}
  \mathcal K_T^\circ
  =
  \left\{
    d\in\mathbb R^k:
    h_{\mathcal K_T}(d)\leq1
  \right\}.
  \label{eq:KTpolar}
\end{equation}

\begin{proposition}[Geometric directional outer bound]
\label{prop:Tgeometricbound}
For every \(d\in\mathbb R^k\),
\begin{equation}
  h_{\mathcal K_T}(d)
  =
  \operatorname*{ess\,sup}
  \omega_T(d;X^{(T)}).
  \label{eq:supportequalsrange}
\end{equation}
Consequently, under the assumptions of Theorem \ref{thm:Tcontraction},
\begin{equation}
  \Theta_I^{(T)}
  \subseteq
  \beta_0+\mathcal E_T(C_T),
  \qquad
  \mathcal E_T(C_T)
  =
  \left\{
    d\in\mathbb R^k:
    h_{\mathcal K_T}(d)\leq C_T
  \right\}.
  \label{eq:anisotropicouterset}
\end{equation}
If \(C_T>0\), then
\begin{equation}
  \mathcal E_T(C_T)=C_T\mathcal K_T^\circ.
  \label{eq:scaledpolar}
\end{equation}
Moreover,
\begin{equation}
  \left\{
    d:h_{\mathcal K_T}(d)=0
  \right\}
  =
  \mathcal D_T.
  \label{eq:KTkernel}
\end{equation}
Thus the kernel of the geometric bound is exactly the absorption space of
time-invariant coefficient directions.
\end{proposition}

\begin{proof}
For any regressor path \(x^{(T)}\),
\[
  \omega_T(d;x^{(T)})
  =
  \max_{1\leq s,t\leq T}(x_t-x_s)'d.
\]
Because the maximum is over finitely many pairs,
\[
  \operatorname*{ess\,sup}
  \omega_T(d;X^{(T)})
  =
  \max_{1\leq s,t\leq T}
  \operatorname*{ess\,sup}(X_t-X_s)'d.
\]
Taking the closed convex hull does not change the supremum of a linear
functional.  This proves \eqref{eq:supportequalsrange}.  The inclusion
\eqref{eq:anisotropicouterset} then follows directly from
\eqref{eq:directionalnecessary}, and \eqref{eq:scaledpolar} follows from
positive homogeneity of the support function.

Finally, symmetry of \(\mathcal K_T\) implies
\(h_{\mathcal K_T}(d)=0\) if and only if
\[
  (X_t-X_s)'d=0
  \quad\text{almost surely for every }s,t\leq T.
\]
By definition, this is equivalent to \(d\in\mathcal D_T\).
\end{proof}

\begin{remark}[The geometric bound is sharper than its radial summary]
\label{rem:Tgeometricradial}
The set \(\mathcal E_T(C_T)\) retains the direction-specific variation in the
regressors.  By contrast, the Euclidean bound in
\eqref{eq:Tcontractionbound} replaces the entire support function by
\[
  \overline\kappa_T
  =
  \inf_{\|d\|=1}h_{\mathcal K_T}(d).
\]
Hence
\begin{equation}
  \mathcal E_T(C_T)
  \subseteq
  \left\{
    d:\|d\|\leq C_T/\overline\kappa_T
  \right\}
  \label{eq:geometricversusradial}
\end{equation}
whenever \(0<\overline\kappa_T<\infty\).  The inclusion can be strict when
regressor variation differs substantially across coefficient directions.
Accordingly, \eqref{eq:anisotropicouterset} is preferable for reporting
directional or componentwise bounds, while
\eqref{eq:Tcontractionbound} remains a convenient scalar summary.
\end{remark}

The essential supremum in \(\overline\kappa_T\) can be determined by extremely
rare regressor paths.  This is appropriate for population identification, but
it can be difficult to estimate and can obscure the amount of identifying
variation present over economically relevant parts of the regressor
distribution.  The following high-quantile modulus provides an intermediate
measure.

For a nonnegative random variable \(Z\), define its left-continuous
\(u\)-quantile by
\begin{equation}
  q_u(Z)
  =
  \inf\left\{
    z\in\mathbb R:
    P(Z\leq z)\geq u
  \right\}.
  \label{eq:quantiledefinitionT}
\end{equation}
For \(\eta\in(0,1)\), define
\begin{equation}
  \overline\kappa_{T,\eta}
  =
  \inf_{\|d\|=1}
  q_{1-\eta}
  \left(
    \omega_T(d;X^{(T)})
  \right).
  \label{eq:quantilespreadmodulus}
\end{equation}

\begin{proposition}[High-quantile directional-spread bound]
\label{prop:Tquantilebound}
Under the assumptions of Theorem \ref{thm:Tcontraction}, every
\(b\in\Theta_I^{(T)}\) satisfies, for every \(\eta\in(0,1)\),
\begin{equation}
  q_{1-\eta}
  \left(
    \omega_T(b-\beta_0;X^{(T)})
  \right)
  \leq C_T.
  \label{eq:quantilerangenecessary}
\end{equation}
Consequently, if \(\overline\kappa_{T,\eta}>0\), then
\begin{equation}
  \Theta_I^{(T)}
  \subseteq
  \left\{
    b\in\mathcal B:
    \|b-\beta_0\|
    \leq
    \frac{C_T}{\overline\kappa_{T,\eta}}
  \right\}.
  \label{eq:Tquantilecontraction}
\end{equation}

The exact geometric outer set also admits the representation
\begin{equation}
  \mathcal E_T(C_T)
  =
  \bigcap_{\eta\in(0,1)}
  \left\{
    d:
    q_{1-\eta}
    \left(
      \omega_T(d;X^{(T)})
    \right)
    \leq C_T
  \right\}.
  \label{eq:geometricquantileintersection}
\end{equation}
\end{proposition}

\begin{proof}
Equation \eqref{eq:directionalnecessary} implies
\[
  \omega_T(b-\beta_0;X^{(T)})\leq C_T
  \quad\text{almost surely}.
\]
Every quantile of this random variable is therefore no larger than \(C_T\),
which proves \eqref{eq:quantilerangenecessary}.

If \(b\neq\beta_0\), write
\[
  b-\beta_0=\|b-\beta_0\|v,
  \qquad \|v\|=1.
\]
Positive homogeneity of \(\omega_T\) and of nonnegative quantiles gives
\[
  q_{1-\eta}
  \left(
    \omega_T(b-\beta_0;X^{(T)})
  \right)
  =
  \|b-\beta_0\|
  q_{1-\eta}
  \left(
    \omega_T(v;X^{(T)})
  \right).
\]
The final quantile is at least
\(\overline\kappa_{T,\eta}\), proving
\eqref{eq:Tquantilecontraction}.

Finally, for any nonnegative random variable \(Z\),
\[
  \operatorname*{ess\,sup}Z\leq C_T
  \quad\Longleftrightarrow\quad
  q_{1-\eta}(Z)\leq C_T
  \quad\text{for every }\eta\in(0,1).
\]
Combining this equivalence with
\eqref{eq:supportequalsrange} proves
\eqref{eq:geometricquantileintersection}.
\end{proof}

\begin{remark}[Population identification and statistical stability]
\label{rem:Tquantilestatistical}
The modulus \(\overline\kappa_T\) uses the extreme support of the regressor
process.  Consequently, point identification based on
\(\overline\kappa_T=\infty\) can rely on regressor paths having arbitrarily
small probability.  For a fixed \(\eta>0\),
\(\overline\kappa_{T,\eta}\) instead measures directional variation over all
but an \(\eta\)-fraction of the population and is therefore more stable under
tail truncation.

Given independent sampled regressor paths
\(X_1^{(T)},\ldots,X_n^{(T)}\), a natural sample analogue is
\begin{equation}
  \widehat{\overline\kappa}_{T,\eta}
  =
  \inf_{\|d\|=1}
  \widehat q_{1-\eta}
  \left(
    \left\{
      \omega_T(d;X_i^{(T)})
    \right\}_{i=1}^n
  \right).
  \label{eq:samplequantilemodulus}
\end{equation}
This quantity is useful as a design-strength diagnostic and can be reported
over several values of \(\eta\).  It is not, by itself, a confidence bound:
formal inference requires uniform control over directions \(d\), sampling
error in the empirical quantiles, and any estimation of \(C_T\).

A finite grid of directions generally gives an upper approximation to the
infimum in \eqref{eq:samplequantilemodulus} and can therefore overstate
identifying strength.  Certified lower bounds require either global
optimization or an explicit covering-error correction.
\end{remark}

\ \ 

\subsubsection{Four Examples of Identified Sets}
We provide 4 examples for designs where we analytically compute the identified set. These examples are interesting because they highlight the importance of the support of the regressors in shaping the volume of the identified set. 
\begin{example}[{\bf Set Identification with Binary $X$}]
\label{ex:Tbinarynonpoint}
Let \(T=2\) and \(k=2\). Write
\[
  X_t=(D_{1t},D_{2t})',
\]
where the four variables \(\{D_{jt}:j,t\in\{1,2\}\}\) are mutually
independent Bernoulli\((1/2)\). Let \(\varepsilon\) be independent of
\(X^{(2)}\) and satisfy, for some \(\sigma>0\),
\[
  P(\varepsilon=-\sigma)=\tau,
  \qquad
  P(\varepsilon=\sigma)=1-\tau.
\]
Suppose \(\beta_0\) is an interior point of \(\cB\), and generate the observed
outcomes by
\[
  Y_t=X_t'\beta_0+\varepsilon,
  \qquad t=1,2.
\]
This is a special case of model \textup{(A)} with \(A_0=0\) and
\(U_1=U_2=\varepsilon\): zero is a conditional \(\tau\)-quantile of each
\(U_t\), including under the paper's weak/strict convention. The sharp set is
\begin{equation}
  \ThetaI^{(2)}
  =
  \cB\cap
  \left\{
    \beta_0+d:
    |d_1|+|d_2|\leq 2\sigma
  \right\}.
  \label{eq:Tbinarysharpdiamond}
\end{equation}
Because \(\beta_0\) is an interior point of \(\cB\) and \(\sigma>0\), this set
contains slopes other than \(\beta_0\).

To see necessity, fix \(b=\beta_0+d\). Conditional on a regressor path
\((X_1,X_2)=(x_1,x_2)\), the support interval of the candidate residual
\(R_t(b)=Y_t-X_t'b\) is
\[
  [-\sigma-x_t'd,\,\sigma-x_t'd].
\]
Lemma \ref{lem:overlap} therefore requires
\[
  |(x_2-x_1)'d|\leq 2\sigma.
\]
Every pair \(x_1,x_2\in\{0,1\}^2\) has positive probability, and hence
\[
  \sup_{x_1,x_2\in\{0,1\}^2}|(x_2-x_1)'d|
  =|d_1|+|d_2|.
\]
Thus feasibility implies the restriction in
\eqref{eq:Tbinarysharpdiamond}.

For sufficiency, suppose \(|d_1|+|d_2|\leq2\sigma\), let
\(\iota=(1,1)'\), and define the candidate fixed effect and disturbances by
\[
  A^{(b)}=-\frac{1}{2}\iota'd,
  \qquad
  U_t^{(b)}
  =Y_t-X_t'b-A^{(b)}
  =\varepsilon+\left(\frac{1}{2}\iota-X_t\right)'d.
\]
For every \(x_t\in\{0,1\}^2\),
\[
  \left|
    \left(\frac{1}{2}\iota-x_t\right)'d
  \right|
  \leq \frac{1}{2}(|d_1|+|d_2|)
  \leq\sigma.
\]
Consequently, conditional on \(X^{(2)}\), the two support points of
\(U_t^{(b)}\) straddle zero, with probability \(\tau\) on the lower point and
probability \(1-\tau\) on the upper point. Zero is therefore a conditional
\(\tau\)-quantile of \(U_t^{(b)}\) for both periods. This constructs an
admissible latent structure for every \(b\) in the right-hand side of
\eqref{eq:Tbinarysharpdiamond}, proving sharpness.

This failure of point identification is not caused by an absorbed
time-invariant direction. Indeed,
\[
  G_2
  =E[(X_2-X_1)(X_2-X_1)']
  =\frac{1}{2}I_2\succ0,
  \qquad
  \mathcal D_2=\{0\}.
\]
Instead, the binary design has bounded directional spread. The true composite
residual \(V_t=\varepsilon\) admits the envelope
\(L_2=-\sigma\), \(U_2=\sigma\), and hence
\[
  C_2=2\sigma,
  \qquad
  \esssup\omega_2(d;X^{(2)})=|d_1|+|d_2|,
  \qquad
  \overline\kappa_2
  =\inf_{\|d\|=1}(|d_1|+|d_2|)=1.
\]
Thus the pathwise restriction in Theorem \ref{thm:Tcontraction} is sharp in
this example but leaves the nondegenerate diamond
\eqref{eq:Tbinarysharpdiamond}.  \hfill $\square$
\end{example}
\ \ 

\begin{example}[{\bf Partial Identification with Uniform Regressors}]
\label{ex:Tuniformnonpoint}
Let \(T=2\), \(k=2\), and write
\[
  X_t=(Z_t,W_t)',
\]
where the four variables \(\{Z_t,W_t:t\in\{1,2\}\}\) are mutually independent
and
\[
  Z_t\stackrel{\mathrm{iid}}{\sim}\operatorname{Unif}[0,1],
  \qquad
  W_t\stackrel{\mathrm{iid}}{\sim}\operatorname{Unif}[0,1].
\]
Let \(\varepsilon\) be independent of \(X^{(2)}\) and satisfy, for some
\(\sigma>0\),
\[
  P(\varepsilon=-\sigma)=\tau,
  \qquad
  P(\varepsilon=\sigma)=1-\tau.
\]
Suppose \(\beta_0\) is an interior point of \(\cB\), and generate
\[
  Y_t=X_t'\beta_0+\varepsilon,
  \qquad t=1,2.
\]
The true structure takes \(A_0=0\) and \(U_1=U_2=\varepsilon\), so zero is a
conditional \(\tau\)-quantile of each disturbance under model \textup{(A)}.

Fix a candidate \(b=\beta_0+d\), where \(d=(d_Z,d_W)'\). Conditional on a
regressor path \(X^{(2)}=(x_1,x_2)\), the support interval of the candidate
residual \(R_t(b)=Y_t-X_t'b\) is
\[
  [-\sigma-x_t'd,\,\sigma-x_t'd].
\]
Lemma \ref{lem:overlap} therefore requires
\begin{equation}
  \left|
    (Z_2-Z_1)d_Z+(W_2-W_1)d_W
  \right|
  \leq 2\sigma
  \qquad\text{almost surely}.
  \label{eq:Tuniformoverlap}
\end{equation}
The vector \((Z_2-Z_1,W_2-W_1)'\) has essential support \([-1,1]^2\).
Consequently,
\[
  \esssup
  \left|
    (Z_2-Z_1)d_Z+(W_2-W_1)d_W
  \right|
  =|d_Z|+|d_W|,
\]
and \eqref{eq:Tuniformoverlap} implies
\[
  |d_Z|+|d_W|\leq2\sigma.
\]

This condition is also sufficient. Let \(\iota=(1,1)'\), and, for any
\(d\) satisfying \(|d_Z|+|d_W|\leq2\sigma\), define
\[
  A^{(b)}=-\frac{1}{2}\iota'd,
  \qquad
  U_t^{(b)}
  =Y_t-X_t'b-A^{(b)}
  =\varepsilon+\left(\frac{1}{2}\iota-X_t\right)'d.
\]
For every \(x_t\in[0,1]^2\),
\[
  \left|
    \left(\frac{1}{2}\iota-x_t\right)'d
  \right|
  \leq\frac{1}{2}(|d_Z|+|d_W|)
  \leq\sigma.
\]
Thus the two support points of \(U_t^{(b)}\) lie weakly on opposite sides of
zero, with probability \(\tau\) on the lower point and probability
\(1-\tau\) on the upper point. Zero is therefore a conditional
\(\tau\)-quantile of \(U_t^{(b)}\) in both periods, including when one support
point equals zero under the paper's weak/strict quantile convention. This
constructs an admissible latent structure for every such candidate and proves
that
\begin{equation}
  \ThetaI^{(2)}
  =
  \cB\cap
  \left\{
    \beta_0+d:|d_Z|+|d_W|\leq2\sigma
  \right\}.
  \label{eq:Tuniformsharpdiamond}
\end{equation}
Because \(\beta_0\) is an interior point of \(\cB\) and \(\sigma>0\), this is
a nondegenerate two-dimensional set: both slope coefficients can vary locally
while remaining observationally equivalent to \(\beta_0\).

Moreover,
\[
  G_2
  =E[(X_2-X_1)(X_2-X_1)']
  =\frac{1}{6}I_2
  \succ0.
\]
The failure of point identification is therefore not caused by a
time-invariant direction or a deficient aggregate rank condition. It arises
because the regressor path has bounded directional spread. In particular,
\[
  \esssup\omega_2(d;X^{(2)})=|d_Z|+|d_W|<\infty
\]
for every fixed \(d\). \hfill $\square$
\end{example}
\ \ 

\begin{remark}[Identical sharp sets]
\label{rem:Tbinaryuniform}
With the same \(\sigma\) and \(\cB\), the preceding examples produce exactly
the same sharp set. To see why, let
\(d=(d_1,d_2)'\). In the binary design, the support of each \(X_t\) is
\(\{0,1\}^2\); in the uniform design, its support is \([0,1]^2\). These two
supports have the same convex hull and the same directional width. Consequently,
in both designs,
\[
  \esssup\left|(X_2-X_1)'d\right|
  =
  \sup_{v\in[-1,1]^2}|v'd|
  =|d_1|+|d_2|.
\]
For the uniform design, the extreme differences are approached rather than
attained with positive probability, which is why the relevant object is the
essential supremum. Because the common-midpoint construction in each example
also establishes sufficiency, both designs have the sharp set
\[
  \ThetaI^{(2)}
  =
  \cB\cap
  \left\{
    \beta_0+d:|d_1|+|d_2|\leq2\sigma
  \right\}.
\]
If \(\cB\) contains the entire diamond, its area is \(8\sigma^2\); otherwise,
the same diamond is truncated by the same parameter space in both examples.

The equality of the sharp sets is informative. The two designs have different
second moments,
\[
  G_2=\frac{1}{2}I_2
  \quad\text{in the binary design},
  \qquad
  G_2=\frac{1}{6}I_2
  \quad\text{in the uniform design},
\]
but identical directional support widths. {\it Thus neither continuity of the
regressors nor the magnitude of their within-path variance determines the size
of the sharp set in these examples. What matters is the directional extent of
the support.} The uniform example therefore shows that filling in all the
interior points of a bounded support need not sharpen identification. This is
precisely what distinguishes both examples from the Gaussian design, whose
directional support is unbounded in every nonzero direction and hence yields
point identification.
\end{remark}

\ \ 

\begin{example}[{\bf Mixed Point and Set Identified Case}]
\label{ex:Tcontinuousnonpoint}
Non-point identification can arise even when both regressors have continuous
support. Let \(T=2\), \(k=2\), and write
\[
  X_t=(Z_t,W_t)',
\]
where
\[
  Z_t\stackrel{\mathrm{iid}}{\sim}N(0,1),
  \qquad
  W_t\stackrel{\mathrm{iid}}{\sim}\operatorname{Unif}[0,1],
\]
with all variables mutually independent. Let \(\varepsilon\) be independent
of \(X^{(2)}\) and satisfy, for some \(\sigma>0\),
\[
  P(\varepsilon=-\sigma)=\tau,
  \qquad
  P(\varepsilon=\sigma)=1-\tau.
\]
Suppose \(\beta_0\) is an interior point of \(\cB\), and generate
\[
  Y_t=X_t'\beta_0+\varepsilon,
  \qquad t=1,2.
\]
The true structure takes \(A_0=0\) and \(U_1=U_2=\varepsilon\), so zero is a
conditional \(\tau\)-quantile of each disturbance under model \textup{(A)}.

Fix a candidate \(b=\beta_0+d\), where \(d=(d_Z,d_W)'\). Conditional on a
regressor path \(X^{(2)}=(x_1,x_2)\), the support interval of the candidate
residual \(R_t(b)=Y_t-X_t'b\) is
\[
  [-\sigma-x_t'd,\,\sigma-x_t'd].
\]
Lemma \ref{lem:overlap} therefore requires
\begin{equation}
  \left|
    (Z_2-Z_1)d_Z+(W_2-W_1)d_W
  \right|
  \leq 2\sigma
  \qquad\text{almost surely}.
  \label{eq:Tcontinuousoverlap}
\end{equation}
Because \(Z_2-Z_1\sim N(0,2)\) has unbounded support and is independent of
\(W_2-W_1\), condition \eqref{eq:Tcontinuousoverlap} is impossible when
\(d_Z\neq0\). Hence feasibility requires \(d_Z=0\). Once \(d_Z=0\),
\[
  \esssup |(W_2-W_1)d_W|=|d_W|,
\]
so \eqref{eq:Tcontinuousoverlap} further requires
\[
  |d_W|\leq2\sigma.
\]

This condition is also sufficient. For \(b=\beta_0+(0,c)'\), with
\(|c|\leq2\sigma\), define
\[
  A^{(b)}=-\frac{c}{2},
  \qquad
  U_t^{(b)}
  =Y_t-X_t'b-A^{(b)}
  =\varepsilon+\left(\frac12-W_t\right)c.
\]
Because
\[
  \left|\left(\frac12-W_t\right)c\right|
  \leq\frac{|c|}{2}
  \leq\sigma,
\]
the two support points of \(U_t^{(b)}\) straddle zero, with probability
\(\tau\) on the lower point and probability \(1-\tau\) on the upper point.
Thus zero is a conditional \(\tau\)-quantile of \(U_t^{(b)}\) in both
periods. This constructs an admissible latent structure for every such
candidate and proves that
\begin{equation}
  \ThetaI^{(2)}
  =
  \cB\cap
  \left\{
    \beta_0+(0,c)':|c|\leq2\sigma
  \right\}.
  \label{eq:Tcontinuoussharpsegment}
\end{equation}
This means that the coefficient on $Z$ is point identified while the one on $W$ is not. 
Because \(\beta_0\) is an interior point of \(\cB\) and \(\sigma>0\), the
identified set contains slopes other than \(\beta_0\).

Moreover,
\[
  G_2
  =
  E[(X_2-X_1)(X_2-X_1)']
  =
  \begin{pmatrix}
    2 & 0\\
    0 & 1/6
  \end{pmatrix}
  \succ0.
\]
The failure of point identification is therefore not caused by a
time-invariant direction or a deficient aggregate rank condition. The
coefficient on the normally distributed regressor is point identified, while
the coefficient on the bounded continuously distributed regressor remains set
identified. Thus continuous regressors and positive-definite covariance are
not sufficient for point identification. \hfill $\square$
\end{example}

\begin{example}[{\bf Gaussian regressors and point identification}]
\label{ex:Tgaussianpoint}
Let \(T=2\), \(k=2\), and write \(X_t=(Z_t,W_t)'\), where
\[
  Z_1,Z_2,W_1,W_2\stackrel{\mathrm{iid}}{\sim}N(0,1).
\]
Let \(\varepsilon\) be independent of \(X^{(2)}\) and satisfy, for some
\(\sigma>0\),
\[
  P(\varepsilon=-\sigma)=\tau,
  \qquad
  P(\varepsilon=\sigma)=1-\tau.
\]
For \(\beta_0\in\cB\), generate
\[
  Y_t=X_t'\beta_0+\varepsilon,
  \qquad t=1,2.
\]
This satisfies model \textup{(A)} with \(A_0=0\) and
\(U_1=U_2=\varepsilon\).

Fix \(b=\beta_0+d\). Conditional on \(X^{(2)}\), the candidate residual
\(R_t(b)=Y_t-X_t'b\) has support interval
\[
  [-\sigma-X_t'd,\,\sigma-X_t'd].
\]
Consequently, Lemma \ref{lem:overlap} requires
\begin{equation}
  |(X_2-X_1)'d|\leq2\sigma
  \qquad\text{almost surely}.
  \label{eq:Tgaussianoverlap}
\end{equation}
If \(d\neq0\), however,
\[
  (X_2-X_1)'d\sim N(0,2\|d\|^2),
\]
which has unbounded support. Hence \eqref{eq:Tgaussianoverlap} cannot hold
for any \(d\neq0\). Since \(\beta_0\) is feasible under the displayed
data-generating process,
\[
  \ThetaI^{(2)}=\{\beta_0\}.
\]

Thus the contrast with Examples \ref{ex:Tbinarynonpoint} and
\ref{ex:Tcontinuousnonpoint} is not continuous versus discrete regressors.
In the preceding continuous example, the uniform regressor leaves one bounded
coefficient direction. Here every nonzero coefficient direction has
unbounded within-panel support, so no nonzero slope perturbation can be
accommodated within the finite residual width.

The two-point distribution for \(\varepsilon\) is used only to make the
example explicit. The same argument applies whenever the true composite
residuals admit a common conditional envelope of finite width \(C_2\), provided
\(X_2-X_1\) is Gaussian with positive-definite covariance. \hfill $\square$
\end{example}

\begin{remark}[The fixed-tail Gaussian bound]
\label{rem:Tgaussianfiniteeta}
For the design in Example \ref{ex:Tgaussianpoint},
\[
  q_{1-\eta}\!\left(|(X_2-X_1)'d|\right)
  =
  \sqrt{2}\,\|d\|\Phi^{-1}(1-\eta/2).
\]
Proposition \ref{prop:Tquantilebound} therefore implies, for every fixed
\(\eta\in(0,1)\),
\begin{equation}
  \ThetaI^{(2)}
  \subseteq
  \left\{
    b\in\cB:
    \|b-\beta_0\|
    \leq
    \frac{\sqrt{2}\sigma}{\Phi^{-1}(1-\eta/2)}
  \right\}.
  \label{eq:gaussianfinitequantilebound}
\end{equation}
This ball is not a second characterization of the identified set. It is a
weaker outer bound that summarizes the identifying variation below a fixed
tail probability. The inclusions hold simultaneously for every
\(\eta\in(0,1)\), and their radii converge to zero as \(\eta\downarrow0\).
Their intersection is therefore the singleton \(\{\beta_0\}\), in agreement
with the exact conclusion of Example \ref{ex:Tgaussianpoint}.
\end{remark}

\begin{remark}[What the Gaussian result does and does not use]
Gaussianity is not essential.  For scalar \(X_t\), any design for which
\(|X_2-X_1|\) is essentially unbounded gives
\(\overline\kappa_2=\infty\).  In \(k\) dimensions, it is enough that
\((X_2-X_1)'d\) be essentially unbounded for every nonzero \(d\).  Conversely,
unbounded regressor support alone is not sufficient: the finite-width composite
residual envelope is also essential to this argument.  If \(C_2=\infty\), the
corollary gives no conclusion.  If \(\Omega\) is singular, every direction in
its null space has zero time variation and is subject to the absorption
nonidentification described above.
\end{remark}

\begin{remark}[Why the old pathwise modulus misses the Gaussian result]
When \(k=1\) and \(T=2\),
\[
  \kappa_2^{\mathrm{pw}}(X^{(2)})
  =
  |X_2-X_1|.
\]
Under the jointly Gaussian design of the Example \ref{ex:Tgaussianpoint}
\[
  \underline\kappa_2^{\mathrm{pw}}
  =
  \essinf|X_2-X_1|
  =
  0,
\]
because the two regressors can be arbitrarily close with positive probability.
When \(k>1\) and \(T=2\), the pathwise modulus is identically zero because for
each realized difference \(X_2-X_1\) there is a unit direction orthogonal to it.
Neither fact prevents identification.  Candidate feasibility must hold for
almost every regressor path for one \emph{fixed} coefficient difference
\(b-\beta_0\).  That logical order is captured by
\(\overline\kappa_T\), which fixes the direction before taking the essential
supremum across paths.
\end{remark}

\begin{corollary}[Explicit cross-path contraction rates]
\label{cor:Trates}
Under the assumptions of Theorem \ref{thm:Tcontraction}, suppose that, for all
sufficiently large \(T\),
\[
  C_T\leq \overline C T^\zeta,
  \qquad
  \overline\kappa_T\geq \underline c T^\alpha
\]
for constants \(\overline C<\infty\), \(\underline c>0\), and
\(\alpha>\zeta\).  Then
\[
  d_H\!\left(\Theta_I^{(T)},\{\beta_0\}\right)
  \leq
  \frac{\overline C}{\underline c}
  T^{-(\alpha-\zeta)}
\]
and
\[
  \operatorname{diam}\!\left(\Theta_I^{(T)}\right)
  \leq
  2\frac{\overline C}{\underline c}
  T^{-(\alpha-\zeta)}.
\]
If \(C_T\leq\overline C\), any divergence
\(\overline\kappa_T\to\infty\) gives the rate
\(O(\overline\kappa_T^{-1})\).
\end{corollary}

\begin{proof}
Substitute the two displayed bounds into
\eqref{eq:Tcontractionbound}.  The diameter inequality follows from the
triangle inequality around \(\beta_0\).
\end{proof}

\begin{remark}[Why a substantive support condition is needed]
Model \textup{(A)} imposes no stationarity, independence, exchangeability, or
common distribution across \(U_t\).  Consequently, increasing \(T\), or giving
the regressors unbounded support, does not by itself create identification.  The
support-envelope condition prevents the composite residual distribution from
expanding enough to absorb arbitrarily large false index shifts.  More generally,
point identification may follow from the exact dual-separation condition in
Proposition \ref{prop:Tdualseparation} even when the bounded-envelope argument
does not apply.  Elimination of incidental-parameter bias for a particular
large-\(T\) estimator remains distinct from singleton identification under model
\textup{(A)} alone.
\end{remark}

% The quantile-varying factor-loading extension follows the complete
% identification analysis of the baseline model.
\section{A quantile-varying factor loading}
\label{sec:loadingduality}

This section changes the maintained cross-quantile structure rather than merely
applying the baseline analysis at two quantile indices.  The baseline permits a
different latent anchor and latent completion at every quantile.  Here, by
contrast, the median and target-quantile restrictions must hold simultaneously
under one probability law for \((Y,X,A)\), with the same scalar \(A\).
Quantile-specific heterogeneity is restricted to the separable form
\(A_q=\rho(q)A\).

Fix
\[
  q_0=\frac12
  \qquad\text{and}\qquad
  q_1=\tau\neq q_0.
\]
At these two quantile indices, the joint maintained model is
\begin{equation}
  Q_{q_\ell}(Y_t\mid X,A)
  =
  X_t'\beta(q_\ell)+\rho(q_\ell)A,
  \qquad
  \ell\in\{0,1\},\quad t=1,\ldots,T.
  \label{eq:hlmodel}
\end{equation}
As elsewhere in the paper, equality in \eqref{eq:hlmodel} is interpreted in the
possibly set-valued sense: a number \(c\) is an admissible conditional
\(q\)-quantile if
\[
  P(Y_t\leq c\mid X,A)\geq q,
  \qquad
  P(Y_t<c\mid X,A)\leq q.
\]
No continuity or uniqueness of the conditional quantile is imposed.

The word ``joint'' is important: one coupling must satisfy all \(2T\)
restrictions in \eqref{eq:hlmodel}.  At a single quantile \(q\), the product
\(\rho(q)A\) could simply be relabeled as an unrestricted anchor \(A_q\), so
\(\rho(q)\) would have no separate identifying content.  It is the common-
\(A\) restriction together with the normalization \(\rho(1/2)=1\) that makes
\(\rho(\tau)\) a meaningful relative loading.

The mechanism is the same sorting argument as before, applied to a longer
list of residuals.  Dividing the \(\tau\)-quantile equation by \(\rho(\tau)>0\)
shows that the same number \(A\) must be a conditional median of the \(T\)
median residuals \(Y_t-X_t'\beta(1/2)\) and a conditional \(\tau\)-quantile of
the \(T\) rescaled residuals \((Y_t-X_t'\beta(\tau))/\rho(\tau)\).  A
candidate value of the loading changes the scale of the second family relative
to the first, and a wrong scale misaligns the two families across individuals
in a way that no common sorting can repair.  This is why cross-quantile
comparisons carry information about \(\rho(\tau)\) that neither quantile
restriction carries on its own, and why every tool of the baseline analysis
applies after the \(T\) residual coordinates are replaced by \(2T\).

\subsection{Normalizations, relative loadings, and the common-anchor
reparameterization}
\label{subsec:hlnormalization}

The representation in \eqref{eq:hlmodel} has a scale indeterminacy.  For every
nonzero constant \(c\),
\[
  A^\ast=cA,
  \qquad
  \rho^\ast(q)=\frac{\rho(q)}{c}
\]
generates the same conditional quantile indices.  We therefore impose the reference
normalization
\begin{equation}
  \rho(q_0)=\rho(1/2)=1.
  \label{eq:hlscalenormalization}
\end{equation}
This fixes both the scale and the orientation of \(A\).  The remaining coefficient
\(\rho(\tau)\) is a relative loading and is economically meaningful; it should not
be normalized to one.

The relevant location issue is not whether each regressor visibly changes over
time, but whether some linear combination of the regressor vector is constant.
To remove that ambiguity, we impose the same no-absorption rank condition as in
the baseline model:
\begin{equation}
  \mathcal D_T
  :=
  \left\{
    c\in\mathbb R^k:
    X_t'c=X_s'c\ \text{almost surely for every }s,t
  \right\}
  =
  \{0\}.
  \label{eq:hlnoabsorptivedirection}
\end{equation}
This slightly stronger rank formulation also excludes a time-invariant linear
combination of regressors whose individual coordinates happen to vary over time.
If a nonzero \(c\in\mathcal D_T\) existed, the common random variable
\(X_t'c\) could be absorbed into the fixed effect through
\[
  A^\ast=A+X_t'c,
  \qquad
  \beta^\ast(q)=\beta(q)-\rho(q)c,
\]
leaving every conditional quantile index unchanged.  Condition
\eqref{eq:hlnoabsorptivedirection} eliminates this invariance.

Consequently, no additional location normalization is required.  In particular,
one should not impose \(E[A]=0\), a median-zero restriction on \(A\), or an
arbitrary zero restriction on a component of \(\beta(1/2)\).  Shifting \(A\) by a
constant cannot be offset by an intercept because no intercept is included, and
the observable outcome levels determine which locations of the latent anchor are
feasible.

For the main results, impose the common-orientation restriction
\begin{equation}
  r:=\rho(\tau)>0.
  \label{eq:hlpositive}
\end{equation}
Positivity is a substantive shape restriction, not a consequence of
\eqref{eq:hlscalenormalization}.  In applications it is convenient to restrict
\(r\) to a compact interval
\([\underline r,\overline r]\subset(0,\infty)\), although the population
duality theorem below is candidate-specific and requires only \(r>0\).

\begin{proposition}[Absorption equivalence in the loading model]
\label{prop:hlabsorption}
Let \(c\in\mathcal D_T\), where \(\mathcal D_T\) is defined in
\eqref{eq:absorption-space}.  Then the joint candidates
\[
  (b_0,b_1,r)
  \quad\text{and}\quad
  (b_0-c,b_1-rc,r)
\]
are observationally equivalent whenever both belong to the parameter space.
Thus \(\mathcal D_T=\{0\}\), equivalently
\(\operatorname{rank}(G_T)=k\) under square integrability, is necessary to remove
the fixed-effect absorption ambiguity for the two coefficient vectors.  It is not
sufficient for point identification of \((b_0,b_1,r)\).
\end{proposition}

\begin{proof}
Because \(c\in\mathcal D_T\), the random variable \(X_t'c\) is common across
periods.  Put \(A^\ast=A+X_t'c\).  For \(q_0=1/2\), where
\(\rho(q_0)=1\),
\[
  X_t'(b_0-c)+A^\ast=X_t'b_0+A.
\]
At \(q_1=\tau\),
\[
  X_t'(b_1-rc)+rA^\ast=X_t'b_1+rA.
\]
The transformation of \(A\) is one-to-one conditional on \(X\), so the
conditioning sigma-fields and both generalized-quantile restrictions are
unchanged.  Replacing \(c\) by \(-c\) proves equivalence.
\end{proof}

Write
\[
  b_0=\beta(q_0),
  \qquad
  b_1=\beta(\tau),
\]
and define
\begin{equation}
  \gamma=\frac1r,
  \qquad
  d=\frac{b_1}{r}=\gamma b_1.
  \label{eq:hlprojectiveparameters}
\end{equation}
For a candidate
\(\theta=(b_0,b_1,r)\), or equivalently
\(\vartheta=(b_0,d,\gamma)\), define the two families of transformed residuals
\begin{align}
  Z_{0t}(\theta)
  &=
  Y_t-X_t'b_0,
  \label{eq:hlmedianresidual}\\
  Z_{1t}(\theta)
  &=
  \frac{Y_t-X_t'b_1}{r}
  =
  \gamma Y_t-X_t'd.
  \label{eq:hltauresidual}
\end{align}
Because \(r>0\), model \eqref{eq:hlmodel} at \(q_0\) and \(q_1\) is equivalent
to requiring that the \emph{same} latent number \(A\) be a conditional
\(q_\ell\)-quantile of every \(Z_{\ell t}(\theta)\):
\begin{align}
  P\{Z_{\ell t}(\theta)\leq A\mid X,A\}
  &\geq q_\ell,
  \label{eq:hlcoupleweak}\\
  P\{Z_{\ell t}(\theta)<A\mid X,A\}
  &\leq q_\ell,
  \label{eq:hlcouplestrict}
\end{align}
for \(\ell=0,1\) and \(t=1,\ldots,T\).  The transformation
\eqref{eq:hlprojectiveparameters} is useful computationally: both residual
families are affine in \((b_0,d,\gamma)\).

\begin{remark}[Why the median restriction must be imposed jointly]
\label{rem:hljointmedian}
At the single quantile \(\tau\), the loading \(r\) can be absorbed into a
redefinition of the unrestricted latent effect.  The normalization
\(\rho(1/2)=1\) has identifying content for \(r\) only when the median and
\(\tau\) restrictions are imposed on the same \(A\) in the same coupling.
Accordingly, the sharp set below is a joint median--\(\tau\) set.  Treating
\(b_0\) as a nuisance parameter and projecting it out remains sharp; discarding
the median restriction does not.
\end{remark}

\subsection{The joint sharp set and its target-quantile projection}
\label{subsec:hlsharpset}

Let \(\Pi(\Pobs)\) be the collection of Borel probability laws of
\((Y,X,A)\) having observed \((Y,X)\)-marginal \(\Pobs\).  For
\(\theta=(b_0,b_1,r)\), let
\(\mathcal C_\theta^{0,\tau}(\Pobs)\) contain every
\(\pi\in\Pi(\Pobs)\) satisfying
\eqref{eq:hlcoupleweak}--\eqref{eq:hlcouplestrict} under \(\pi\), simultaneously
for both \(\ell=0,1\) and all \(t\).  Define
\begin{equation}
  \Theta_I^{0,\tau}
  =
  \left\{
  (b_0,b_1,r):
  r>0,\ 
  \mathcal C_\theta^{0,\tau}(\Pobs)\neq\varnothing
  \right\}.
  \label{eq:hljointset}
\end{equation}
No component of \(b_0\) or \(b_1\) is normalized to zero.  Apart from the usual
parameter-space restrictions, both coefficient vectors are allowed to vary freely.

The target-quantile projected sharp set is
\begin{equation}
  \Theta_{I,\tau}^{\mathrm{pt}}
  =
  \left\{
  (b_1,r):
  \text{there exists }b_0\in\mathcal B_0
  \text{ such that }(b_0,b_1,r)\in\Theta_I^{0,\tau}
  \right\}.
  \label{eq:hlpointwiseprojection}
\end{equation}
Here the superscript ``pt'' refers only to the coordinates reported after
projection onto \((\beta(\tau),\rho(\tau))\).  Feasibility is not pointwise: it
is determined jointly from the restrictions at \(1/2\) and \(\tau\).
If \(b_0\) is known or separately restricted to a set
\(\mathcal B_0^I\), replace \(b_0\in\mathcal B_0\) in
\eqref{eq:hlpointwiseprojection} by \(b_0\in\mathcal B_0^I\).  This is a slice or
a restricted projection of the same joint sharp set; no plug-in approximation is
required.

\subsection{Automatic cross-quantile ordering}
\label{subsec:hlordering}

The joint formulation automatically imposes the noncrossing restriction that is
appropriate for a common latent effect.

\begin{lemma}[Order of generalized quantile selections]
\label{lem:hlquantileorder}
Let \(q<q'\).  If \(v\) satisfies
\[
  P(Y\leq v)\geq q,
  \qquad P(Y<v)\leq q,
\]
and \(v'\) satisfies the analogous two inequalities at \(q'\), then
\(v\leq v'\).
\end{lemma}

\begin{proof}
If \(v>v'\), then
\(\{Y\leq v'\}\subseteq\{Y<v\}\), and therefore
\[
  q'
  \leq P(Y\leq v')
  \leq P(Y<v)
  \leq q,
\]
a contradiction.
\end{proof}

Set
\[
  s_\tau=\operatorname{sgn}(\tau-\tfrac12).
\]
Lemma \ref{lem:hlquantileorder} implies that every feasible coupling satisfies
\begin{equation}
  s_\tau
  \left[
    X_t'(b_1-b_0)+(r-1)A
  \right]
  \geq0,
  \qquad t=1,\ldots,T,
  \label{eq:hlautomaticordering}
\end{equation}
almost surely.  This condition is not an additional maintained assumption; it is
an implication of imposing the two quantile restrictions on one conditional
distribution.
\subsection{Observable cross-quantile crossing inequalities}
\label{subsec:hlcrossing}

The common-anchor representation also produces a simple observable outer bound
analogous to Proposition \ref{prop:crossing}.  Unlike the within-quantile crossing
inequalities, its cross-quantile members can restrict the relative loading
\(r=\rho(\tau)\).

For \(\ell,m\in\{0,1\}\), define the directional crossing constant
\begin{equation}
  c_{\ell m}
  =
  \max\{q_\ell,1-q_m\}.
  \label{eq:hlcrossingconstant}
\end{equation}

\begin{proposition}[Generalized crossing inequalities]
\label{prop:hlcrossing}
If \(\theta=(b_0,b_1,r)\in\Theta_I^{0,\tau}\), then, for every
\(\ell,m\in\{0,1\}\), every \(s,t\in\{1,\ldots,T\}\), every
\(v\in\mathbb R\), and \(P_X\)-almost every \(x\),
\begin{equation}
  P^{\mathrm{obs}}
  \left\{
    Z_{\ell s}(\theta)\leq v<Z_{mt}(\theta)
    \,\middle|\,X=x
  \right\}
  \leq
  c_{\ell m}.
  \label{eq:hlcrossing}
\end{equation}
When \(\ell=m\), the restriction is nontrivial only for \(s\neq t\).
When \(\ell\neq m\), it also applies with \(s=t\).
\end{proposition}

\begin{proof}
Fix a feasible coupling
\(\pi\in\mathcal C_\theta^{0,\tau}(\Pobs)\), and let
\[
  E_{\ell s,mt}(v;\theta)
  =
  \{Z_{\ell s}(\theta)\leq v<Z_{mt}(\theta)\}.
\]
Conditional on \((X,A)\), split according to the position of \(A\).  If
\(A\leq v\), then
\[
  E_{\ell s,mt}(v;\theta)
  \subseteq
  \{Z_{mt}(\theta)>v\}
  \subseteq
  \{Z_{mt}(\theta)>A\}.
\]
The weak \(q_m\)-quantile restriction in
\eqref{eq:hlcoupleweak} therefore gives
\[
  \pi\{E_{\ell s,mt}(v;\theta)\mid X,A\}
  \leq 1-q_m
  \qquad\text{on }\{A\leq v\}.
\]
If \(A>v\), then
\[
  E_{\ell s,mt}(v;\theta)
  \subseteq
  \{Z_{\ell s}(\theta)\leq v\}
  \subseteq
  \{Z_{\ell s}(\theta)<A\},
\]
so the strict \(q_\ell\)-quantile restriction in
\eqref{eq:hlcouplestrict} gives
\[
  \pi\{E_{\ell s,mt}(v;\theta)\mid X,A\}
  \leq q_\ell
  \qquad\text{on }\{A>v\}.
\]
Combining the two cases,
\begin{align}
  \pi\{E_{\ell s,mt}(v;\theta)\mid X,A\}
  \leq
  &(1-q_m)\one\{A\leq v\}
  \notag\\
  &+
  q_\ell\one\{A>v\}.
  \label{eq:hlconditionalcrossing}
\end{align}
Taking expectations conditional on \(X=x\) yields
\begin{align*}
  &P^{\mathrm{obs}}\{E_{\ell s,mt}(v;\theta)\mid X=x\}\\
  &\quad\leq
  (1-q_m)P_\pi(A\leq v\mid X=x)
  +
  q_\ell P_\pi(A>v\mid X=x)\\
  &\quad\leq
  \max\{q_\ell,1-q_m\},
\end{align*}
which proves \eqref{eq:hlcrossing}.
\end{proof}

Because \(q_0=1/2\) and \(q_1=\tau\), Proposition
\ref{prop:hlcrossing} contains four families:
\begin{align}
  P^{\mathrm{obs}}\{Z_{0s}\leq v<Z_{0t}\mid X=x\}
  &\leq \frac12,
  \label{eq:hlcross00}\\
  P^{\mathrm{obs}}\{Z_{1s}\leq v<Z_{1t}\mid X=x\}
  &\leq c_\tau:=\max\{\tau,1-\tau\},
  \label{eq:hlcross11}\\
  P^{\mathrm{obs}}\{Z_{0s}\leq v<Z_{1t}\mid X=x\}
  &\leq \max\left\{\frac12,1-\tau\right\},
  \label{eq:hlcross01}\\
  P^{\mathrm{obs}}\{Z_{1s}\leq v<Z_{0t}\mid X=x\}
  &\leq \max\left\{\tau,\frac12\right\}.
  \label{eq:hlcross10}
\end{align}
Here and below the dependence of the transformed residuals on \(\theta\) is
suppressed to simplify notation.  If \(Z_L\) denotes the transformed residual
associated with the lower of \(1/2\) and \(\tau\), and \(Z_H\) the one associated
with the higher quantile index, the two cross-quantile bounds can be written more
compactly as
\begin{align}
  P^{\mathrm{obs}}\{Z_{Ls}\leq v<Z_{Ht}\mid X=x\}
  &\leq\frac12,
  \label{eq:hlcrosslowhigh}\\
  P^{\mathrm{obs}}\{Z_{Hs}\leq v<Z_{Lt}\mid X=x\}
  &\leq c_\tau.
  \label{eq:hlcrosshighlow}
\end{align}

\begin{corollary}[Observable cross-quantile moments]
\label{cor:hlcrossingmoments}
Every \(\theta\in\Theta_I^{0,\tau}\) satisfies
\begin{equation}
  \E_{\mathrm{obs}}
  \left[
    h(X)
    \left\{
      \one\bigl(
        Z_{\ell s}(\theta)\leq v<Z_{mt}(\theta)
      \bigr)
      -
      c_{\ell m}
    \right\}
  \right]
  \leq0
  \label{eq:hlcrossingmoment}
\end{equation}
for every bounded nonnegative Borel function \(h\), every
\(\ell,m\in\{0,1\}\), every \(s,t\), and every \(v\in\mathbb R\).
\end{corollary}

\begin{proof}
Multiply \eqref{eq:hlcrossing} by \(h(X)\geq0\) and apply iterated
expectations.
\end{proof}

\paragraph{Identifying content for the relative loading.}
The within-\(\tau\) family \eqref{eq:hlcross11} does not restrict
\(r=\rho(\tau)\).  Indeed, because \(r>0\), the change of threshold
\(u=rv\) gives
\[
\begin{aligned}
  &\sup_{v\in\mathbb R}
  P^{\mathrm{obs}}
  \left\{
    \frac{Y_s-X_s'b_1}{r}
    \leq v<
    \frac{Y_t-X_t'b_1}{r}
    \,\middle|\,X=x
  \right\}\\
  &\qquad=
  \sup_{u\in\mathbb R}
  P^{\mathrm{obs}}
  \left\{
    Y_s-X_s'b_1
    \leq u<
    Y_t-X_t'b_1
    \,\middle|\,X=x
  \right\}.
\end{aligned}
\]
By contrast, the cross-quantile inequalities are
\begin{align}
  P^{\mathrm{obs}}
  \left\{
    Y_s-X_s'b_0\leq v,\
    Y_t-X_t'b_1>rv
    \,\middle|\,X=x
  \right\}
  &\leq
  \max\left\{\frac12,1-\tau\right\},
  \label{eq:hlcross01original}\\
  P^{\mathrm{obs}}
  \left\{
    Y_s-X_s'b_1\leq rv,\
    Y_t-X_t'b_0>v
    \,\middle|\,X=x
  \right\}
  &\leq
  \max\left\{\tau,\frac12\right\}.
  \label{eq:hlcross10original}
\end{align}
The relative scale \(r\) cannot be removed from these two events by relabeling a
single threshold.  They therefore provide directly observable outer restrictions
on \(\rho(\tau)\).

Define the joint crossing outer set
\begin{equation}
  \Theta_{I,\mathrm{cross}}^{0,\tau}
  =
  \left\{
    \theta:
    \sup_{v\in\mathbb R}
    P^{\mathrm{obs}}
    \left\{
      Z_{\ell s}(\theta)\leq v<Z_{mt}(\theta)
      \,\middle|\,X=x
    \right\}
    \leq c_{\ell m}
    \text{ for all }\ell,m,s,t
    \text{ and a.e. }x
  \right\}.
  \label{eq:hlcrossingouterset}
\end{equation}
Then
\begin{equation}
  \Theta_I^{0,\tau}
  \subseteq
  \Theta_{I,\mathrm{cross}}^{0,\tau},
  \qquad
  \Theta_{I,\tau}^{\mathrm{pt}}
  \subseteq
  \left\{
    (b_1,r):
    \text{some }b_0\in\mathcal B_0
    \text{ satisfies }
    (b_0,b_1,r)\in
    \Theta_{I,\mathrm{cross}}^{0,\tau}
  \right\}.
  \label{eq:hlcrossingcontainment}
\end{equation}
These containments are generally strict: as in the single-quantile model, crossing
inequalities use only particular consequences of common-anchor compatibility and
do not replace the full coupling or dual characterization.

\paragraph{Computation.}
With finite conditional support, for a fixed candidate it is enough to check
thresholds among the distinct combined values
\[
  \left\{
    z_{j\ell t}(\theta):
    j=1,\ldots,J,\ \ell=0,1,\ t=1,\ldots,T
  \right\}.
\]
Thus the entire crossing screen requires only probability summation and no
optimization over latent masses.  With continuously distributed outcomes or
regressors, \eqref{eq:hlcrossingmoment} and a chosen nonnegative instrument class
give the corresponding observable outer restrictions.

\subsection{Sharp joint coupling and pointwise projection}
\label{subsec:hlcouplingsharp}

\begin{proposition}[Exact joint coupling characterization]
\label{prop:hlcouplingsharp}
The set in \eqref{eq:hljointset} is the sharp identified set for
\((\beta(1/2),\beta(\tau),\rho(\tau))\) under the two restrictions in
\eqref{eq:hlmodel}, the scale normalization
\eqref{eq:hlscalenormalization}, and the common-orientation restriction
\eqref{eq:hlpositive}.
\end{proposition}

\begin{proof}
If a structural law satisfies \eqref{eq:hlmodel}, its marginal law of
\((Y,X,A)\) belongs to \(\Pi(\Pobs)\).  At the median,
\[
  Y_t\leq X_t'b_0+A
  \quad\Longleftrightarrow\quad
  Z_{0t}(\theta)\leq A,
\]
and similarly for the strict event.  At \(\tau\), positivity of \(r\) gives
\[
  Y_t\leq X_t'b_1+rA
  \quad\Longleftrightarrow\quad
  Z_{1t}(\theta)\leq A,
\]
again with the same equivalence for strict inequalities.  Thus every admissible
structural law induces an element of
\(\mathcal C_\theta^{0,\tau}(\Pobs)\).

Conversely, take
\(\pi\in\mathcal C_\theta^{0,\tau}(\Pobs)\).  Under the conditional law of
\(Y\) given \((X,A)\) generated by \(\pi\), the equivalences just displayed turn
\eqref{eq:hlcoupleweak}--\eqref{eq:hlcouplestrict} into the generalized-quantile
conditions in \eqref{eq:hlmodel}, at both quantile indices and every period.
The observed marginal is \(\Pobs\) by construction.  Hence every candidate in
\eqref{eq:hljointset} is attained by a complete latent-variable structure, while
the first argument shows that no candidate outside the set can be attained.
\end{proof}

\subsection{Population sharp observable duality}
\label{subsec:hlpopulationdual}

Write \(W=(Y,X)\), and suppose its support \(\mathcal W\) is compact.  Let
\(\mathcal X\) be the projection of \(\mathcal W\) onto the support of \(X\).
For a fixed candidate \(\theta\), define
\begin{align}
  \underline z_\theta
  &=
  \min_{\substack{w\in\mathcal W\\
                  \ell\in\{0,1\},\,t\leq T}}
  Z_{\ell t}(\theta;w),
  \notag\\
  \overline z_\theta
  &=
  \max_{\substack{w\in\mathcal W\\
                  \ell\in\{0,1\},\,t\leq T}}
  Z_{\ell t}(\theta;w),
  \qquad
  K_\theta=[\underline z_\theta,\overline z_\theta].
  \label{eq:hlanchorinterval}
\end{align}

\begin{lemma}[Exact compactification of the common anchor]
\label{lem:hlanchorcompact}
If \(\mathcal C_\theta^{0,\tau}(\Pobs)\neq\varnothing\), every feasible coupling
satisfies \(P(A\in K_\theta)=1\).
\end{lemma}

\begin{proof}
If \(A<\underline z_\theta\), then
\[
  P\{Z_{\ell t}(\theta)\leq A\mid X,A\}=0<q_\ell
\]
for every \(\ell,t\).  If \(A>\overline z_\theta\), then
\[
  P\{Z_{\ell t}(\theta)<A\mid X,A\}=1>q_\ell
\]
for every \(\ell,t\).  Either event is inconsistent with feasibility.
\end{proof}

A dual multiplier is an array
\[
  \lambda
  =
  \{\lambda_{\ell t}^-,\lambda_{\ell t}^+:
  \ell=0,1,\ t=1,\ldots,T\}
  \in C_+(\mathcal X\times K_\theta)^{4T}.
\]
For \(w=(y,x)\), \(a\in K_\theta\), and candidate \(\theta\), define
\begin{align}
  S_\theta(w,a;\lambda)
  =
  \sum_{\ell=0}^1\sum_{t=1}^T
  &\lambda_{\ell t}^-(x,a)
  \left\{
    \one\bigl(Z_{\ell t}(\theta;w)\leq a\bigr)-q_\ell
  \right\}
  \notag\\
  &+
  \lambda_{\ell t}^+(x,a)
  \left\{
    q_\ell-\one\bigl(Z_{\ell t}(\theta;w)<a\bigr)
  \right\}.
  \label{eq:hldualscore}
\end{align}
Its observable dual functional is
\begin{equation}
  \Psi_\theta^{0,\tau}(\lambda)
  =
  \E_{\mathrm{obs}}
  \left[
    \max_{a\in K_\theta}S_\theta(W,a;\lambda)
  \right].
  \label{eq:hldualfunctional}
\end{equation}
The latent common anchor has disappeared from
\eqref{eq:hldualfunctional}; for fixed \(\theta\) and \(\lambda\), the functional
depends only on the observed law.

\begin{theorem}[Sharp observable duality for a quantile-varying loading]
\label{thm:hlobservabledual}
Suppose \(\mathcal W\) is compact and \(r>0\).  Then
\begin{equation}
  \theta\in\Theta_I^{0,\tau}
  \quad\Longleftrightarrow\quad
  \Psi_\theta^{0,\tau}(\lambda)\geq0
  \quad\text{for every }
  \lambda\in C_+(\mathcal X\times K_\theta)^{4T}.
  \label{eq:hlsharpdualcriterion}
\end{equation}
If \(\theta\notin\Theta_I^{0,\tau}\), there is a nonnegative continuous
multiplier \(\lambda\) such that
\begin{equation}
  \Psi_\theta^{0,\tau}(\lambda)<0.
  \label{eq:hlstrictcertificate}
\end{equation}
Consequently, the pointwise sharp set
\(\Theta_{I,\tau}^{\mathrm{pt}}\) is obtained by applying
\eqref{eq:hlsharpdualcriterion} to the joint candidates and then projecting as in
\eqref{eq:hlpointwiseprojection}.
\end{theorem}

\begin{proof}
Index the \(2T\) transformed residual coordinates by \(i=(\ell,t)\), and attach
quantile constant \(q_i=q_\ell\) to coordinate \(i\).  For a coupling \(\pi\), form
the \(4T\) signed measures obtained by integrating
\[
  \one\{Z_{\ell t}(\theta)\leq A\}-q_\ell,
  \qquad
  q_\ell-\one\{Z_{\ell t}(\theta)<A\}
\]
over Borel subsets of \((X,A)\).  The coupling is feasible if and only if all these
measures are nonnegative.  This is exactly the signed-measure formulation used in
the proof of Theorem \ref{thm:observabledual}; that proof does not require the
constants attached to different coordinates to be equal.

Apply the same normalized-cone and Sion minimax argument on
\(\mathcal W\times K_\theta\).  Maximization over couplings with observed marginal
\(\Pobs\) is attained by selecting, observation by observation, an anchor maximizing
\(S_\theta(W,a;\lambda)\).  The resulting value is
\(\Psi_\theta^{0,\tau}(\lambda)\).  Feasibility is therefore equivalent to its
nonnegativity for every multiplier.  Compactness and upper semicontinuity give a
strictly negative minimizing multiplier whenever no feasible coupling exists.
\end{proof}

\paragraph{An exact order-reduced dual.}
The automatic ordering condition can reduce the scalar maximization in
\eqref{eq:hldualfunctional}.  Define the observable, candidate-specific
correspondence
\begin{equation}
  \mathcal K_\theta^{\mathrm{nc}}(x)
  =
  \left\{
  a\in K_\theta:
  s_\tau
  \left[x_t'(b_1-b_0)+(r-1)a\right]\geq0
  \text{ for every }t
  \right\}.
  \label{eq:hlncanchor}
\end{equation}
If this set is empty on a set of regressor paths having positive probability,
\(\theta\) is infeasible.  If it is nonempty on the support of \(X\), define
\begin{equation}
  \Psi_\theta^{\mathrm{nc}}(\lambda)
  =
  \E_{\mathrm{obs}}
  \left[
    \max_{a\in\mathcal K_\theta^{\mathrm{nc}}(X)}
    S_\theta(W,a;\lambda)
  \right].
  \label{eq:hlncdual}
\end{equation}
Then
\begin{equation}
  \theta\in\Theta_I^{0,\tau}
  \quad\Longleftrightarrow\quad
  \Psi_\theta^{\mathrm{nc}}(\lambda)\geq0
  \quad\text{for every }\lambda\geq0.
  \label{eq:hlnccriterion}
\end{equation}
Indeed, Lemma \ref{lem:hlquantileorder} shows that restricting couplings to the
closed graph of \(\mathcal K_\theta^{\mathrm{nc}}\) removes no feasible coupling.
Repeating the proof of Theorem \ref{thm:hlobservabledual} on that compact graph
gives \eqref{eq:hlnccriterion}.  Thus the order reduction is exact, rather than an
outer approximation.

\subsection{Exact finite-support primal and dual programs}
\label{subsec:hlfinitedual}

Fix a positive-probability regressor path \(X=x\), and suppose
\begin{equation}
  P^{\mathrm{obs}}(Y=y^j\mid X=x)=p_j>0,
  \qquad
  j=1,\ldots,J.
  \label{eq:hlfinitesupport}
\end{equation}
For a candidate \(\theta\), set
\begin{align}
  z_{j0t}(\theta)
  &=
  y_t^j-x_t'b_0,
  \notag\\
  z_{j1t}(\theta)
  &=
  \gamma y_t^j-x_t'd.
  \label{eq:hlfiniteresiduals}
\end{align}
Let \(v_1<\cdots<v_P\) be the distinct values among the \(2TJ\) numbers in
\eqref{eq:hlfiniteresiduals}; hence \(P\leq2TJ\).  Introduce
\[
  q_{jp}
  =
  P^\ast(Y=y^j,A=v_p\mid X=x).
\]
Consider the linear feasibility problem
\begin{align}
  q_{jp}&\geq0,
  &&j\leq J,\quad p\leq P,
  \label{eq:hlprimalnonnegative}\\
  \sum_{p=1}^Pq_{jp}&=p_j,
  &&j\leq J,
  \label{eq:hlprimalmarginal}\\
  \sum_{j=1}^Jq_{jp}
  \left\{
    \one(z_{j\ell t}(\theta)\leq v_p)-q_\ell
  \right\}&\geq0,
  &&\ell=0,1,\quad t\leq T,\quad p\leq P,
  \label{eq:hlprimalweak}\\
  \sum_{j=1}^Jq_{jp}
  \left\{
    q_\ell-\one(z_{j\ell t}(\theta)<v_p)
  \right\}&\geq0,
  &&\ell=0,1,\quad t\leq T,\quad p\leq P.
  \label{eq:hlprimalstrict}
\end{align}

\begin{theorem}[Exact finite-support joint-quantile LP]
\label{thm:hlprimal}
For the path \(X=x\), a joint median--\(\tau\) coupling exists for candidate
\(\theta\) if and only if
\eqref{eq:hlprimalnonnegative}--\eqref{eq:hlprimalstrict} is feasible.
Thus, when \(X\) has finite support, \(\theta\in\Theta_I^{0,\tau}\) if and only
if the LP is feasible in every positive-probability \(X\)-cell.
\end{theorem}

\begin{proof}
LP feasibility directly constructs the joint law of the observed support point and
the common anchor.  Conditional on an anchor having positive marginal mass,
division of
\eqref{eq:hlprimalweak}--\eqref{eq:hlprimalstrict} by that mass gives exactly
\eqref{eq:hlcoupleweak}--\eqref{eq:hlcouplestrict}.

Conversely, begin with any feasible coupling whose observed conditional support is
\eqref{eq:hlfinitesupport}.  Lemma \ref{lem:hlanchorcompact} excludes anchor mass
below \(v_1\) or above \(v_P\).  Keep mass already located at a residual value.
For \(a\in(v_p,v_{p+1})\), move its joint mass, support point by support point,
to \(v_p\).  For every transformed residual value \(z\),
\[
  \one(z\leq a)=\one(z\leq v_p),
  \qquad
  \one(z<v_p)\leq\one(z<a).
\]
Hence all weak lower-probability restrictions are unchanged, while every strict
upper-probability restriction is weakly relaxed.  Applying the move to every gap
produces masses satisfying
\eqref{eq:hlprimalnonnegative}--\eqref{eq:hlprimalstrict}.  This is an exact,
candidate-specific anchor reduction; it is not a numerical discretization of
\(A\).
\end{proof}

For nonnegative numbers
\(\lambda_{\ell pt}^-\) and \(\lambda_{\ell pt}^+\), define
\begin{align}
  s_{jp}(\theta;\lambda)
  =
  \sum_{\ell=0}^1\sum_{t=1}^T
  &\lambda_{\ell pt}^-
  \left\{
    \one(z_{j\ell t}(\theta)\leq v_p)-q_\ell
  \right\}
  \notag\\
  &+
  \lambda_{\ell pt}^+
  \left\{
    q_\ell-\one(z_{j\ell t}(\theta)<v_p)
  \right\},
  \label{eq:hlfinitedualscore}\\
  \Psi_{x,\theta}^{\mathrm{fin}}(\lambda)
  &=
  \sum_{j=1}^Jp_j
  \max_{1\leq p\leq P}s_{jp}(\theta;\lambda).
  \label{eq:hlfinitedualfunctional}
\end{align}

\begin{corollary}[Exact finite-support observable dual]
\label{cor:hlfinitedual}
The primal LP
\eqref{eq:hlprimalnonnegative}--\eqref{eq:hlprimalstrict} is feasible if and
only if
\begin{equation}
  \Psi_{x,\theta}^{\mathrm{fin}}(\lambda)\geq0
  \quad\text{for every }\lambda\geq0.
  \label{eq:hlfinitedualcriterion}
\end{equation}
If the primal is infeasible, a nonnegative \(\lambda\) makes
\(\Psi_{x,\theta}^{\mathrm{fin}}(\lambda)<0\).
\end{corollary}

\begin{proof}
Stack the masses \(q_{jp}\) into \(q\), write the marginal equalities as
\(Bq=p\), and write the quantile inequalities as \(Gq\geq0\).  The coefficient
of \(q_{jp}\) in a nonnegative linear combination of the rows of \(G\) is
\(s_{jp}(\theta;\lambda)\).  Farkas' alternative says that the primal is
infeasible if and only if there are \(u\in\mathbb R^J\) and \(\lambda\geq0\)
such that
\[
  u_j\geq s_{jp}(\theta;\lambda)
  \quad\text{for every }j,p,
  \qquad
  \sum_{j=1}^Jp_ju_j<0.
\]
Minimizing the left side over admissible \(u_j\) sets
\(u_j=\max_ps_{jp}(\theta;\lambda)\), which gives exactly a strict violation of
\eqref{eq:hlfinitedualcriterion}.  The converse follows by taking these epigraph
values as \(u_j\).
\end{proof}

A normalized certificate can be obtained from the LP
\begin{equation}
  \delta_x(\theta)
  =
  \min_{\lambda,\,\xi}
  \sum_{j=1}^Jp_j\xi_j
  \label{eq:hlcertificateobjective}
\end{equation}
subject to
\begin{align}
  \xi_j&\geq s_{jp}(\theta;\lambda),
  &&j\leq J,\quad p\leq P,
  \label{eq:hlcertificateepigraph}\\
  \lambda_{\ell pt}^-,\lambda_{\ell pt}^+&\geq0,
  &&\ell=0,1,\quad p\leq P,\quad t\leq T,
  \label{eq:hlcertificatenonnegative}\\
  \sum_{\ell=0}^1\sum_{p=1}^P\sum_{t=1}^T
  \left(
    \lambda_{\ell pt}^-+\lambda_{\ell pt}^+
  \right)&\leq1.
  \label{eq:hlcertificatenormalization}
\end{align}
The zero multiplier is feasible, and Corollary \ref{cor:hlfinitedual} implies
\begin{equation}
  \delta_x(\theta)=0
  \quad\Longleftrightarrow\quad
  \theta\text{ is feasible in path }x,
  \qquad
  \delta_x(\theta)<0
  \quad\Longleftrightarrow\quad
  \theta\text{ is infeasible in path }x.
  \label{eq:hlcertificateinterpretation}
\end{equation}

\begin{algorithm}[Exact two-quantile membership oracle]
\label{alg:hlmembership}
For a candidate
\(\theta=(b_0,b_1,r)\) with \(r>0\):
\begin{enumerate}[label=\arabic*.]
\item Form \(\gamma=1/r\), \(d=\gamma b_1\), and compute the two residual
      families in \eqref{eq:hlfiniteresiduals}.
\item In each positive-probability regressor path, sort their distinct values to
      obtain the exact anchor list \(v_1,\ldots,v_P\).
\item Solve either the primal feasibility problem
      \eqref{eq:hlprimalnonnegative}--\eqref{eq:hlprimalstrict} or the normalized
      dual-certificate problem
      \eqref{eq:hlcertificateobjective}--\eqref{eq:hlcertificatenormalization}.
\item Accept the joint candidate if and only if every path is primal feasible,
      equivalently if and only if every path has \(\delta_x(\theta)=0\).
\item To construct \(\Theta_{I,\tau}^{\mathrm{pt}}\), retain \((b_1,r)\) whenever
      at least one \(b_0\in\mathcal B_0\) is accepted.
\end{enumerate}
\end{algorithm}

\paragraph{Size and full-set construction.}
In one path, \(P\leq2TJ\).  The primal uses at most
\(JP\leq2TJ^2\) nonnegative variables, \(J\) marginal equalities, and
\(4TP\leq8T^2J\) quantile inequalities.  The transformed parameterization is
also useful for constructing the whole set in low dimension.  Every indicator
array can change only on a residual-tie hyperplane of the form
\begin{equation}
  \gamma_\ell y_t^j-x_t'd_\ell
  =
  \gamma_m y_s^h-x_s'd_m,
  \qquad
  \ell,m\in\{0,1\},
  \label{eq:hlhyperplanes}
\end{equation}
where
\[
  (\gamma_0,d_0)=(1,b_0),
  \qquad
  (\gamma_1,d_1)=(\gamma,d).
\]
Equation \eqref{eq:hlhyperplanes} is affine in
\((b_0,d,\gamma)\).  Consequently, for finite support and low-dimensional
parameters, one may enumerate all relatively open cells and faces of this
hyperplane arrangement, solve one LP per face, take the union of accepted faces,
and finally recover
\[
  r=\frac1\gamma,
  \qquad
  b_1=\frac d\gamma.
\]
This avoids an exogenous grid over either \(A\) or \(\rho(\tau)\).

\subsection{Scope of the two-quantile result}
\label{subsec:hlscope}

\begin{remark}[Two-quantile versus functional sharpness]
The projected set \(\Theta_{I,\tau}^{\mathrm{pt}}\) is sharp under the restrictions imposed
at \(1/2\) and at the specified \(\tau\).  If the model is maintained
simultaneously for every \(q\) in a larger set \(\mathcal Q\), all those quantile
restrictions must use the same coupling with \(A\).  For a finite grid
\(\mathcal Q\), the theorem and the LP extend exactly by adding one transformed
residual family and \(2T\) multiplier arrays for each \(q\).  The projection of
that joint set onto \((\beta(\tau),\rho(\tau))\) can be strictly smaller than the
two-quantile set.  A continuum of quantile indices leads to an
infinite-dimensional multiplier problem and requires additional measurability and
compactness conditions; it is not established merely by solving the problem on a
finite quantile grid.
\end{remark}

\begin{remark}[Negative or zero loadings]
The common-anchor transformation
\eqref{eq:hltauresidual} uses \(r>0\).  A negative loading reverses inequalities
when the outcome is rescaled and therefore is not covered by the stacked-residual
formulation above; a zero loading removes \(A\) from the \(\tau\)-quantile
altogether.  Both cases can be analyzed directly with the threshold indicators
\[
  \one\{Y_t\leq X_t'b_1+rA\},
  \qquad
  \one\{Y_t<X_t'b_1+rA\},
\]
but they require separate parameter regions and anchor bounds.  Thus positivity
should be reported as a maintained relative-loading restriction, not described as
without loss of generality.
\end{remark}

\subsection{Additional finite-support outer screens}
\label{subsec:hlscreens}

\begin{corollary}[Combined support overlap]
\label{cor:hlsupportoverlap}
For a candidate \(\theta=(b_0,b_1,r)\), let
\[
  \ell_{\ell t}(x,\theta)
  =
  \operatorname*{ess\,inf}
  \{Z_{\ell t}(\theta)\mid X=x\},
  \qquad
  u_{\ell t}(x,\theta)
  =
  \operatorname*{ess\,sup}
  \{Z_{\ell t}(\theta)\mid X=x\}.
\]
If \(\theta\in\Theta_I^{0,\tau}\), then, for almost every \(x\),
\[
  \bigcap_{\ell=0}^1\bigcap_{t=1}^T
  [\ell_{\ell t}(x,\theta),u_{\ell t}(x,\theta)]
  \neq\varnothing.
\]
Under finite conditional support this is an extrema-only outer screen.
\end{corollary}

\begin{proof}
In a feasible joint coupling, the common anchor must lie in the conditional
support interval of every transformed residual.  Below an interval its weak
lower-quantile probability is zero; above an interval its strict
upper-quantile probability is one.  The proof is the stacked-coordinate version
of Lemma \ref{lem:overlap}.
\end{proof}

\begin{corollary}[Pairwise stacked-coordinate outer set]
\label{cor:hlpairwise}
For each pair of transformed coordinates \((\ell,s)\) and \((m,t)\), solve the
exact two-coordinate coupling LP while conditioning on the full regressor path.
The intersection of all accepted pairwise sets contains
\(\Theta_I^{0,\tau}\).  It is generally larger because the pair-specific anchors
need not be compatible with one common anchor for all \(2T\) coordinates.
\end{corollary}

Together with Proposition \ref{prop:hlcrossing}, these results give the
loading-model cascade
\[
  \begin{gathered}
  \text{combined support overlap}
  \ \longrightarrow\
  \text{within- and cross-quantile crossing}\\
  {}\longrightarrow
  \text{pairwise stacked LPs}
  \ \longrightarrow\
  \text{full joint LP}.
  \end{gathered}
\]
For continuous designs, the dual-sieve construction in
Sections~\ref{subsec:continuous}--\ref{subsec:sieveevaluation} applies after
replacing the \(T\) residual coordinates and common quantile constant by the
\(2T\) coordinates \(Z_{\ell t}\) and their constants \(q_\ell\).  The
interval-envelope alternative noted in Appendix~\ref{app:envelope} can likewise
be applied to these transformed coordinates.  Envelope programs are outer;
representative-point and empirical calculations retain the numerical and
sampling qualifications described there.

\section{Monte Carlo experiments}
\label{sec:montecarlo}

The Monte Carlo analysis will have two purposes.  The first is economic: to
measure how the sharp identified set changes with the time dimension, the
support and distribution of the regressors, and the width of the composite
residual distribution.  The second is computational: to compare the support,
crossing, pairwise, primal, and dual procedures developed above.  The analysis
is deliberately separated from statistical inference.  All reported sets in
the baseline finite-support designs will be population sharp sets computed from
the known design probabilities; sample coverage will be studied in subsequent
work.

\subsection{Baseline designs}

The baseline data-generating process is
\begin{equation}
  Y_{it}=X_{it}'\beta_0+A_i+U_{it},
  \qquad
  U_{it}=\sigma_t(X_i,A_i)(V_{it}-\tau),
  \label{eq:mcdgp}
\end{equation}
where every scale \(\sigma_t\) is positive and bounded.  Each \(V_{it}\) is
marginally uniform on \([0,1]\) conditional on \((X_i,A_i)\), while the vector
\((V_{i1},\ldots,V_{iT})\) is generated from a Gaussian copula.  Varying the
copula correlation permits strong positive or negative serial dependence without
altering the conditional quantile restriction.  The individual effect is allowed
to depend on the complete regressor path through
\[
  A_i=m_A(X_i)+\eta_i,
  \qquad \eta_i\mid X_i\sim\operatorname{Unif}[-M_A,M_A].
\]
Thus the designs retain arbitrary correlation between \(A_i\) and \(X_i\),
while the bounded innovation and bounded scale generate a transparent
finite-width composite-residual envelope.

We will consider four sequences of designs.  First, finite-support regressor
paths will permit exact construction of the population identified set and a
direct comparison of the full-panel set with the three outer screens.  Second,
bounded continuous designs will vary the radius and directional shape of the
regressor support.  Third, jointly Gaussian designs will vary the smallest
eigenvalue of \(\operatorname{Var}(X_{i2}-X_{i1})\) and the envelope width
\(C_T\), thereby illustrating Example \ref{ex:Tgaussianpoint} and the distinction
between population point identification and identification driven by rare tail
paths.  Fourth, for \(T\in\{2,3,5,10\}\), we will vary the rate at which new
periods expand the difference body \(\mathcal K_T\).  These designs separate
the mechanical nesting of the identified sets from genuine contraction induced
by additional directional support.

For every design we will report componentwise widths, diameter, volume when the
dimension permits, and distance from the sharp set to \(\beta_0\).  We will
also report the anisotropic outer region \(\beta_0+C_T\mathcal K_T^\circ\)
and the high-quantile radial bounds over several values of \(\eta\).  Comparing
these objects will show how much information is lost by reducing the full
directional geometry to a single scalar modulus.

\subsection{Computation and the factor-loading extension}

In the finite-support designs, the complete hyperplane-face construction will be
used for one- and two-dimensional slopes.  For each candidate face we will record
which stage first rejects it, the number and size of the pathwise programs, and
the time required by the primal and dual formulations.  The dual multipliers will
be retained to assess whether simple crossing restrictions account for most
rejections or whether higher-order common-anchor compatibility is empirically
important.  For continuous designs, nested conservative partitions and nested
dual multiplier sieves will be compared; representative-point grids will be
reported only as numerical approximations.

For the loading model, conditional outcome distributions will be generated so
that the median and \(\tau\)-quantile selections are
\[
  v_{0t}=X_{it}'b_0+A_i,
  \qquad
  v_{1t}=X_{it}'b_1+r_0A_i,
\]
with the coefficient and support choices imposing the required order between
\(v_{0t}\) and \(v_{1t}\).  When \(\tau>1/2\), for example, conditional mass
\(1/2\) is placed at \(v_{0t}\), mass \(\tau-1/2\) at \(v_{1t}\), and the
remaining mass strictly above \(v_{1t}\); the construction is reversed when
\(\tau<1/2\).  Serial dependence is introduced by coupling these conditional
draws across periods while preserving their marginal probabilities.  We will
compare the projected sharp set for \((b_1,r_0)\) with the set obtained from
within-quantile restrictions alone.  The difference isolates the identifying
content of cross-quantile common-anchor restrictions for the relative loading.

\section{Conclusion}
\label{sec:conclusion}

This paper develops a sharp identification analysis for linear panel quantile
models with unrestricted individual heterogeneity and a fixed number of periods.
The maintained full-path restriction is quantile strict exogeneity, not
independence.  It does not justify differencing and does not make
the idiosyncratic errors independent over time.  Its identifying content instead
comes from the requirement that all observed residual coordinates be compatible
with one latent scalar anchor.  The coupling representation makes this content
explicit, the observable dual characterizes it sharply, the finite-support
linear programs make the sharp set constructive, and the nested dual sieve
extends the construction to continuous designs.  The dual values also supply
scalar objective functions: under compatibility, their sets of global optimizers
are the sharp set in finite support and nested outer sets for continuous
support.  This allows joint coefficient optimization and profiling, rather
than requiring a prespecified grid of separate membership tests.

The analysis also clarifies when a short panel is informative.  Time-invariant
regressor directions remain absorbed by the individual effect, and adding periods
alone need not shrink the set.  What matters is directional index variation
relative to the conditional support of the composite residual.  This distinction
both yields finite-\(T\) point identification under the conditions of the
Gaussian example and produces explicit contraction bounds as the number of
periods grows.  The baseline analysis is pointwise in the quantile index and
therefore permits a separate latent completion at each quantile.  The
factor-loading extension deliberately strengthens that model by imposing one
common anchor across quantiles; this joint restriction can identify economically
meaningful relative loadings.  Statistical inference for the resulting sets is
left for future work.

\clearpage
\appendix
\section{Proof of the observable duality theorem}
\label{app:dualityproof}

\begin{proof}[Proof of Theorem~\ref{thm:observabledual}]
The proof has four steps.

\textit{Step 1: conditional quantile restrictions as positive measures.}
Let
\[
  \Gamma_b
  =
  \{\pi\in\mathcal P(\mathcal W\times K_b):
    \pi_W=\Pobs\}.
\]
This is the collection of all couplings of the observed variables with an anchor
supported on \(K_b\).  It is convex and weakly compact.

Fix \(\pi\in\Gamma_b\), and let \(\nu_\pi\) be its \((X,A)\)-marginal.  For every
\(t\), define finite signed Borel measures on \(\mathcal X\times K_b\) by
\begin{align}
  M_{t,\pi}^-(B)
  &=
  \int
  \one\{(x,a)\in B\}
  \{\one(y_t-x_t'b\leq a)-\tau\}
  \,d\pi(y,x,a),
  \label{eq:signedmeasureweak}\\
  M_{t,\pi}^+(B)
  &=
  \int
  \one\{(x,a)\in B\}
  \{\tau-\one(y_t-x_t'b<a)\}
  \,d\pi(y,x,a).
  \label{eq:signedmeasurestrict}
\end{align}
Their Radon--Nikodym derivatives with respect to \(\nu_\pi\) are
\begin{align*}
  \frac{dM_{t,\pi}^-}{d\nu_\pi}(X,A)
  &=
  \pi\{R_t(b)\leq A\mid X,A\}-\tau,\\
  \frac{dM_{t,\pi}^+}{d\nu_\pi}(X,A)
  &=
  \tau-\pi\{R_t(b)<A\mid X,A\}.
\end{align*}
Consequently, \(\pi\) satisfies
\eqref{eq:coupleweak}--\eqref{eq:couplestrict} if and only if all \(2T\) signed
measures in \eqref{eq:signedmeasureweak}--\eqref{eq:signedmeasurestrict} are
nonnegative.  Since \(\mathcal X\times K_b\) is a compact metric space, a finite
signed Borel measure \(M\) is nonnegative if and only if
\(\int f\,dM\geq0\) for every \(f\in C_+(\mathcal X\times K_b)\).
It follows that
\begin{equation}
  \pi\text{ is feasible}
  \quad\Longleftrightarrow\quad
  L_b(\pi,\lambda):=
  \int S_b(w,a;\lambda)\,d\pi(w,a)\geq0
  \quad\text{for every }\lambda\geq0.
  \label{eq:lagrangianfeasible}
\end{equation}

\textit{Step 2: a normalized feasibility value.}
Let
\begin{equation}
  \Lambda_b^1
  =
  \left\{\lambda\in C_+(\mathcal X\times K_b)^{2T}:
  \sum_{t=1}^T
  \{\|\lambda_t^-\|_\infty+\|\lambda_t^+\|_\infty\}
  \leq1\right\}
  \label{eq:normalizedmultipliers}
\end{equation}
and define
\begin{equation}
  V_b
  =
  \sup_{\pi\in\Gamma_b}
  \inf_{\lambda\in\Lambda_b^1}L_b(\pi,\lambda).
  \label{eq:feasibilityvalue}
\end{equation}
Because \(0\in\Lambda_b^1\), the inner infimum is never positive.  By
\eqref{eq:lagrangianfeasible}, it equals zero if \(\pi\) is feasible.  If \(\pi\)
is infeasible, at least one of the signed measures in
\eqref{eq:signedmeasureweak}--\eqref{eq:signedmeasurestrict} is not nonnegative.
The continuous-function characterization of positive measures then supplies a
nonnegative continuous multiplier with a strictly negative integral; after
normalization, the inner infimum is strictly negative.

The function
\[
  \pi\longmapsto
  \inf_{\lambda\in\Lambda_b^1}L_b(\pi,\lambda)
\]
is upper semicontinuous: for fixed \(\lambda\), the integrand
\eqref{eq:dualscore} is bounded and upper semicontinuous, and an infimum of
upper-semicontinuous functions is upper semicontinuous.  Since \(\Gamma_b\) is
compact, the supremum in \eqref{eq:feasibilityvalue} is attained.  Therefore
\begin{equation}
  V_b=0
  \quad\Longleftrightarrow\quad
  \Gamma_b\text{ contains a feasible coupling},
  \qquad
  V_b<0\text{ otherwise}.
  \label{eq:valuefeasibility}
\end{equation}

\textit{Step 3: minimax.}
For fixed \(\lambda\geq0\), the score in \eqref{eq:dualscore} is upper
semicontinuous in \((w,a)\).  Indeed,
\(\{(w,a):y_t-x_t'b\leq a\}\) is closed, so its indicator is upper
semicontinuous, while \(\{(w,a):y_t-x_t'b<a\}\) is open, so the negative of its
indicator is upper semicontinuous.  Nonnegative continuous multiplication
preserves these properties.  Hence
\(\pi\mapsto L_b(\pi,\lambda)\) is affine and upper semicontinuous.  For fixed
\(\pi\), \(L_b(\pi,\lambda)\) is continuous and linear in \(\lambda\) under the
product sup norm.  The set \(\Gamma_b\) is compact and convex, and
\(\Lambda_b^1\) is convex.  Sion's minimax theorem \citep{sion1958}
therefore gives
\begin{equation}
  V_b
  =
  \inf_{\lambda\in\Lambda_b^1}
  \sup_{\pi\in\Gamma_b}L_b(\pi,\lambda).
  \label{eq:minimax}
\end{equation}

\textit{Step 4: eliminate the coupling.}
For any \(\pi\in\Gamma_b\),
\[
  L_b(\pi,\lambda)
  \leq
  \E_{\mathrm{obs}}
  \left[
    \max_{a\in K_b}S_b(W,a;\lambda)
  \right].
\]
Conversely, the score is Borel measurable and upper semicontinuous in \(a\) on
the compact set \(K_b\).  The measurable maximum theorem supplies a measurable
selector \(a_\lambda(w)\) from its argmax correspondence.  Coupling
\(A=a_\lambda(W)\) with \(W\sim\Pobs\) attains the preceding upper bound.  Thus
\begin{equation}
  \sup_{\pi\in\Gamma_b}L_b(\pi,\lambda)
  =
  \Psi_b(\lambda).
  \label{eq:pointwiseanchor}
\end{equation}
Combining \eqref{eq:valuefeasibility}, \eqref{eq:minimax}, and
\eqref{eq:pointwiseanchor} yields
\[
  b\in\ThetaI
  \quad\Longleftrightarrow\quad
  \inf_{\lambda\in\Lambda_b^1}\Psi_b(\lambda)=0.
\]
Since \(\Psi_b(0)=0\), the last equality holds if and only if
\(\Psi_b(\lambda)\geq0\) for every \(\lambda\in\Lambda_b^1\).  Positive
homogeneity extends the condition to the full cone
\(C_+(\mathcal X\times K_b)^{2T}\), proving
\eqref{eq:sharpdualcriterion}.  If \(b\notin\ThetaI\), then \(V_b<0\), so
\eqref{eq:minimax} supplies a multiplier satisfying
\eqref{eq:strictdualcertificate}.
\end{proof}

\section{Computational details and proofs}
\label{app:computation}

\subsection{The residual-value reduction and finite-support program}
\label{app:finiteproof}

\begin{proof}[Proof of Theorem~\ref{thm:lp}]
\textit{LP feasibility implies a coupling.}
Use \(q_{jp}\) as the joint probability of residual support point
\(z^j(b)\) and anchor \(v_p\).  Equation \eqref{eq:lpmarginal} gives the observed
residual marginal.  If the anchor mass
\(w_p=\sum_jq_{jp}\) is positive, divide
\eqref{eq:lpweak}--\eqref{eq:lpstrict} by \(w_p\).  The resulting inequalities say
exactly that \(v_p\) is a coordinatewise conditional \(\tau\)-quantile of the
residual vector given \(A=v_p\).  If \(w_p=0\), the corresponding restrictions are
vacuous.  Thus the \(q_{jp}\)'s construct the required coupling.

\textit{A coupling implies LP feasibility.}
Start from any feasible coupling of the finite residual vector with an arbitrary
real-valued anchor.  It cannot put anchor mass below \(v_1\): conditional on such an
anchor, \(P(R_t(b)\leq A\mid A)=0<\tau\) for every \(t\).
It cannot put anchor mass above \(v_P\): conditional on such an anchor,
\(P(R_t(b)<A\mid A)=1>\tau\) for every \(t\).

Keep anchor mass already at a residual value.  If
\(a\in(v_p,v_{p+1})\), move its mass to the lower endpoint \(v_p\), retaining its
joint mass with every residual support point.  For every residual coordinate \(z\),
\[
  \one(z\leq a)=\one(z\leq v_p),
  \qquad
  \one(z<v_p)\leq\one(z<a).
\]
The lower-quantile inequality is therefore unchanged, while the strict
upper-quantile probability can only fall, so that inequality is weakly relaxed.
Apply this move to every open gap and aggregate at equal residual values.  The
observed residual marginal is unchanged, and the resulting masses \(q_{jp}\) satisfy
\eqref{eq:lpnonnegative}--\eqref{eq:lpstrict}.

The cellwise statement \eqref{eq:lpsharpset} follows by disintegrating a global
coupling across the finitely many \(X\)-paths and, conversely, combining the
pathwise couplings with their observed path probabilities.
\end{proof}

\subsection{The finite-support dual and the criterion proof}
\label{subsec:finitedual}

Use the score \(s_{jp}(b;\lambda)\) defined in
\eqref{eq:finitedualscore}.  Its observable average after maximizing over the
candidate-effect columns is
\begin{equation}
  \Psi_{x,b}^{\mathrm{fin}}(\lambda)
  =\sum_{j=1}^Jp_j\max_{1\leq p\leq P}s_{jp}(b;\lambda).
  \label{eq:finitedualfunctional}
\end{equation}
The following finite-dimensional alternative proves the criterion
characterization in Proposition~\ref{prop:dualcertificate}.

\begin{corollary}[Exact finite-support observable dual]
\label{cor:finitedual}
For the path \(X=x\), the exact coupling LP
\eqref{eq:lpnonnegative}--\eqref{eq:lpstrict} is feasible if and only if
\begin{equation}
  \Psi_{x,b}^{\mathrm{fin}}(\lambda)\geq0
  \quad\text{for every }
  \lambda_{pt}^-,\lambda_{pt}^+\geq0.
  \label{eq:finitedualcriterion}
\end{equation}
If the primal coupling LP is infeasible, some nonnegative \(\lambda\) makes the
left-hand side of \eqref{eq:finitedualcriterion} strictly negative.
\end{corollary}

\begin{proof}
Stack the primal masses \(q_{jp}\) into a vector \(q\).  Write the marginal
equalities as \(Aq=p\).  Bring the quantile restrictions to the left-hand
side, orienting each as a nonnegative expression, and write them as
\(Gq\geq0\).
The coefficient of \(q_{jp}\) in the linear combination of the rows of \(G\) with
nonnegative weights \(\lambda\) is exactly \(s_{jp}(b;\lambda)\).

The alternative theorem for linear inequalities states that
\[
  q\geq0,\qquad Aq=p,\qquad Gq\geq0
\]
is infeasible if and only if there are \(u\in\R^J\) and \(\lambda\geq0\) such
that
\[
  A'u-G'\lambda\geq0,
  \qquad
  p'u<0.
\]
The first inequality is equivalent to
\[
  u_j\geq s_{jp}(b;\lambda)
  \qquad\text{for every }j,p.
\]
It follows that
\[
  \Psi_{x,b}^{\mathrm{fin}}(\lambda)
  =
  \sum_jp_j\max_ps_{jp}(b;\lambda)
  \leq
  \sum_jp_ju_j<0.
\]
Conversely, if \(\Psi_{x,b}^{\mathrm{fin}}(\lambda)<0\), set
\(u_j=\max_ps_{jp}(b;\lambda)\).  Then \((u,\lambda)\) satisfies the displayed
alternative system and certifies primal infeasibility.
\end{proof}

\begin{proof}[Proof of Proposition~\ref{prop:dualcertificate}]
The zero multiplier with \(d_j=0\) is feasible, so \(\delta_x(b)\leq0\).
If the primal is feasible, Corollary \ref{cor:finitedual} implies that every
feasible \(\lambda\) has
\(\sum_jp_j\max_ps_{jp}(b;\lambda)\geq0\).  Minimization over the epigraph
variables \(d_j\) therefore gives \(\delta_x(b)=0\).  If the primal is
infeasible, Corollary \ref{cor:finitedual} supplies a multiplier with a strictly
negative dual functional.  Positive homogeneity permits normalization to satisfy
\eqref{eq:dualcertificatenormalize}, giving \(\delta_x(b)<0\).
\end{proof}

\subsection{Program size, certificates, and numerical checks}
\label{subsec:primaldualcomputation}

In a path with \(J\) outcome vectors and \(P\leq JT\) distinct residual
values, the primal has \(JP\) nonnegative unknowns, \(J\) marginal
equalities, and \(2TP\) quantile inequalities.  The dual in
Proposition~\ref{prop:dualcertificate} has \(J+2TP\) unknowns and \(JP+1\)
principal constraints, in addition to the sign restrictions.  The dual evaluates the criterion used in the outer coefficient optimization;
the primal is useful when a rationalizing distribution is wanted.  A negative
dual value also supplies an explicit excluding inequality.  Their relative running times depend on the problem and
the solver, not only on the number of unknowns.

\paragraph{Checking a numerical answer.}
Choose and report a feasibility tolerance \(\epsilon_{\mathrm{LP}}>0\).
For a returned primal table, check nonnegativity, row totals, and both quantile
restrictions at that tolerance.  For a returned dual multiplier, check its
nonnegativity and normalization, and evaluate
\(\sum_jp_j\max_p s_{jp}(b;\lambda)\) using every candidate anchor.
A value below \(-\epsilon_{\mathrm{LP}}\) is a numerical exclusion with that
reported margin.  An exactly established zero optimum is equivalent to feasibility by
Proposition~\ref{prop:dualcertificate}; a rounded numerical value near zero
is not the same conclusion.  Check the solver's objective bounds or optimality
gap.  When a solver reports an ambiguous status, the other formulation can be
used as a check rather than as a required second solve at every coefficient.  Conflicting statuses call for rescaling, tighter tolerances, and
another solve, not classification from a ``nearly feasible'' flag.

Paths can be solved independently.  Solutions can be reused when nearby
coefficients induce the same residual orderings; in that case the programs
are identical.  The purpose of storing primal tables or dual multipliers is
to retain the evidence for the reported classification.

\paragraph{Adding dual constraints as needed.}
For a large dual, it is not necessary to load all \(JP\) constraints
\(d_j\geq s_{jp}(b;\lambda)\) initially.  Start with a nonempty subset of
anchor constraints for each \(j\), solve that program, and compute
\[
  p^\ast(j)\in\arg\max_{1\leq p\leq P}s_{jp}(b;\lambda).
\]
Add \(d_j\geq s_{j p^\ast(j)}(b;\lambda)\) whenever it is violated by more
than the chosen tolerance.  Repeat until no constraint is violated.  Because
the list is finite, this procedure terminates after finitely many additions,
up to solver tolerance.  A negative objective from an intermediate program
is not yet a certificate: all omitted anchor constraints must first be
checked.

\subsection{Constructing the set from residual orderings}
\label{app:orderings}

\begin{proof}[Proof of Proposition~\ref{prop:arrangement}]
Every LP coefficient involving \(b\) is an indicator of either
\(z_{mjt}(b)\leq z_{mhs}(b)\) or \(z_{mjt}(b)<z_{mhs}(b)\).  The weak and strict
ordering of every residual pair is fixed on a relatively open face of the
arrangement, including a face on which one or more ties hold.  Hence the distinct
anchor classes and all LP coefficients are constant on that face.  The
normalization in \eqref{eq:dualcertificatenormalize}, the observed
probabilities, and the path weights also stay fixed.  Hence each
\(\delta_{x^m}(b)\), and thus \(\delta(b)\) and \(\mathcal L(b)\), is
constant on the face.  Theorem~\ref{thm:lp} gives the corresponding primal
statement, and \eqref{eq:globaldualzeroset} gives \eqref{eq:faceunion}.
There are finitely many weak and strict residual orderings even without a
polyhedral restriction on \(\cB\).  Thus \(\mathcal L\) takes finitely
many values on any coefficient region.  This also justifies attainment in
nonempty slices and the profile characterization in
\eqref{eq:profiledualzeroset}.
\end{proof}

\begin{algorithm}[Full sharp set by cells and faces]
\label{alg:fullset}
\begin{enumerate}[label=\arabic*.]
\item Generate and deduplicate the hyperplanes \eqref{eq:tiehyperplane}, restricted
      to \(\cB\).
\item Enumerate all nonempty relatively open cells and lower-dimensional faces.
\item Choose one relative-interior point \(b_F\) from each face \(F\), evaluate
      the pathwise dual LPs, and form \(\delta(b_F)\).
\item Return the union of all faces with \(\delta(b_F)=0\), retaining their
      weak/strict boundary descriptions.  Under compatibility, this is the
      complete set of global maximizers of \(\delta\).
\end{enumerate}
Testing only full-dimensional cells is not exact: lower-dimensional tie faces can
have a different LP and must be checked.  
\end{algorithm}

The number of distinct nonvacuous hyperplanes is no larger than
\(\sum_m\binom{TJ_m}{2}\).  For a fixed coefficient dimension \(k\), an
arrangement of \(N\) hyperplanes has \(O(N^k)\) regions and faces, before
accounting for the boundary description of \(\cB\).  This explains why the
construction is most useful for low-dimensional slopes.  Preliminary outer
checks can discard a face before the full-panel program is solved.

\subsection{Pairwise and crossing outer screens}
\label{subsec:outer}

The following relaxations trade tightness for speed.  Each is necessary for full-panel
feasibility, so their intersection is also an outer set.

\paragraph{Pairwise sharp outer set.}
For every pair \(s<t\), retain the conditional law of \((Y_s,Y_t)\) given the
\emph{full} path \(X\), build the exact two-coordinate version of Algorithm
\ref{alg:membership}, and intersect the accepted sets:
\[
  \ThetaI\subseteq\Theta_I^{\mathrm{pair}}
  :=\bigcap_{s<t}\Theta_{I,st}.
\]
This follows by marginalizing any feasible full-panel coupling to the retained pair.
It is sharp for each retained pair but
generally not for the full panel because the pair-specific anchors need not be
compatible.  With finite support, each pairwise path LP has at most \(2J_m^2\)
variables and \(8J_m\) quantile inequalities.

For a candidate \(b\), reject as soon as a pairwise path program is infeasible.  Only surviving candidates need a full-panel calculation.

\paragraph{Crossing outer set.}
Proposition \ref{prop:crossing} gives
\begin{equation}
  \Theta_I^{\mathrm{cross}}
  =
  \bigcap_m\bigcap_{s\ne t}
  \left\{b:
  \sup_v P^{\mathrm{obs}}\{R_s(b)\leq v<R_t(b)\mid X=x^m\}
  \leq c_\tau\right\}.
  \label{eq:crossouterset}
\end{equation}
Then \(\ThetaI\subseteq\Theta_I^{\mathrm{cross}}\).  Under finite conditional
support, it is enough, for a fixed \(b\), to check thresholds among the distinct
residual values in that path.  No optimization over latent masses is required.
With continuous \(X\), the conditional moment inequalities
\eqref{eq:observablecmi} and a chosen instrument class give an observable sample
analogue, again an outer relaxation at the population level.

For each path, ordered pair, and distinct residual threshold, sum the observed probabilities of the crossing event.  Reject the candidate if any sum exceeds \(c_\tau\).

\begin{corollary}[Computable outer sets for model \textup{(A)}]
\label{cor:commonoutersets}
Let \(\Theta_I^{\mathrm{pair}}\), \(\Theta_I^{\mathrm{cross}}\), and
\(\Theta_I^{\mathrm{supp}}\) denote, respectively, the pairwise sharp,
crossing, and support-overlap sets defined above.  Then
\[
  \Theta_I
  \subseteq
  \Theta_I^{\mathrm{pair}}
  \cap
  \Theta_I^{\mathrm{cross}}
  \cap
  \Theta_I^{\mathrm{supp}}.
\]
None of the three containments is asserted to be an equality.  Passing all three
screens is therefore necessary but not sufficient for full-panel feasibility.
\end{corollary}

\begin{proof}
Support overlap follows from Lemma \ref{lem:overlap}, and the crossing
containment follows from Proposition \ref{prop:crossing}.  Marginalizing a
feasible \(T\)-coordinate common-anchor coupling to any pair of coordinates gives
feasibility of every pairwise LP.  Pairwise couplings constructed separately need
not share one joint anchor, so the converse need not hold.
\end{proof}
\subsection{Recovery by nested multiplier classes}
\label{app:sieveproof}

\begin{proof}[Proof of Proposition~\ref{prop:dualsieverecovery}]
Every feasible \(b\) satisfies all sieve inequalities by Theorem
\ref{thm:observabledual} and Remark \ref{rem:boreldual}; nesting gives the two
containments.  If \(b\notin\Theta_I\), Theorem
\ref{thm:observabledual} supplies a continuous nonnegative multiplier
\(\lambda\) with \(\Psi_b(\lambda)<0\).  Moreover,
\[
  |\Psi_b(\lambda)-\Psi_b(\widetilde\lambda)|
  \leq
  \sum_{t=1}^T
  \left(
    \|\lambda_t^- -\widetilde\lambda_t^-\|_\infty
    +
    \|\lambda_t^+ -\widetilde\lambda_t^+\|_\infty
  \right).
\]
A sufficiently accurate sieve approximation therefore also has a negative
dual value and excludes \(b\).
\end{proof}

For the construction in \eqref{eq:cellhatbasis}, the required nesting is
obtained by subdividing each regressor cell and retaining all old knots.
Uniform approximation follows if
\[
  \max_g\operatorname{diam}(D_g^{(Q)}\cap\mathcal X)\longrightarrow0,
  \qquad
  \max_\ell(\kappa_{\ell+1}-\kappa_\ell)\longrightarrow0.
\]
Indeed, evaluate a continuous nonnegative multiplier at a representative
point of each cell and at each knot, and interpolate in \(a\).  Uniform
continuity on the compact domain bounds the approximation error uniformly.
These representative points are used only to establish approximation of a
\emph{multiplier}; they do not replace the observed regressors in a residual.
Cell indicators are bounded Borel functions, so their inequalities are valid
by Remark~\ref{rem:boreldual}.  If \(K_b\) is a singleton, use a constant
basis function on that singleton instead of an anchor knot grid.

When \(X\) has finite support, apply the construction in each exact path
using its own compact residual interval \(K_{m,b}\).  The uniform
approximation then concerns \(a\) alone.  The set recovery in
Proposition~\ref{prop:dualsieverecovery} is pointwise in the candidate slope;
Hausdorff convergence requires additional uniformity and closedness conditions.

\subsection{Exact evaluation of the anchor maximum}
\label{app:breakpoints}

Fix a weighted point \(W_i\), a candidate \(b\), and weights of the form
\eqref{eq:cellhatbasis}.  Write \(S_i(a)=S_b(W_i,a;\lambda)\).  Between
successive members of \(\mathcal A_{i,Q}(b)\), all event indicators are
constant and every multiplier is affine in \(a\).  Hence \(S_i\) is affine
on each such open interval.  At a residual value \(c\), continuity of the
multiplier in \(a\) gives
\begin{align*}
  S_i(c)-S_i(c^-)
    &=\sum_{t:r_{it}(b)=c}\lambda_t^-(X_i,c)\geq0,\\
  S_i(c)-S_i(c^+)
    &=\sum_{t:r_{it}(b)=c}\lambda_t^+(X_i,c)\geq0.
\end{align*}
At a knot that is not a residual value, the score is continuous.  The endpoints
of \(K_b\) are included among the knots.  These observations prove
\[
  \max_{a\in K_b}S_b(W_i,a;\lambda)
  =\max_{a\in\mathcal A_{i,Q}(b)}S_b(W_i,a;\lambda).
\]
No evaluations at artificial points \(c\pm\varepsilon\) are needed.

Because the interpolation functions are nonnegative and sum to one, the
supremum norm of a multiplier is its largest knot value across nonempty
regressor cells.  Thus the auxiliary bounds \(z_t^\pm\) in
\eqref{eq:dualsievenormalization} implement the normalization in
\(\Lambda_{Q,b}^1\).  Empty cells can be omitted.  After duplicate breakpoints
are removed, there are at most \(N(L_Q+1+T)\) score constraints.  The program
has \(N+2TG_Q(L_Q+1)+2T\) unknowns, including the normalization bounds.

\subsection{Implementing and interpreting a sequence of sieve calculations}
\label{app:sieveimplementation}

The law used in the expectation should be held fixed while comparing sieve
resolutions.  With genuinely nested classes and the same anchor interval,
enlarging the class cannot raise the minimized dual value.  For the population
law this gives nested outer sets.  With a fixed weighted law it gives nesting
for that weighted-law calculation, not a sampling guarantee about the
population set.

\begin{algorithm}[Criterion-based dual-sieve computation]
\label{alg:sieve}
Fix a coefficient region, a law or weighted representation of that law, and
valid compact residual intervals for the coefficients considered.
\begin{enumerate}[label=\arabic*.]
\item At resolution \(Q\), choose exact regressor paths when \(X\) has finite
      support, or cells of the complete regressor-path space otherwise.
      Specify the rule for the anchor knots, including the interval endpoints.
\item For each objective evaluation requested by the coefficient optimizer,
      compute the actual residuals at \(b\), form
      \(\mathcal A_{i,Q}(b)\), and solve
      \eqref{eq:empiricaldualsieve}--\eqref{eq:dualsievenormalization}.
      Return the objective value and retain its optimizing weights and solver
      accuracy information.
\item Maximize this value over coefficients, or minimize its negative.  With
      the exact population law and a nonempty outer set, recover its entire
      zero-level set as the set of optimal solutions.  When only bounds on a
      target are needed, optimize that target subject to the zero-value
      restriction.  For a numerical or empirical law, report the computed
      optimizing set and objective value with the qualifications below.
\item Increase \(Q\) by subdividing cells when needed and inserting knots
      without discarding earlier ones.  Repeat the coefficient optimization.
      Reusing previous solutions can aid the search, but retaining just one
      optimizer at each resolution does not recover the sequence of sets.
\end{enumerate}
Report the expectation used, partitions, knots, support intervals, coefficient
region, optimization method, and numerical tolerances.  An arbitrary bounded
search region yields only the part of the identified set in that region;
to report the entire set, the region must be known to contain it.  Refining
the coefficient optimization and refining the multiplier class are different
operations.  Neither LP accuracy nor multiple local starts alone certifies
recovery of all global solutions.
\end{algorithm}

With empirical weights, the exact residual range for the empirical law is
\[
  [\min_{i,t}r_{it}(b),\max_{i,t}r_{it}(b)],
\]
or the analogous range within each path when \(X\) has finite support.
These ranges need not contain the population residual support.  In the
finite-\(X\) case, one can also apply Theorem~\ref{thm:lp} directly to the
empirical conditional outcome distribution.  That computes the sharp set of
the empirical law, not exact membership for the continuously distributed
population.  Data-dependent partitions, increasing resolution, estimated
support endpoints, and statistical inference require additional analysis.

\paragraph{An alternative outer approximation.}
\label{app:envelope}
One can instead partition the joint \((Y,X)\)-space and the anchor interval
and use within-cell residual extrema to construct a finite primal outer
relaxation.  Such a construction can provide an additional preliminary check,
but it involves the dimension of \((Y,X)\).  No claim is made here that its
limiting intersection is sharp.  The nested dual construction is the procedure
for which Proposition~\ref{prop:dualsieverecovery} establishes sharp-set recovery.
\bibliographystyle{apalike}
\bibliography{references}
\end{document}